\pdfoutput=1
\RequirePackage{ifpdf}
\ifpdf 
\documentclass[pdftex]{sigma}
\else
\documentclass{sigma}
\fi

\numberwithin{equation}{section}

\newtheorem{Theorem}{Theorem}[section]
\newtheorem{Corollary}[Theorem]{Corollary}
\newtheorem{Lemma}[Theorem]{Lemma}
\newtheorem{Proposition}[Theorem]{Proposition}
 { \theoremstyle{definition}
\newtheorem{Definition}[Theorem]{Definition}

\newtheorem{Example}[Theorem]{Example}
\newtheorem{Remark}[Theorem]{Remark} }

\font\tengoth=eufm10 at 10pt
\font\sevengoth=eufm7 at 6pt
\newfam\gothfam
\textfont\gothfam=\tengoth
\scriptfont\gothfam=\sevengoth

\newcommand{\fB}{{\mathfrak B}}

\newcommand{\fS}{{\mathfrak S}}

\newcommand{\g}{{\mathfrak g}}

\newcommand{\fb}{{\mathfrak b}}

\newcommand{\fg}{{\mathfrak g}}
\newcommand{\fh}{{\mathfrak h}}

\newcommand{\fk}{{\mathfrak k}}

\newcommand{\fq}{{\mathfrak q}}

\newcommand{\fs}{{\mathfrak s}}

\renewcommand\sp{\mathfrak {sp}}

\newcommand{\1}{\mathbf{1}}

\newcommand{\cA}{\mathcal{A}}
\newcommand{\cB}{\mathcal{B}}
\newcommand{\cC}{\mathcal{C}}
\newcommand{\cD}{\mathcal{D}}
\newcommand{\cE}{\mathcal{E}}
\newcommand{\cF}{\mathcal{F}}

\newcommand{\cH}{\mathcal{H}}
\newcommand{\cK}{\mathcal{K}}

\newcommand{\cM}{\mathcal{M}}
\newcommand{\cN}{\mathcal{N}}

\newcommand{\cR}{\mathcal{R}}
\newcommand{\cS}{\mathcal{S}}

\newcommand\bg{{\bf{g}}}
\newcommand\bs{{\bf{s}}}
\newcommand\bt{{\bf{t}}}

\newcommand{\dd}{{\tt d}}

\newcommand{\subeq}{\subseteq}

\newcommand{\into}{\hookrightarrow}
\newcommand{\eps}{\varepsilon}

\newcommand{\shalf}{{\textstyle{\frac{1}{2}}}}

\newcommand{\N}{{\mathbb N}}
\newcommand{\Z}{{\mathbb Z}}
\newcommand{\R}{{\mathbb R}}
\newcommand{\C}{{\mathbb C}}

\newcommand{\bP}{{\mathbb P}}

\newcommand{\T}{{\mathbb T}}

\newcommand{\bE}{{\mathbb E}}

\newcommand{\bS}{{\mathbb S}}

\renewcommand{\hat}{\widehat}

\renewcommand{\tilde}{\widetilde}

\newcommand{\Aff}{\mathop{{\rm Af\/f}}\nolimits}

\newcommand{\GL}{\mathop{{\rm GL}}\nolimits}
\newcommand{\SL}{\mathop{{\rm SL}}\nolimits}

\newcommand{\PSL}{\mathop{{\rm PSL}}\nolimits}
\newcommand{\SO}{\mathop{{\rm SO}}\nolimits}
\newcommand{\SU}{\mathop{{\rm SU}}\nolimits}
\newcommand{\OO}{\mathop{\rm O{}}\nolimits}

\newcommand{\U}{\mathop{\rm U{}}\nolimits}

\newcommand{\Sp}{\mathop{{\rm Sp}}\nolimits}
\newcommand{\Sym}{\mathop{{\rm Sym}}\nolimits}

\newcommand{\gl} {\mathop{{\mathfrak{gl} }}\nolimits}

\newcommand{\fsl} {\mathop{{\mathfrak{sl} }}\nolimits}

\newcommand{\Exp}{\mathop{{\rm Exp}}\nolimits}

\newcommand{\Hom}{\mathop{{\rm Hom}}\nolimits}

\newcommand{\Aut}{\mathop{{\rm Aut}}\nolimits}

\newcommand{\Diff}{\mathop{{\rm Dif\/f}}\nolimits}

\newcommand{\id}{\mathop{{\rm id}}\nolimits}

\newcommand{\rad}{\mathop{{\rm rad}}\nolimits}
\renewcommand{\dim}{\mathop{{\rm dim}}\nolimits}

\newcommand{\supp}{\mathop{{\rm supp}}\nolimits}

\newcommand{\conv}{\mathop{{\rm conv}}\nolimits}

\newcommand{\Spann}{\mathop{{\rm span}}\nolimits}
\newcommand{\ev}{\mathop{{\rm ev}}\nolimits}

\newcommand{\Sesq}{\mathop{{\rm Sesq}}\nolimits}

\newcommand{\vphi}{\varphi}
\renewcommand{\phi}{\varphi}

\newcommand{\Rarrow}{\Rightarrow}

\newcommand{\oline}{\overline}

\newcommand{\la}{\langle}
\newcommand{\ra}{\rangle}

\newcommand{\Mot}{{\rm Mot}}

\newcommand{\res}{\vert}

\newcommand{\Spec}{{\rm Spec}}

\newcommand{\ssssarr}{\hbox to 15pt{\rightarrowfill}}
\newcommand{\sssarr}{\hbox to 20pt{\rightarrowfill}}
\newcommand{\ssarr}{\hbox to 30pt{\rightarrowfill}}
\newcommand{\sarr}{\hbox to 40pt{\rightarrowfill}}
\newcommand{\arr}{\hbox to 60pt{\rightarrowfill}}
\newcommand{\larr}{\hbox to 60pt{\leftarrowfill}}
\newcommand{\Arr}{\hbox to 80pt{\rightarrowfill}}

\newcommand{\pmat}[1]{\begin{pmatrix} #1 \end{pmatrix}}

\newcommand{\du}[2]{\langle#1,#2 \rangle}
\newcommand{\due}{\du{\, \cdot\, }{\, \cdot\, }}
\newcommand{\E}{{\bE}}
\newcommand{\mfg}{\mathfrak{g}}
\newcommand{\mfh}{\mathfrak{h}}
\newcommand{\mfq}{\mathfrak{q}}

\newcommand{\hE}{\hat\cE}

\begin{document}

\allowdisplaybreaks

\newcommand{\arXivNumber}{1510.07445}

\renewcommand{\PaperNumber}{058}

\FirstPageHeading

\ShortArticleName{Ref\/lection Positive Stochastic Processes Indexed by Lie Groups}
\ArticleName{Ref\/lection Positive Stochastic Processes\\ Indexed by Lie Groups}

\Author{Palle E.T.~JORGENSEN~$^\dag$, Karl-Hermann NEEB~$^\ddag$ and Gestur \'OLAFSSON~$^\S$}
\AuthorNameForHeading{P.E.T.~Jorgensen, K.-H.~Neeb and G.~\'Olafsson}

\Address{$^\dag$~Department of Mathematics, The University of Iowa, Iowa City, IA 52242, USA}
\EmailD{\href{mailto:palle-jorgensen@uiowa.edu}{palle-jorgensen@uiowa.edu}}

\Address{$^\ddag$~Department Mathematik, FAU Erlangen-N\"urnberg,\\
\hphantom{$^\ddag$}~Cauerstrasse 11, 91058-Erlangen, Germany}
\EmailD{\href{mailto:neeb@math.fau.de}{neeb@math.fau.de}}

\Address{$^\S$~Department of mathematics, Louisiana State University, Baton Rouge, LA 70803, USA}
\EmailD{\href{mailto:olafsson@math.lsu.edu}{olafsson@math.lsu.edu}}

\ArticleDates{Received October 28, 2015, in f\/inal form June 09, 2016; Published online June 21, 2016}

\Abstract{Ref\/lection positivity originates from one of the Osterwalder--Schrader axioms for constructive quantum f\/ield theory. It serves as a bridge between euclidean and relativistic quantum f\/ield theory. In mathematics, more specif\/ically, in representation theory, it is related to the Cartan duality of symmetric Lie groups (Lie groups with an involution) and results in a transformation of a unitary representation of a symmetric Lie group to a unitary representation of its Cartan dual. In this article we continue our investigation of representation theoretic aspects of ref\/lection positivity by discussing ref\/lection positive Markov processes indexed by Lie groups, measures on path spaces, and invariant gaussian measures in spaces of distribution vectors. This provides new constructions of ref\/lection positive unitary representations.}


\Classification{22E45; 60G15; 81S40}

{\small \tableofcontents}

\section{Introduction}
Ref\/lection positivity is one of the cornerstones in \textit{constructive quantum field theory}. It was f\/irst formulated in the fundamental work of Osterwalder and Schrader~\cite{OS73,OS75}, see also~\cite{GJ81}. Mathematically, a~quantum physical system corresponds to a unitary representation of the corresponding symmetry group $G$. In euclidean quantum f\/ield theory this group is the euclidean motion group and in relativistic quantum f\/ield theory the symmetry group is the Poincar\'e group. Ref\/lection positivity enters the picture when it comes to passing from euclidean quantum f\/ields to relativistic ones. The time ref\/lection and the passage to imaginary time, sometimes called Wick rotation, changes the standard euclidean inner product into the Lorentz form of relativity. The time ref\/lection corresponds to an involution $\tau $ of the euclidean motion group~$G$ and then the duality of symmetric Lie algebras $\fg=\fg^\tau \oplus \fg^{-\tau} \leftrightarrow \fg^c=\fg^\tau \oplus i\fg^{-\tau}$ leads to a duality between the Lie algebra of the euclidean motion group and that of the Poincar\'e group.

In the following we call a pair $(G,\tau)$, consisting of a Lie group $G$ and an involutive automorphism $\tau\colon G\to G$ a \textit{symmetric Lie group}. On the group level, the aforementioned duality is implemented by the {\it Cartan duality} (or~$c$-duality for short) $G\leftrightarrow G^c$, where~$G^c$ is the simply connected Lie group with Lie algebra~$\fg^c$. Passing from real to imaginary time and from the euclidean motion group to the universal covering of the (identity component of the) Poincar\'e group is an important special case of $c$-duality which has been studied independently from the quantum f\/ield theoretic context.

In addition to $c$-duality, the other three basic notions considered in the basic theory of ref\/lection positivity are that of ref\/lection positive Hilbert spaces, ref\/lection positive kernels, and ref\/lection positive representations. A ref\/lection positive representation together with the Osterwalder--Schrader quantization leads to an inf\/initesimally unitary representation of the Lie algebra $\fg^c$. The problem in representation theory is then to determine if this representation arises as the derived representation of a unitary representation of $G^c$, which establishes a passage from a unitary representations of $G$ to one of~$G^c$.

This integration process is often accomplished using ref\/lection positive kernels and geometric actions of the Lie group or its Lie algebra~\cite{MNO15} or, as in the case of the L\"uscher--Mack theorem~\cite{LM75}, using semigroups and invariant cones. A second step is then to determine the resulting representation in terms of decomposition into irreducible representation. For further representation theoretic results related to ref\/lection positivity, we refer to \cite{HO96,J86,J87, JOl98,JOl00,NO14,NO15a,NO15b,O00,S86}. Our present paper concentrates on various mathematical aspects of constructions of ref\/lection positive representations: inf\/inite-dimensional analysis, functional integration and gaussian measures, and stochastic processes. On the way we recall several basic facts in these areas to make this article more self-contained.

Shortly after the groundbreaking work of Osterwalder and Schrader, A.~Klein and L.~Landau built a bridge between ref\/lection positivity, path spaces and stochastic processes \cite{Kl77,Kl78,KL75}. One of our goals in this article is to connect the ideas of Klein and Landau to representation theory and our previous work on ref\/lection positive representations~\cite{JOl98,JOl00,NO14,NO15a,NO15b}. This is done by replacing the real line, viewed as continuous time, by an arbitrary Lie group $G$, not necessarily f\/inite-dimensional, and the positive time axes $\R_+$ by a semigroup $S\subset G$ invariant under $s\mapsto s^\# = \tau (s)^{-1}$. This leads naturally to the concepts of $(G,\tau)$-measure spaces, ref\/lection positive measure spaces, and a positive semigroup structures introduced in Section~\ref{sec:1} of this article. Many concepts extend naturally to general triples $(G,S,\tau)$. But in the generalization of the Abel--Klein reconstruction theorem (Theorem~\ref{thm:1.8G2}) which reconstructs a $(G,S,\tau)$-measure space from a positive semigroup structure, we need to assume that $G=S\cup S^{-1}$; then $S$ is called {\it total}. It would be interesting to see how far our techniques can be extended beyond total subsemigroups.

In the following section we turn back to the classical case where the symmetric Lie group is $(\R, \R_+,-\id)$, but the measure space is the path space $P(Q) = Q^\R$ of a polish topological group $Q$, mostly assumed locally compact. If $\nu$ is a measure on $Q$ and $(\mu_t)_{t \geq 0}$ is a convolution semigroup of symmetric probability measures on $Q$ satisfying $\nu * \mu_t = \nu$ for every $t > 0$, then the reconstruction theorem applies to the corresponding $\R_+$-action on $L^2(Q,\nu)$ and leads to an invariant probability measure on the path space $Q^\R$ corresponding to a
ref\/lection positive representation (Theorem~\ref{thm:4.11}). However, the corresponding measure on $P(Q)$ is the product of the measure $\nu$, on the constant paths, and a probability measure on the pinned paths
\begin{gather*} P_*(Q) := \{ \omega \in P(Q) \colon \omega(0)= \1\}, \end{gather*}
which is also invariant under a suitable one-parameter group of transformations. Since we presently do not know how to obtain a similar factorization if $(\R,\R_+,-\id_\R)$ is replaced by some $(G,S,\tau)$, we discuss the one-dimensional case in some detail.

 One of the advantages of the ref\/lection positivity condition is that it allows to construct representations of a symmetric Lie group (classically the euclidean motion group) in $L^2$-spaces of measures on spaces of distributions, which after the Osterwalder--Schrader quantization lead to
 unitary representations of the $c$-dual group. In this connection the physics literature considered ref\/lection positive distributions $D \in \cS^\prime$, invariant under the euclidean group. The Hilbert space $\cH_D \subseteq \cS^\prime$ specif\/ied by the corresponding positive def\/inite kernel~\cite{Sch64} leads to a~ref\/lection positive representation of the euclidean motion group \cite{NO14, NO15a}. This motivates our discussion of $G$-invariant gaussian measures corresponding to a unitary representation $(\pi,\cH)$ of~$G$. The replacement for the Gelfand triple $(\cS,L^2(\R^n), \cS^\prime)$ is the Gelfand triple $(\cH^\infty, \cH,\cH^{-\infty})$, where~$\cH^\infty$ denotes the Fr\'echet space of smooth vectors in~$\cH$ and~$\cH^{-\infty}$ is it conjugate linear dual (the space of distribution vectors). Here the natural question is for which representations the gaussian measure of $\cH$ can be realized on $\cH^{-\infty}$? A partial answer is given in Theorem~\ref{thm:3.18}, Corollary~\ref{con:3.18}, and Proposition~\ref{prop:4.20}. All of this links naturally to the theory of generalized Wiener spaces and white noise processes for $\R^n$, indexed by $L^2(\R^n)$; see, e.g.,~\cite{AJ12,AJL11,Ap88,HOUZ10}.

Finally we would like to mention that there is also an important branch of applications of ref\/lection positivity in statistical physics
which does not refer to semigroups at all; see~\cite{FOS83} and the more recent \cite{NO15b}, where the corresponding group $G$ may be a torus, hence does not contain any proper open subsemigroup. We hope to develop this point in future work.

Our article is organized as follows. In the f\/irst section we start by recalling the basic facts about ref\/lection positive Hilbert spaces and ref\/lection positive representations~\cite{NO14, NO15a}. We always have a symmetric Lie group $(G,\tau)$ and a~subsemigroup $S$ invariant under $s^\# = \tau (s)^{-1}$. Based on ideas from~\cite{KL75}, we introduce the special class of ref\/lection positive Hilbert spaces of Markov type. The main result is Proposition~\ref{prop:3.9}.

In Section~\ref{sec:1} we discuss stochastic processes indexed by a symmetric Lie group $(G,\tau)$, where the forward time $\R_+$ is replaced by a subsemigroup $S\subset G$ invariant under $s \mapsto s^\sharp = \tau(s)^{-1}$. Here we prove our generalization of the reconstruction theorem (Theorem~\ref{thm:1.8G2}) which reconstructs for a so-called positive semigroup structure a $(G,S,\tau)$-measure space from which it can be derived. {}From the representation theoretic perspective, it corresponds to f\/inding euclidean realizations of unitary representations of the $c$-dual group~$G^c$. Unfortunately, the reconstruction process requires that $G=S\cup S^{-1}$, a property which is brief\/ly discussed in~Section~\ref{se:2.5}.

In Section~\ref{sec:3} we build a bridge between Markov processes and the reconstruction process. To this end, we return to the classical setting,
where the symmetric Lie group is $(\R, \R_+,-\id)$, but the measure space is the path space $P(Q) = Q^\R$ of a polish group $Q$ and the corresponding $(\R,\R_+,-\id_\R)$-measure space is given by a one-parameter group $P_t f = f * \mu_t$, $t \geq 0$, of left invariant Markov kernels on some space $L^2(Q,\nu)$. Here our main result concerns a factorization of the measure space $P(Q)$ as $Q \times P_*(Q)$ and a corresponding factorization of the measure preserving action of $\R$. For the special case where $\mu_t$ is the gaussian semigroup on $\R^d$, the corresponding measure on $P_*(\R^d)$
is the Wiener measure, but the measure on $P(\R^d)$ is a product of Lebesgue and Wiener measure. The main results in Section~\ref{sec:3} is Theorem~\ref{thm:4.11} which relates all this to the reconstruction process and hence to ref\/lection positive representations.

In Section~\ref{sec:2} we discuss the second quantization functor and how it can be used to derive from an orthogonal ref\/lection positive representation a gaussian $(G,S,\tau)$-probability spaces (Proposition~\ref{pr:4.6}). To understand the ambiguity in this construction we also
discuss equivalence of gaussian measures for reproducing kernel Hilbert spaces, and in Subsection~\ref{se:4.4} we connect this issue
to our previous work~\cite{NO14} on distributions on $G$ and distribution vectors of a unitary representation of $G$. In Theorem~\ref{thm:3.18} we determine when a gaussian measure $\gamma_{\cH}$ of a Hilbert space~$\cH$ can be realized in the dual of the space $\bigcap \cD (A^n)$ of smooth vector of a self adjoint operator~$A$. This is then applied to a unitary representation of~$G$ in Corollary~\ref{con:3.18}.

The article ends with two appendixes where needed material on stochastic processes and Markov semigroups is collected.

\section{Ref\/lection positive representations}\label{sec:0}

In this preliminary section we collect some results and def\/initions from \cite{NO14, NO15a} concerning ref\/lection positive representations.

\begin{Definition} \label{def:1.2} A {\it reflection positive Hilbert space} is a triple $(\cE,\cE_+,\theta)$, where $\cE$ is a Hilbert space, $\theta$ a unitary involution and $\cE_+$ a closed subspace which is {\it $\theta$-positive} in the sense that the hermitian form $\la v,w\ra_\theta := \la \theta v, w \ra$ is positive semidef\/inite on~$\cE_+$.
\end{Definition}

For a ref\/lection positive Hilbert space $(\cE,\cE_+,\theta)$, let $\cN:=\{u \in\cE_+\colon \du{\theta u}{u}=0\}$ and let $\hat\cE$ be the completion of $\cE_+/\cN$ with respect to the inner product $\due_\theta$. Let $q \colon \cE_+ \to \hat\cE, v \mapsto q(v)=\hat v$ be the canonical map. Then $\cE_+^\theta := \{ v \in \cE_+ \colon \theta v = v\}$ is the maximal subspace of $\cE_+$ on which $q$ is isometric.

\begin{Definition} Let $(\cE,\cE_+,\theta)$ be a ref\/lection positive Hilbert space. If $\cE_0 \subeq \cE_+^\theta$ is a closed subspace,
$\cE_- := \theta(\cE_+)$, and $E_0, E_\pm$ the orthogonal projections onto $\cE_0$ and $\cE_\pm$, then we say that $(\cE,\cE_0, \cE_+,\theta)$ is a
{\it reflection positive Hilbert space of Markov type} if
\begin{gather}\label{eq:Markov}
E_+ E_0 E_- = E_+ E_-.
\end{gather}
\end{Definition}

\begin{Lemma}\label{le:Markov} Suppose that $\cE$ is a Hilbert space, $\theta\colon \cE\to \cE$ is a unitary involution, $\cE_+$ is a closed
subspace and $\cE_0\subseteq \cE_+^\theta$. Let $\cE_-=\theta (\cE_+)$.
 \begin{itemize}\itemsep=0pt
 \item[\rm{(i)}] If $(\cE,\cE_0,\cE_+,\theta)$ satisfies the Markov condition~\eqref{eq:Markov}, then $\cE_+$ is $\theta$-positive, $\cE_0
 =\cE_+^\theta$, $\cN=\cE_+\ominus \cE_0$ and $q\colon \cE_0\to \hat\cE$ is a unitary isomorphism.
 \item[\rm{(ii)}] If $q_0\colon \cE_0\to \hat\cE$ is a unitary isomorphism, then~\eqref{eq:Markov} holds.
 \end{itemize}
\end{Lemma}

\begin{proof} (i) By def\/inition, we have $\theta E_+=E_-\theta$ so $\cE^\theta_+ \subseteq \cE_+\cap \cE_-$. Let
$u\in\cE^\theta_+\ominus \cE_0$. Then the operator on the right hand side of \eqref{eq:Markov} reproduces $u$ while
the left hand side yields $0$. Hence $u=0$ and $\cE_0=\cE_+^\theta$. As $\cE_0\subset \cE_+^\theta$ it follows that
$E_0 \theta = \theta E_0 = E_0$ which implies that
\begin{gather*} E_+ \theta E_+ = E_+ E_- \theta = E_+ E_0 E_- \theta = E_0 \theta E_+ = E_0 E_+ = E_0 .\end{gather*}
It follows that, for $u\in\cE_+$, we have
\begin{gather*}\du{\theta u}{u}= \du{E_+\theta E_+u}{u}=\du{E_0u}{u}=\|E_0u\|^2\ge 0 .\end{gather*}
We obtain in particular that $\|q(v)\| = \|E_0v\|$ for $v \in \cE_+$, so that
$q\res_{\cE_0} \colon \cE_0 \to \hat\cE$ is a unitary isomorphism.

(ii) Since $q\res_{\cE_+^\theta}$ is isometric, we obtain
$\cE_0=\cE_+^\theta$ and thus $\cN=\cE_+ \cap \theta(\cE_0)^\bot
= \cE_+\ominus\cE_0$. This leads to the orthogonal decomposition
$\cE_+ + \cE_- = \theta(\cN) \oplus \cE_0 \oplus \cN$
and to $\cE_0 = \cE_+ \cap \cE_-$. Now~\eqref{eq:Markov} follows.
\end{proof}

In the following, if $(\cE,\cE_0,\cE_+,\theta)$ satisf\/ies the Markov
condition~(\ref{eq:Markov}), then we will always assume that
$\cE_0=\cE_+^\theta$ and hence only write $(\cE,\cE_+,\theta)$.

\begin{Definition} If $\tau$ is an involutive automorphism of the Lie group $G$, then we call $(G,\tau)$ a~{\it symmetric Lie group}.
A~{\it symmetric semigroup} is a triple $(G,S,\tau)$, where $(G,\tau)$ is a symmetric Lie group and $S \subeq G$ is a subsemigroup satisfying
\begin{itemize}\itemsep=0pt
\item[\rm(S1)] $S$ is invariant under $s \mapsto s^\sharp = \tau(s)^{-1}$,
\item[\rm(S2)] $HS = S$ for $H := (G^\tau)_0$,
\item[\rm(S3)] $\1 \in \oline S$.
\end{itemize}

We def\/ine a left invariant {\it partial order $\prec_S$} on $G$ by $g \prec_S h$ if $g^{-1}h \in S$, i.e., $h \in gS$.
\end{Definition}

\begin{Example} \label{ex:semigroups}
(1) $(\R,\R_+, -\id_\R)$ and $(\Z,\N_0, -\id_\Z)$ are the most
elementary examples of symmetric semigroups.

(2) Semigroups with polar decomposition: Let $(G,\tau)$ be a symmetric Lie group and $H$ be an open subgroup of $G^\tau :=\{g\in G\colon \tau (g)=g\}$. We denote the derived involution $\mfg \to \mfg $ by the same letter and def\/ine $\mfh = \{x\in \mfg \colon \tau (x)=x\} =\mfg^\tau$ and $\mfq =\{x\in \mfg\colon \tau (x)=-x\}=\mfg^{-\tau}$. Then $\mfg = \mfh \oplus\mfq$. We say that the open subsemigroup $S\subeq G$ has a polar decomposition if there exists an $H$-invariant open convex cone $C\subset \mfq$ containing no af\/f\/ine lines such that $S=H\exp C$ and the map $H\times C\to S$, $h,X\mapsto h\exp X$ is a dif\/feomorphism (cf.~\cite{HO96, La94,Ne00}). Typical examples are the complex
Olshanski semigroups in complex simple Lie groups such as $\SU_{p,q}(\C)_\C \cong \SL_{p+q}(\C)$. They exist if and only if $G/K$ is a bounded symmetric domain. This is equivalent to the existence of a $G$-invariant convex cone $C\subset i\mfg$ such that $G\exp C$ is a semigroup. More generally we have the causal symmetric spaces of non-compact type like $\SO_o (1,n+1)/\SO_o (1,n)$. In this case $\fq\simeq \R^{n+1}$ and $C$ can be
taken as the open light-cone.

(3) The simply connected covering group $G := \tilde\SL_2(\R)$ of $\SL_2(\R)$ carries an involution $\tau$ acting on $\fsl_2(\R)$ by
\begin{gather*} \tau\pmat{x & y \\ z & -x} = \pmat{x & -y \\ -z & -x},\end{gather*}
and there exists a closed subsemigroup $S \subeq G$ whose boundary is
\begin{gather*} \partial S = H(S) := S \cap S^{-1} = \exp(\fb) \qquad \mbox{with}
\quad \fb := \left\{ \pmat{ x & y \\ 0 & -x} \colon x,y \in \R\right\}.\end{gather*}
This semigroup satisf\/ies $S^\sharp = S$, the subgroup $H(S)$ is $\tau$-invariant, but strictly larger than $G^\tau_0$ (see also Section~\ref{se:2.5} for more on this semigroup).
\end{Example}

\begin{Definition} \label{def:1.4}
For a symmetric semigroup $(G,S,\tau)$, a unitary representation $U$ of $G$ on $(\cE,\cE_+,\theta)$ is called {\it reflection positive} if $\theta U_g \theta= U_{\tau(g)}$ for $g \in G$ and $U_s\cE_+ \subeq \cE_+$ for every $s \in S$.
\end{Definition}

\begin{Remark} Let $G_\tau =G\rtimes \{1,\tau\}$. Then $\theta U_g\theta =U_{\tau (g)}$ holds for every $g \in G$ if and only if $U$ extends to
a unitary representation of $G_\tau$ by def\/ining $U_\tau = \theta$.
\end{Remark}

\begin{Remark} \label{rem:dil}
(a) If $(U_g)_{g \in G}$ is a ref\/lection positive representation of $(G,S,\tau)$ on $(\cE,\cE_+,\theta)$, then we obtain contractions
$(\hat U_s)_{s \in S}$, on $\hat\cE$, determined by
\begin{gather*} \hat U_s \circ q = q \circ U_s\res_{\cE_+},\end{gather*}
and this leads to an involutive representation $(\hat U,\hat\cE)$ of~$S$ by contractions (cf.\ \cite[Corollary~ 3.2]{JOl98} and~\cite{NO14}). We then call $(U,\cE,\cE_+,\theta)$ a~{\it euclidean realization} of $(\hat U,\hat\cE)$. We refer to~\cite{MNO15} and~\cite{LM75} for methods for obtaining a unitary representation of $G^c$ from $\hat U$.

(b) For $(G,S,\tau) = (\R,\R_+,-\id_\R)$, continuous ref\/lection positive unitary one-parameter groups $(U_t)_{t \in \R}$ lead to a strongly continuous semigroup $(\hat U,\hat\cE)$ of hermitian contractions and every such semigroup $(C,\cH)$ has a natural euclidean realization obtained as the GNS representation associated to the positive def\/inite operator-valued function $\vphi(t) := C_{|t|}$, $t \in \R$ \cite[Proposition~6.1]{NO15a}.
\end{Remark}

The following proposition is a generalization of \cite[Proposition~5.17]{NO15a} which applies to the special case $(\R,\R_+, -\id_\R)$.

\begin{Proposition} \label{prop:3.9}
Let $(U_g)_{g \in G}$ be a reflection positive unitary representation of $(G,S,\tau)$ on $(\cE, \cE_+, \theta)$, let $\cE_0 \subeq (\cE_+)^\theta$ be a subspace and $\Gamma=q|_{\cE_0}\colon \cE_0\to\hE$. If $(\cE,\cE_0,\cE_+,\theta)$ is of Markov type, then the following assertions hold:
\begin{itemize}\itemsep=0pt
\item[\rm(i)] The reflection positive function $\vphi \colon G \to B(\cE_0)$, $\vphi(g) := E_0 U_g E_0$, is multiplicative on $S$.
\item[\rm(ii)] $\vphi(s) = \Gamma^* \hat U_{s} \Gamma$ for $s\in S$, i.e., $\Gamma$ intertwines $\vphi\res_{S}$ with $\hat U$.
\end{itemize}
\end{Proposition}

\begin{proof}
(i) Let $\cK \subeq \cH$ be the $U$-invariant closed subspace generated by $\cE_0$. Let $(\cE_0)^G$ denote the linear space of all maps $G \to \cE_0$. Then the map
\begin{gather*} \Phi \colon  \ \cK \to (\cE_0)^G, \qquad
\Phi(v)(g) := E_0 U_g v \end{gather*}
is an equivalence of the representation $U$ of $G$ on $\cK$ with the GNS representation def\/ined by~$\vphi$, and the representation $\hat U$ of $S$ on~$\hat\cE$ is equivalent to the GNS representation def\/ined by~$\vphi\res_S$, where the map $q \colon \cE_+ \to \hat\cE$ corresponds to the restriction $f \mapsto f\res_S$ \cite[Proposition~1.11]{NO14}. The inclusion $\iota \colon \cE_0 \into \cH_\vphi$ is given by $\iota(v)(g) = E_0 U_gv = \vphi(g)v$ for $g \in G$, and likewise the inclusion $\hat\iota \colon \cE_0 \into \cH_{\vphi\res_S}$ is given by $\iota(v)(s) = \vphi(s)v$ for $s\in S$. Lemma~\ref{le:Markov} implies the surjectivity of the inclusion $\hat\iota$. In view of \cite[Lemma~10.3]{NO15a}, this is equivalent to the multiplicativity of $\vphi\res_S$.

(ii) If $\Gamma$ is unitary, then \cite[Lemma~5.16(ii)]{NO15a} implies that $q = \Gamma \circ E_0\res_{\cE_+}$. For $s\in S$, the relation $\hat U_s \circ q = q \circ U_s\res_{\cE_+}$ leads to $\hat U_s \Gamma E_0\res_{\cE_+} = \Gamma E_0 U_s\res_{\cE_+}$, so that $\Gamma^* \hat U_s \Gamma = E_0 U_s E_0 = \vphi(s)$, i.e., $\Gamma$ intertwines $\vphi(s)$ with $\hat U_s$.
\end{proof}

\begin{Lemma} Assume that $U$ is a reflection positive representation of $(G,S,\tau)$ on $(\cE,\cE_+,\theta)$. Let $\cE_0\subseteq \cE_+^\theta$. If $\varphi (g)=E_0U_gE_0$ is multiplicative and $\cE_0$ is $S$-cyclic in $\cE_+$ then $(\cE,\cE_+,\theta)$ is of Markov type.
\end{Lemma}
\begin{proof} As in the proof of Lemma~\ref{le:Markov}, this follows from \cite[Lemma~10.3]{NO15a}.
\end{proof}

\begin{Remark} If $G = S \cup S^{-1}$ and $\cE_0$ is $G$-cyclic in $\cE$. Then $\cE_0$ is $S$-cyclic in $\cE_+$ and $\cE_+ + \cE_-$ is dense in $\cE$, so that the Markov condition leads to $\cE = \cN \oplus \cE_0 \oplus \theta(\cN)$.
\end{Remark}

\section{Ref\/lection positive Lie group actions on measure spaces}\label{sec:1}

In this section we generalize several results from \cite{Kl78} and \cite{KL75} to situations where $(\R,\R_+,-\id_\R)$ is replaced by a symmetric semigroup $(G,S,\tau)$. This leads us to the concept of a $(G,S,\tau)$-measure space generalizing Klein's Osterwalder--Schrader
path spaces for $(\R,\R_+,-\id_\R)$. The Markov $(G,S,\tau)$-measure spaces generalize the path spaces studied by Klein and Landau in~\cite{KL75}.
The main result of this section is the correspondence between $(G,S,\tau)$-measure spaces and the corresponding positive semigroup structures. For $(G,S,\tau)= (\R,\R_+,-\id_\R)$ this has been developed in~\cite{Kl78, KL75}, motivated by Nelson's work on the Feynman--Kac formula~\cite{Nel64}.

\subsection[$(G,S,\tau)$-measure spaces]{$\boldsymbol{(G,S,\tau)}$-measure spaces}
In this section we introduce basic concepts related to~$(G,S,\tau)$-measure spaces and draw some consequences from the def\/initions.

\begin{Definition} \label{def:0.1G} Let $(G,\tau)$ be a group with an involutive automorphism~$\tau$. A {\it $(G,\tau)$-measure space} is a quadruple $((Q,\Sigma,\mu), \Sigma_0, U,\theta)$ consisting of the following ingredients:
\begin{itemize}\itemsep=0pt
\item[\rm(GP1)] a measure space $(Q,\Sigma,\mu)$,
\item[\rm(GP2)] a sub-$\sigma$-algebra $\Sigma_0$ of $\Sigma$,
\item[\rm(GP3)] a measure preserving action $U \colon G \to
\Aut(\cA)$ on the $W^*$-algebra $\cA := L^\infty(Q,\Sigma,\mu)$
which is strongly continuous in measure, i.e.\footnote{Since the $G$-action on $\cA$ is measure preserving, it def\/ines natural
representations on $L^1(Q,\Sigma,\mu)$ and on $L^2(Q,\Sigma,\mu)$.
However, in general we do not have an action of $G$ on the set $Q$ itself,
but $G$ acts naturally on the set~$\Sigma/{\sim}$, where
$A \sim B$ if $\mu(A \Delta B) = 0$. This set corresponds to the idempotents
in the algebra~$\cA$. The continuity condition
(GP3) is equivalent to the continuity of the corresponding
unitary representation of $G$ on $L^2(Q,\Sigma,\mu)$
which in turn is equivalent to $\mu(gE\Delta E) \to 0$ for $g \to\1$ and
$E \in \Sigma$ with $\mu(E) < \infty$.
See the discussion in Appendix~\ref{subsec:6.3} or \cite[p.~107]{Si74}
for more details.},
\begin{gather*} \lim_{g \to \1} \mu(|U_g f -f| \geq \eps) = 0 \qquad \mbox{for} \quad \eps > 0
\qquad \mbox{and} \qquad
f \in \cA \cap L^1(Q,\Sigma,\mu),\end{gather*} and
\item[\rm(GP4)] a measure preserving involutive automorphism $\theta$ of
$L^\infty(Q,\Sigma,\mu)$
for which $\theta U_g \theta = U_{\tau(g)}$ for $g \in G$ and $\theta E_0 \theta = E_0$,
where
$E_0 \colon L^\infty(Q,\Sigma,\mu) \to L^\infty(Q,\Sigma_0,\mu)$ is
the conditional expectation.
\item[\rm(GP5)] $\Sigma$ is generated by the sub-$\sigma$-algebras $\Sigma_g
:= U_g \Sigma_0$, $g \in G$.
\end{itemize}
If $\mu$ is a probability measure, we speak of a {\it $(G,\tau)$-probability space}.
If $(G,S,\tau)$ is a symmetric semigroup, then we
write $\Sigma_\pm$ for the sub-$\sigma$-algebra generated by
$(\Sigma_s)_{s \in S^{\pm 1}}$,
and $E_\pm$ for the corresponding conditional expectations.
\end{Definition}

\begin{Definition} (a) A $(G,\tau)$-measure space is called
{\it reflection positive with respect to $S$} if
\begin{gather*} \la \theta f, f \ra\geq 0 \qquad \mbox{for} \quad f \in
\cE_+ := L^2(Q,\Sigma_+, \mu).\end{gather*}
This is equivalent to $E_+ \theta E_+ \geq 0$ as an operator on
$L^2(Q,\Sigma,\mu)$ and obviously implies $\theta E_0 = E_0$.
If this condition is satisf\/ied and, in addition,
$\Sigma_0$ is invariant under the group $H(S) := S\cap S^{-1}$,
then we call it a {\it $(G,S,\tau)$-measure space}\footnote{Note that $E_+ \theta E_+ \geq 0$ is equivalent to the condition that the kernel
$K(A,B) := \mu(A \cap \theta(B))$ on $\Sigma_+$ is positive def\/inite.}.

(b) A {\it Markov $(G,S,\tau)$-measure space} is a
$(G,S,\tau)$-measure space with the {\it Markov property} $E_+ E_- = E_+ E_0 E_-$
(cf.\ Def\/inition~\ref{def:1.2}(b)).
\end{Definition}

Proposition~\ref{prop:3.9} now implies:

\begin{Proposition} \label{prop:markov}
For any $(G,S,\tau)$-measure space $((Q,\Sigma,\mu), \Sigma_0, U,\theta)$, we put $\cE := L^2(Q,\Sigma,\mu)$, $\cE_0 := L^2(Q,\Sigma_0,\mu)$ and
$\cE_\pm:= L^2(Q,\Sigma_\pm,\mu)$. Then the natural action of $G$ on $\cE$ defines a~reflection positive representation of $(G,S,\tau)$. The Markov property is equivalent to the natural map $\cE_0 \to \hat\cE$ being unitary and this implies that the function $\vphi \colon S \to B(\cE_0), \vphi(s) = E_0 U_s E_0$
is multiplicative.
\end{Proposition}

\begin{Definition} \label{def:1.5G}
(a) A {\it positive semigroup structure} $(\cH, P, \cA, \Omega)$ for a
symmetric semigroup $(G,S,\tau)$ consists of
\begin{itemize}\itemsep=0pt
\item[\rm(PS1)] a Hilbert space $\cH$,
\item[\rm(PS2)] a strongly continuous $*$-representation
$(P_s)_{s \in S}$ of $(S,\sharp)$ by contractions on~$\cH$,
\item[\rm(PS3)] a commutative von Neumann algebra $\cA$ on $\cH$ normalized
by the operators $(P_h)_{h \in H(S)}$, and
\item[\rm(PS4)] a unit vector $\Omega \in \cH$, such that
 \begin{itemize}\itemsep=0pt
 \item[\rm(a)] $P_s\Omega = \Omega$ for every $s \in S$.
 \item[\rm(b)] $\Omega$ is cyclic for the (not necessarily selfadjoint)
subalgebra $\cB \subeq B(\cH)$ generated by $\cA$ and $\{ P_s \colon s \in S\}$.
 \item[\rm(c)] For positive elements $A_1,\ldots, A_n \in \cA$
and $s_1,\ldots, s_{n-1} \in S$, we have
\begin{gather*} \la A_1 P_{s_1} A_2 \cdots P_{s_{n-1}} A_n \Omega, \Omega \ra \geq 0.\end{gather*}
 \end{itemize}
\end{itemize}

(b) A {\it standard positive semigroup structure} for a
symmetric semigroup $(G,S,\tau)$ consists of a $\sigma$-f\/inite measure
space $(M, \fS, \nu)$ and
\begin{itemize}\itemsep=0pt
\item[\rm(SPS1)] a representation
$(P_s)_{s \in S}$ of $S$ on $L^\infty(M,\nu)$ by positivity preserving operators, i.e.,
$P_s f \geq 0$ for $f \geq 0$.
\item[\rm(SPS2)] $P_s 1 = 1$ for $s \in S$ (the Markov condition).
\item[\rm(SPS3)] $P$ is involutive with respect to $\nu$, i.e.,
\begin{gather*} \int_M P_s(f) h d\nu = \int_M f P_{s^\sharp}(h)\, d\nu \qquad \mbox{for} \quad
s \in S, \quad f,h \geq 0.\end{gather*}
\item[\rm(SPS4)] $P$ is strongly continuous in measure, i.e.,
for each $f \in L^1(M ,\nu) \cap L^\infty(M,\nu)$ and every $\delta > 0$,
$s_0 \in S$, we have
$\lim\limits_{s \to s_0} \nu(\{|P_s f - P_{s_0}f| \geq \delta\}) = 0$.
\end{itemize}
\end{Definition}

The preceding def\/inition generalizes the corresponding classical concepts
for the case $(G,S,\tau)$ $= (\R,\R_+,-\id_\R)$ (\cite{Kl78} for~(a) and \cite{KL75} for~(b)).

\begin{Remark} (a) For an operator $P$ on
$L^2(X,\mu)$, the condition to be positivity preserving does not imply that
$P$ is symmetric. In fact, if $(Pf)(x) = f(Tx)$ for some measure preserving
map $T \colon X \to X$, then $P$ is positivity preserving and isometric.
Hence we cannot expect $P$ to be symmetric if $T$ is not an involution.

(b) Applying (SPS3) with $h = 1$, we obtain $\int_M P_s(f)\, d\mu = \int_M f\, d\mu$,
i.e., that the measure $\mu$ is invariant under the operators $(P_s)_{s \in S}$.

(c) If $P \colon M \times \fS \to [0,\infty]$ is a Markov kernel
(cf.\ Appendix~\ref{subsec:6.4}), then{\samepage
\begin{gather*} (Pf)(x) := \int_M P(x,dy) f(y) \end{gather*}
def\/ines a positivity preserving operator on $L^\infty(M,\fS,\nu)$ satisfying $P1 = 1$.}

In view of \cite[Lemma~36.2]{Ba96}, kernels on $M \times \fS$ are
in one-to-one correspondence with additive, positively homogeneous maps
$T$ from the convex cone of non-negative measurable functions $M \to [0,\infty]$ into itself which are {\it Daniell continuous} in the
sense that $\lim\limits_{n\to \infty} T(f_n) = T\big(\lim\limits_{n\to \infty} f_n\big)$ for monotone sequences
$f_n \leq f_{n+1}$.\footnote{The notion of a positivity preserving
operator on $L^\infty(M,\fS,\nu)$ is slightly weaker than this concept.
In particular, it only operates on equivalence classes of functions
in $L^\infty(X,\fS,\nu)$ and not on functions itself. In concrete situations,
the positivity preserving operators actually come from kernels, which makes
them easier to deal with.}
\end{Remark}

\begin{Remark} \label{rem:2.7}
We consider a standard positive semigroup structure for $(G,S,\tau)$.

(a) Conditions (SPS1/2) imply that
$\|P_s\| \leq 1$ on $L^\infty(M,\nu)$.
Further (SPS2/3) imply that the restriction of $P_s$ to
$L^1 \cap L^\infty$ is measure preserving, and from \cite[Proposition~1.2(i)]{KL75}
it follows that $P_{s^\sharp s}$ def\/ines a contraction on $L^2(M,\nu)$.
This implies that we obtain a $*$-representation of~$(S,\sharp)$ on~$L^2(M,\nu)$.
The proof of \cite[Proposition~1.2(ii)]{KL75} further shows that~(SPS4)
implies that this representation is strongly continuous because we have for
$f \in L^1 \cap L^2$:
\begin{align*}
 \|P_sf - P_{s_0}f\|_2^2
&\leq
\int_{\{|P_s f - P_{s_0} f|\geq \eps\}} |P_s f - P_{s_0}f|^2\, d\nu
+ \int_{\{|P_s f - P_{s_0} f|\leq \eps\}}|P_s f - P_{s_0}f|^2\, d\nu \\
&\leq
(2\|f\|_\infty)^2 \nu(\{|P_s f - P_{s_0} f|\geq \eps\})
+ 2\eps \|f\|_1.
\end{align*}

For $h \in H(S)$, both operators $P_h$ and $P_{h^{-1}} = P_h^{-1}$
are positivity preserving and f\/ix $1$, and since
$L^\infty(M,\nu)$ is a commutative von Neumann algebra,
they are algebra automorphisms by \cite[Theorem~3.2.3]{BR87}.
The relation $P_h(fg) = P_h(f) P_h(g)$ then leads to
\begin{gather*}
 P_h M_f P_h^{-1} = M_{P_h(f)}
\end{gather*}
for the multiplication operators $M_f$. In particular, the
action of $H(S)$ on $L^2(M,\nu)$ normalizes $L^\infty(M,\nu)$.

(b) If $\nu$ is a probability measure, $\Omega := 1$ and $\cA := L^\infty(M,\nu)$, then the preceding discussion shows that
we also have a~positive semigroup structure in the sense of Def\/inition~\ref{def:1.5G}(a)
for which~$1$ is a cyclic vector for~$\cA$. Note that Def\/inition~\ref{def:1.5G}(b) is not a special case of (a) because
it does not require the measure $\nu$ to be f\/inite.
\end{Remark}

The following proposition shows that the requirement that $\Omega$ is cyclic for $\cA$ describes those positive semigroup
structures which are standard.

\begin{Proposition}[\protect{\cite[Proposition~3.5]{Kl78}}] \label{prop:2.19}
Let $(M,\fS,\nu)$ be a probability space and $(P_s)_{s \in S}$ be a~positivity
preserving continuous $*$-representation of $S$ by contractions on $L^2(M,\nu )$, i.e.,
\begin{gather*} P_s 1 = 1 \qquad \mbox{and} \qquad
 P_s f \geq 0 \qquad \mbox{for} \quad f\geq 0, \quad s\in S.\end{gather*}
Then $(L^2(M), Q, L^\infty(M), 1)$ is a positive semigroup structure for which~$1$ is a cyclic vector for~$L^\infty(M)$.

Conversely, let $(\cH,P, \cA, \Omega)$ be a positive semigroup structure for
which $\Omega$ is a cyclic vector for~$\cA$. Then there exists a probability
space~$M$ and a positivity preserving semigroup~$\tilde Q$ on~$L^2(M)$ such that
$(\cH,P,\cA,\Omega) \cong (L^2(M), \tilde Q, L^\infty(M), 1)$ as positive
semigroup structures.
\end{Proposition}

\begin{proof} The f\/irst part is an immediate consequence of the def\/initions
(see also Remark~\ref{rem:2.7}(a)), so we only have to prove the second statement.
So let $(\cH,P, \cA, \Omega)$ be a positive semigroup structure for which $\Omega$ is a cyclic vector for~$\cA$.
Let $M$ be a compact space with $\cA \cong C(M)$. Then we obtain on $M$ a~probability measure~$\nu$,
def\/ined by
\begin{gather*} \int_{M} f(x)\, d\nu(x) = \la f \Omega, \Omega \ra \qquad \mbox{for} \quad f \in \cA.\end{gather*}
Since $\Omega$ is $\cA$-cyclic, the map
$C(M) \to \cH, f \mapsto f \Omega$ extends to a unitary map
$L^2(M,\nu) \to \cH$ intertwining $\cA$ with the von Neumann algebra $L^\infty(M,\nu)$.
It is clear that (SPS1-3) are satisf\/ied and Lemma~\ref{lem:6.7} implies the
continuity in measure.
\end{proof}

\begin{Example}[the real oscillator semigroup] \label{ex:osci}
We consider the Hilbert space $\cH = L^2(\R^d)$, with respect to Lebesgue measure.

(a) On $\cH$ we have a unitary representation by the group $\GL_d(\R)$ by
\begin{gather*} (T_h f)(x) := |\det(h)|^{-d/2} f\big(h^{-1}x\big)\qquad \mbox{for} \quad h \in \GL_d(\R), \quad x \in \R^d,\end{gather*}
and we also have two representations of the abelian semigroup
$\Sym_d(\R)_+$ (the convex cone of positive semidef\/inite matrices):
\begin{itemize}\itemsep=0pt
\item[\rm(a)] Each $A \in \Sym_d(\R)_+$ def\/ines a multiplication operator
$(M_A f)(x) := e^{-\la Ax,x\ra} f(x)$ which is positivity preserving on $L^\infty(\R^n)$
but does not preserve~$1$; it preserves the Dirac measure~$\delta_0$ in the origin.
\item[\rm(b)] Each $A \in \Sym_d(\R)_+$ specif\/ies a uniquely determined
(possibly degenerate) gaussian measure $\mu_A$ on $\R^d$ whose Fourier transform is
given by $\hat\mu_A(x) = e^{-\la Ax,x\ra/2}$. Then
the convolution operator $C_A f := f * \mu_A$ is positivity preserving
and leaves Lebesgue measure on $\R^d$ invariant. For $A = \1$,
we thus obtain the heat semigroup as $(\mu_{t\1})_{t \geq 0}$.
\end{itemize}

Any composition of these $3$ types of operators $T_h, M_A$ and $C_A$
is positivity preserving on $L^\infty(\R^d)$, and they generate a~$*$-representation of the Olshanski semigroup $S := H \exp(C)$ in the symmetric Lie group
$G := \Sp_{2d}(\R)$, where $H = \GL_d(\R)$, $C = \Sym_d(\R)_+ \times \Sym_d(\R)_+ \subeq
\Sym_d(\R)^2 = \fq$, and
\begin{gather*} \tau\pmat{A & B \\ C & - A^\top} = \pmat{A & -B \\ -C & - A^\top}
\qquad \mbox{for} \quad
\pmat{A & B \\ C & - A^\top} \in \sp_{2d}(\R) \end{gather*}
with $B^\top = B$, $C^\top = C$ (cf.\ Examples~\ref{ex:semigroups}).
The real Olshanski semigroup $S$ is the f\/ixed point set an antiholomorphic involutive
automorphism of the so-called oscillator semigroup $S_\C = G^c \exp(W)$ which is a~complex Olshanskii semigroup~\cite{Hi89, How88}.
The elements in the interior of~$S$ act on~$L^2(\R^n)$ by kernel operators
with positive gaussian kernels and the elements of~$S_\C$ correspond to
complex-valued gaussian kernels. The semigroup~$S$ contains many
interesting symmetric one-parameter semigroups such as the Mehler semigroup
$e^{-t H_{\rm osc}}$ generated by the oscillator Hamiltonian
\begin{gather} \label{eq:osc-ham}
H_{\rm osc}= -\sum_{j = 1}^n \partial_j^2 + \frac{1}{4} \sum_{j = 1}^n x_j^2 - \frac{d}{2} \1,
\end{gather}
which f\/ixes the Gaussian $e^{-\|x\|^2/4}$.

(b) The parabolic subgroup $P := \Sym_d(\R) \rtimes \GL_d(\R)_+ \subeq \Sp_{2d}(\R)$
is also a symmetric Lie subgroup with
$\fh = \gl_d(\R)$ and $\fq \cong \Sym_d(\R)$.
Here the commutative von Neumann algebra $\cA = L^\infty(\R^d)$ is invariant under
conjugation with the operators $T_h$, so that
$(A,h) \mapsto C_A T_h$ def\/ines a $*$-representation of the semigroup
$S := \Sym_d(\R)_+ \rtimes H$ that leads to a standard positive semigroup
structure on $L^2(\R^d,dx)$.
\end{Example}

The preceding example naturally extends to inf\/inite-dimensional spaces as follows.

\begin{Example} Let $\cH$ be a real Hilbert space.

(a) In Lemma~\ref{lem:3.6} below we show that every
continuous $*$-representation of a topological involutive semigroup
$(S,\sharp)$ by contractions on a real Hilbert space $\cH$ def\/ines a
standard positive semigroup structure on the gaussian probability space
$(\cH^a,\gamma_\cH)$. Here $\cH^a$ is the algebraic dual space of~$\cH$
and $\gamma_\cH$ is the unique probability measure with Fourier
transform $\hat\gamma_\cH(v) = e^{-\|v\|^2/2}$
(Example~\ref{ex:1.1}).

(b) Some of the structure from Example~\ref{ex:osci} extends
to inf\/inite-dimensional Hilbert spaces~$\cH$.
For the multiplication operator $M_A$ to be non-zero, one has
to require that $A \geq 0$ is trace class (cf.\ \cite[p.~153]{Ya85}).
Likewise, the measures $\mu_A$ exist if $A$ is trace class. In an orthonormal
basis $(e_n)_{n \in \N}$ in which $A$ is diagonal with $A e_n = a_n e_n$, the
measure $\mu_A$ is $\otimes_{n = 1}^\infty \gamma_{a_n t}$, where
$d\gamma_t(x) = (2\pi t)^{-1/2} e^{-x^2/2t}\, dx$ are the centered gaussian
measures on $\R$. For a linear operator $h \in \GL(\cH)$ to act on
$L^2(\cH^a, \gamma_\cH)$, it is necessary and suf\/f\/icient that
$g^\top g - \1$ is Hilbert--Schmidt, i.e., that the
polar decomposition has the form $g = u e^X$ with $\|X\|_2 < \infty$
(see also Theorem~\ref{thm:gauss-ker-equiv} and~\cite{NO02} for more details).

(c) For $\cH = \R^d$, we have $d\gamma(x) := d\gamma_\cH(x) = (2\pi)^{-d/2} e^{-\frac{\|x\|^2}{2}}\, dx$
(Example~\ref{ex:gauss-rd}), we obtain in particular a standard positive semigroup
structure for the contraction semigroup
\begin{gather*} \cC\big(\R^d\big):= \{ g \in \GL_d(\R) \colon \|g\| \leq 1\} = \OO_d(\R) \exp(-\Sym_d(\R)_+),\end{gather*}
where $\Sym_d(\R)_+$ denotes the closed convex cone of positive semidef\/inite
symmetric $(d\times d)$-matrices. The corresponding Markov operators are given by
$(\Gamma(h)f)(x) = f(h^{-1}x)$ for $h \in \OO_d(\R)$. Since every element
$s \in \cC(\R^d)$ has a polar decomposition $s = h\exp(-X)$, $X^\top = X \geq 0$,
diagonalization of $X$ reduces the description of the corresponding operator
to the case $d = 1$. For $0 < c < 1$ we have
\begin{gather*} (\Gamma(c)f)(x)
= \int_{\R^d} \tilde K_c(x,y)f(y)\, d{\gamma}(y)
\qquad \mbox{with} \quad \tilde K_c(x,y)
= \big(1-c^2\big)^{-d/2} e^{\frac{\|y\|^2}{2}} e^{- \frac{ \|cx-y\|^2}{2(1-c^2)}}
\end{gather*}
(cf.~\cite[p.~218]{Nel73b}). For $c = e^{-t}$, $t \geq 0$, we thus obtain
the {\it Mehler semigroup}, also called the {\it Ornstein--Uhlenbeck semigroup}.
It can also be described by the {\it Mehler formula} \cite[Section~1.4]{Bo98}
\begin{gather} \label{eq:mehler}
 (\Gamma(c)f)(x) = \int_{\R^d} f(cx + \sqrt{1 - c^2}y)\, d\gamma(y).
\end{gather}
These operators form a hermitian strongly continuous contraction semigroup for which
$\gamma$ is the unique invariant probability measure.
In view of \cite[Section~2.9]{Bo98}, \eqref{eq:mehler} still holds for
inf\/inite-dimensional spaces.

(d) For $\cH = \R^d$, it is instructive to connect the gaussian picture to the Lebesgue
picture. To this end, we note that
\begin{gather*} \Phi \colon \ L^2\big(\R^d\big) \to L^2\big(\R^d,\gamma\big), \qquad
\Phi(f)(x) = (2\pi)^{d/4} e^{\|x\|^2/4} f(x) \end{gather*}
is a unitary isomorphism. Conjugating with $\Phi$, we therefore obtain a
$*$-representation $\pi$ of the real Olshanski semigroup $S$
(Example~\ref{ex:osci}) on $L^2(\R^d,\gamma)$. This transformation has no
ef\/fect on the multiplication operators $M_A$ because they commute with $\Phi$
but it transforms the convolution operators $C_A$ into more complicated Markov
operators. For instance the Laplace operator $\Delta = \sum_j \partial_j^2$ transforms
into
\begin{gather*} \Delta' := \Phi \Delta \Phi^{-1}
= \Delta + \frac{\|x\|^2}{4} - \left(E + \frac{d}{2} \1\right), \qquad \mbox{where} \quad
E = \sum_j x_j \partial_j \end{gather*}
is the Euler operator. The unitary representation of $H = \GL_d(\R)$ transforms into
\begin{gather*}
T_h' f = e^{\frac{1}{4}(\|x||^2 - \|h^{-1}x\|^2)} (T_h f)(x) = |\det h|^{-d/2} e^{\frac{1}{4}(\|x||^2 - \|h^{-1}x\|^2)} f\big(h^{-1}x\big).
\end{gather*}
As the operators $\Phi$ and $\Phi^{-1}$ are positivity preserving, the semigroup
$S$ also acts by positivity preserving operators on~$L^2(\R^n,\gamma)$. If
$s \in S$ preserves~$1$, then the transformed operator~$s'$
preserves the Gaussian~$e^{\|x\|^2/4}$.

(e) The oscillator Hamiltonian $H_{\rm osc}$ from~\eqref{eq:osc-ham} transforms into
\begin{gather*} -\Delta':= \Phi H_{\rm osc} \Phi^{-1}
= \Phi \circ \left(-\Delta + \frac{\|x\|^2}{4} - \frac{d}{2}\1\right)
\circ \Phi^{-1}
= E - \Delta = \sum_j x_j \partial_j - \partial_j^2,\end{gather*}
which also makes sense in the inf\/inite-dimensional case, where
$\Delta'$ is called {\it Umemura's Lap\-la\-ce--Beltrami operator}
\cite[p.~221]{Hid80}. It is the inf\/initesimal generator of the Mehler
semigroup~$e^{t \Delta'}$.
\end{Example}

\begin{Example}[a f\/inite-dimensional example] \label{ex:findim}
(a) We consider the f\/inite set $M = \{1,\ldots, n\}$, endowed with the counting measure,
so that $L^2(M,\nu) \cong \R^d$. Let $S \subeq \GL_n(\R)$ be the subsemigroup of invertible doubly stochastic
matrices $S = (S_{ij})$, i.e., $\sum_i S_{ij} = \sum_j S_{ij} = 1$ for all~$i$,~$j$.
According to the Birkhof\/f--von Neumann theorem, $S \subeq \conv(H)$, where
$H = S \cap S^{-1} \cong S_n$ is the group of permutation matrices.
Then $S\subeq \GL_n(\R)$ a transposition stable subsemigroup for which the action on $\R^n \cong L^\infty(M)$ by matrix multiplication
def\/ines a standard positive semigroup structure.

(b) A natural generalization of the preceding example is obtained by starting with a probability space $(Q,\Sigma,\nu)$ and the group
$H := \Aut(Q,\Sigma,\nu)$ of measure preserving automorphisms of $(Q,\Sigma)$.
Then the probability measures $\mu$ on $Q$ def\/ining by $P_\mu(f)(x) := \int_H f(h^{-1}x)\, d\mu(h)$ invertible operators on $L^\infty(Q,\Sigma,\nu)$
naturally generalizes the semigroup $S \subeq \GL_d(\R)$ from~(a).
\end{Example}

\subsection{Stochastic processes indexed by Lie groups}
We now introduce the concept of stochastic processes where the more
common index set $\R$, thought of as a~time, is replaced by a~Lie group. The
forward direction is then given by a~semi\-group contained in~$G$.

\begin{Definition} \label{def:g-process}
Let $(Q,\Sigma,\mu)$ be a probability space. A~{\it stochastic process} indexed by a group~$G$ is a~family $(X_g)_{g \in G}$ of measurable
functions $X_g \colon Q \to (B,\fB)$, where $(B,\fB)$ is a~measurable space, called the {\it state space} of the process.

For any such process, we obtain a measurable map $\Phi \colon Q \to B^G$, $\Phi(q)=(X_g(q))_{g \in G}$
with respect to the product $\sigma$-algebra $\fB^G$.
Then $\nu := \Phi_*\mu$ is a measure on $B^G$, called the
{\it distribution of the process $(X_g)_{g \in G}$}.
This measure is uniquely determined by the measures $\nu_\bg$ on $G^n$,
obtained for any tuple $\bg := (g_1, \ldots, g_n) \in G^n$
as the image of $\mu$ under the map
\begin{gather*} X_\bg = (X_{g_1}, \ldots, X_{g_n}) \colon \  Q \to B^n\end{gather*}
(cf.~\cite[Section~1.3]{Hid80}). If $\bg=(g)$ for some $g\in G$ then we write $\nu_g$ for $\nu_{\bg}$.

The process $(X_g)_{g \in G}$ is called {\it stationary} if the corresponding distribution on $B^G$ is invariant under the translations
\begin{gather*} (U_g \nu )_h := \nu_{g^{-1}h} \qquad \mbox{for} \quad g,h \in G.\end{gather*}
If $\tau \in \Aut(G)$ is an automorphism, then we call the process {\it $\tau$-invariant} if the distribution is invariant under
\begin{gather*} (\tau\nu )_h := \nu_{\tau^{-1}(h)} \qquad \mbox{for} \quad h \in G.\end{gather*}
\end{Definition}

The following def\/inition generalizes some concepts from $G = \R$ to general topological
groups (cf.~\cite[Section~2]{Kl77}). We only formulate the def\/inition for real valued functions, but
will use without further comments the complex analogue where in the $L^2$-inner products the second factor has to be conjugated.

\begin{Definition} \label{def:1.13} (a) A {\it linear stochastic process} (indexed by the group $G$)
is a stochastic process $(\phi(v,g))_{v \in V, g \in G}$,
where $V$ is a real linear space, $\phi$ is linear in~$v$
and each $\phi(v,g)$ is a real valued measurable function on a f\/ixed
probability space $(Q,\Sigma,\mu)$.

(b) A linear stochastic process is said to be of {\it second order} if each
$\phi(v,g)$ is square integrable and the stochastic process is {\it full} in
the sense that, up to sets of measure zero,
$\Sigma$ is the smallest $\sigma$-algebra
for which all functions $\phi(v,g)$ are measurable.

(c) That $\phi$ is {\it continuous in the quadratic mean} means that $V$ is a topological
vector space, $\phi $ is of second order and $\phi \colon V \times G \to L^2(Q,\Sigma,\mu)$
is continuous.

(d) Assume that $\phi$ is of second order. Then $\phi$ is {\it (wide sense) stationary} means that the kernel
\begin{gather*} K((v,g), (v',g')) := \E\big(\oline{\phi(v,g)}\phi(v',g')\big) = \int_Q \oline{\phi(v,g)}\phi(v',g')\, d\mu\end{gather*}
on $V \times G$ satisf\/ies the invariance condition
\begin{gather*} K((v,g), (v',g')) = K((v,g_0 g), (v',g_0 g'))\qquad \mbox{for} \quad g_0, g,g'\in G, \quad v,v' \in V.\end{gather*}
\end{Definition}

Let $\phi$ be a stochastic process. To adapt better to the discussion in the previous section, we assume that $\phi (v,g)$ is complex-valued and that $V$ is a complex vector space. The real case can be treated in the same way.
If~$\phi$ is a linear stochastic process of second order, then
the square integrability of the functions $\phi(v,g)$
implies that the covariance kernel
\begin{gather*} K((v,g), (v',g')) := \E\big(\oline{\phi(v,g)}\phi(v',g')\big)
= \int_Q \oline{\phi(v,g)}\phi(v',g')\, d\mu\end{gather*}
is def\/ined. It is a positive def\/inite kernel on $V \times G$, and the map
\begin{gather*} \Phi \colon \  L^2(Q,\Sigma,\mu) \to \cH_K \subeq \C^{V \times G}, \qquad
\Phi(F)(v,g) := \la F, \phi(v,g) \ra \end{gather*}
is a partial isometry of $L^2(Q,\Sigma,\mu)$ onto the reproducing kernel Hilbert space
$\cH_K\subeq \C^{V \times G}$. Its adjoint $\Phi^*$ is an isometric embedding of $\cH_K$
whose range is the closed subspace on $L^2(Q,\Sigma,\mu)$ generated by the
$(\phi(v,g))_{v \in V, g \in G}$. Furthermore, $\phi$ is continuous in the quadratic mean if and only if~$K$ is continuous.

We will from now on assume that $\phi $ is continuous in the quadratic mean. Then
\begin{gather*} (v,v') := \E\big( \phi(v,\1)\oline{\phi(v',\1)}\big)
= \int_Q \phi(v,\1)\oline{\phi(v',\1)}\, d\mu \end{gather*}
def\/ines a continuous positive semidef\/inite form on $V$, and in the following
we may therefore replace $V$ by the complex Hilbert space
$\cF$ obtained by completing the quotient
\begin{gather*} V/\{v \in V \colon (v,v) = 0\} \end{gather*} with respect to the induced norm.
We thus obtain a stochastic process $\phi(w,g)$ on $\cF \times G$.
Now we may consider $\cF$ as a closed subspace of $L^2(Q,\Sigma,\mu)$ and write
$E_0 \colon L^2(Q,\Sigma,\mu) \to \cF$ for the orthogonal projection.

If $K$ is continuous and $\cF$ is a complex Hilbert space,
then $K$ also def\/ines a positive def\/inite $B(\cF)$-valued kernel
\begin{gather*} K^\cF \colon \ G \times G \to B(\cF), \qquad
\big\la K^\cF(g,h)v, w \big\ra = K((v,h), (w,g)), \end{gather*}
so that we can realize the closed subspace of $L^2(Q,\mu)$ generated by the
$\phi(v,g)$ as the reproducing kernel Hilbert space
$\cH_{K^\cF} \subeq \C^\cF$ of $\cF$-valued functions on $G$.
Here $\cF$ is isomorphic to the closed subspace of $L^2(Q,\Sigma,\mu)$ generated by the functions $(\phi(v,\1))_{v \in V}$.

Consider now (d) and assume that $\varphi$ is continuous in the quadratic
mean. This condition ensures that we obtain a unitary representation
$(U_g)_{g \in G}$ on the subspace $\cH_K \subeq L^2(Q,\mu)$ satisfying
\begin{gather*} (U_g \phi)(v,h) = \phi(v,gh), \qquad v \in V, \quad g,h \in G.\end{gather*}
Under assumption (c), we thus obtain a strongly continuous
$B(\cF)$-valued positive def\/inite function $r(g) := E_0 U_g E_0$.
The invariance condition leads to a $\Sesq (V)$-valued positive def\/inite function on~$G$
as discussed in~\cite{NO15b}, but we will not follow up on that in this article.

\begin{Example} \label{ex:2.9} Let $(G,\tau)$ be a symmetric Lie group and
$(X_g)_{g \in G}$ be a stationary, $\tau$-invariant, full stochastic process on
$(Q,\Sigma_Q,\mu_Q)$. Then its distribution $(B^G, \fB^G, \nu)$ satisf\/ies the conditions
(GP1,2,4,5) of a $(G,\tau)$-probability space with respect to the natural actions of $G$ and
$\tau$ on $B^G$, where $\Sigma_0$ is the $\sigma$-algebra generated by
$(X_h)_{h \in H}$, i.e., the smallest subalgebra for which these functions are measurable.
In this context (GP3) is equivalent to the continuity of the unitary representation of
$G$ on $L^2(B^G,\fB^G,\nu)$. We refer to Appendix~\ref{subsec:6.3} for a more detailed discussion of the continuity condition.
\end{Example}

\begin{Example}[cf.\ Proposition~\ref{prop:markov}]
Suppose that $((Q,\Sigma,\mu), \Sigma_0, U,\theta)$ is a
$(G,\tau)$-probability space, so that we obtain on
$\cE := L^2(Q,\Sigma,\mu)$ a unitary representation of $G$.
The subspace $\cE_0 := L^2(Q,\Sigma_0,\mu)$ is cyclic under the algebra
$\cB$ generated by $(U_g)_{g \in G}$ and the multiplication operators
$M_f$, $f \in L^\infty(Q,\Sigma_0,\mu)$. The prescription
\begin{gather*} \phi(v,g) := \pi(g)v, \qquad g \in G, \quad v \in \cE_0 \end{gather*}
def\/ines a wide sense stationary process indexed by $\cE_0 \times G$.

The scalar product on $L^2(Q,\Sigma,\mu)$ is completely determined by
the $n$-point functionals:
\begin{gather*} \la U_{g_1} M_{f_1} \cdots M_{f_{n-1}} U_{g_{n}} M_{f_n}v, w \ra_{\cE_0}, \qquad
v,w \in \cE_0,\quad g_j \in G,\quad f_j \in L^\infty(Q,\Sigma_0,\mu),\end{gather*}
resp.,
\begin{gather*} \la U_{g_1} M_{f_1} \cdots U_{g_{n}} f_n, f_0\ra_{\cE_0}
= \E\big(f_0 \phi(f_1, g_1) \phi(f_2, g_1 g_2) \cdots \phi(f_n, g_1 \cdots g_{n})|
\Sigma_0\big),\end{gather*}
or the operators
\begin{gather*} E_0 U_{g_1} M_{f_1} \cdots M_{f_{n-1}} U_{g_{n}} E_0, \qquad
g_j \in G,\quad f_j \in L^\infty(Q,\Sigma_0,\mu).\end{gather*}

In general $\cE_0$ is not cyclic for $G$, so the process $\phi(v,g)$ provides a means to construct the whole space
$L^2(Q,\Sigma,\mu)$ in the spirit of a GNS construction for $\cB$.
\end{Example}

\begin{Example} In \cite{Si74} free euclidean f\/ields
(scalar of positive mass $m > 0$) are dealt with as gaussian processes
indexed by the Hilbert space $N_m$ def\/ined by the norm
\begin{gather*} \|f\|_{N_m}^2 =2 \int_{\R^{d+1}} \frac{|\hat f(k)|^2}{k^2 + m^2}\, d^2 k
\qquad \mbox{for} \quad f \in \cS\big(\R^{d+1}\big), \end{gather*}
so that the corresponding euclidean Hilbert space is the Fock space $\Gamma(N_m)$ (see Section~\ref{se:SecondQ}).
Here $N_m$ carries a unitary representation of the euclidean group $E(d) = \Mot(\R^d)$, and since the time translation group has continuous spectrum,
the corresponding action on $\Gamma(N_m)$ is ergodic by Theorem~\ref{thm:3.28} below (for time translations, space translations and the full euclidean group).

\cite[Theorem~III.6]{Si74} provides a Feynman--Kac--Nelson formula for the gaussian free f\/ield. In this context, Nelson's main achievement in \cite{Nel64} was that he obtained a manifestly euclidean invariant path integral.
\end{Example}

\subsection{Associated positive semigroup structures and reconstruction}

Our goal in this section is to prove the \textit{reconstruction theorem} (Theorem~\ref{thm:1.8G}).
This is a~central result which asserts that positive semigroup structures all come from $(G,S,\tau)$-measure spaces.

We start with the following lemma which is an adaption of \cite[Lemma~2.2]{Kl78} and \cite[Corollary~1.5]{KL75} to $(G,S,\tau)$.

\begin{Lemma} \label{lem:1.6G}
Let $((Q,\Sigma,\mu), \Sigma_0, U,\theta)$ be a $(G,S,\tau)$-measure space, $\cA := L^\infty(Q,\Sigma_0,\mu)$,
$\cE := L^2(Q,\Sigma,\mu)$, $\cE_0 := L^2(Q,\Sigma_0,\mu)$, and $q \colon \cE_+ := L^2(Q,\Sigma_+, \mu) \to \hat\cE$ be the canonical map. Then the following assertions hold:
\begin{itemize}\itemsep=0pt
\item[\rm(a)] For $f \in \cA$, let $M_f$ denote the corresponding multiplication operator on $\cE$. Then there exists a bounded operator $\hat M_f \in B(\hat\cE)$ with $q \circ M_f\res_{\cE_+} = \hat M_f \circ q$ and $\|\hat M_f \| = \|f\|_\infty$.
\item[\rm(b)] $\pi(f) := \hat M_f$ is a faithful weakly continuous representation
of the commutative von Neumann algebra $\cA$ on $\hat\cE$.
\item[\rm(c)] In the Markov case we identify $\hat\cE$ with
$\cE_0$ and $q$ with $E_0$ $($Proposition~{\rm \ref{prop:markov})}.
For $g_1 \prec_S g_2 \prec_S \cdots \prec_S g_n$ in $G$,
non-negative functions $f_1, \ldots, f_n \in \cA$ and $f_{g_j} := U_{g_j} f_j$, we
have
\begin{gather*} \int_Q f_{g_1} \cdots f_{g_n}\, d\mu
= \int_Q \hat M_{f_1} \hat U_{g_1^{-1} g_2} \cdots \hat M_{f_{n-1}}
\hat U_{g_{n-1}^{-1} g_n} \hat M_{f_n} 1\, d\mu. \end{gather*}
\end{itemize}
If, in addition, $\mu$ is finite and positive,
then $\Omega := \mu(Q)^{-1/2} q(1)$ satisfies:
\begin{itemize}\itemsep=0pt
\item[\rm(d)] For $g_1 \prec_S g_2 \prec_S \cdots \prec_S g_n$ in $G$,
$f_1, \ldots, f_n \in \cA$ and $f_{g_j} := U_{g_j} f_j$, we have
\begin{gather*} \int_Q f_{g_1} \cdots f_{g_n}\, d\mu
= \big\la \hat M_{f_1} \hat U_{g_1^{-1} g_2} \cdots \hat M_{f_{n-1}}
\hat U_{g_{n-1}^{-1} g_n} \hat M_{f_n} \Omega, \Omega \big\ra. \end{gather*}
\item[\rm(e)] $\Omega$ is a separating vector for $\cA$
and $\hat U_s \Omega = \Omega$ for every $s \in S$.
\item[\rm(f)] $\Omega$ is cyclic for the algebra $\cB$ generated by $\cA$ and
$(\hat U_s)_{s \in S}$.
\end{itemize}
\end{Lemma}

\begin{proof} (a, b) Since $\theta$ commutes with $E_0$, it preserves the subspace
$\cE_0 := L^2(Q,\Sigma_0, \mu)$. Ref\/lec\-tion positivity and $\cE_0 \subeq \cE_+$
now imply that $\theta(f) = f$ for $f \in \cE_0$.
The preceding argument implies that~$M_f$ commutes with~$\theta$.
As $M_f$ and its adjoint $M_f^* = M_{\oline f}$ preserve $\cE_\pm$, these operators also preserve
$\cN := \cE_+ \cap \theta(\cE_-)^\bot$, so that we obtain a well-def\/ined operator $\hat M_f \colon q(\cE_+) \to \hat \cE.$
For $v \in \cE_+$, we then have
\begin{gather*} \big\| \hat M_f \hat v \big\|^2
= \la \theta M_f v, M_f v \ra = \la \theta M_{|f|^2} v, v \ra \leq \|f\|_\infty^2 \|v\|^2.\end{gather*}
Applying this estimate to the functions $(f^n)_{n \in \N_0}$, we get
$\|\hat M_f\| \leq \|f\|_\infty$ from \cite[Lemma~II.3.8]{Ne00}.
Since $q$ is isometric on $\cE_0$ and $\|M_f\res_{\cE_0}\| = \|f\|_\infty$,
we actually obtain the equality $\|\hat M_f\| = \|f\|_\infty$. This proves~(a) and (b).

(c) (\cite[Corollary~1.5]{KL75}) First we observe that the operator
\begin{gather*} 
M_{f_{g_1}} \cdots M_{f_{g_n}} U_{g_n}
= M_{U_{g_1} f_1} \cdots M_{U_{g_n}f_n} U_{g_n}
= U_{g_1} M_{f_1} U_{g_1^{-1}g_2} M_{f_2} \cdots M_{f_{n-1}} U_{g_{n-1}^{-1}g_n} M_{f_n}
\end{gather*}
on $\cE_+$ is intertwined by $q$ with $\hat U_{g_1} \hat M_{f_1} \hat U_{g_1^{-1}g_2} \cdots \hat U_{g_{n-1}^{-1}g_n} \hat M_{f_n}.$
Applying this operator to the constant function $1$ and integrating
yields the assertion.

{}From now on we assume that $0 < \mu(Q) < \infty$.

(d) With the same argument as under (c), we derive
\begin{align*}
\int_Q f_{g_1} \cdots f_{g_n}\, d\mu
&= \la M_{f_{g_1}} \cdots M_{f_{g_n}}1, 1 \ra \\
&= \la \theta M_{f_{g_1}} \cdots M_{f_{g_n}} U_{g_n} 1, 1 \ra
= \big\la \hat M_{f_1} \hat U_{g_1^{-1} g_2} \cdots \hat M_{f_{n-1}}
\hat U_{g_{n-1}^{-1} g_n} \hat M_{f_n} \Omega, \Omega \big\ra.
\end{align*}

(e) The relation $\hat U_s \Omega = \Omega$
is a trivial consequence of $U_s 1 = 1$ for $s \in S$.
Since $q\res_{\cE_0}$ is isometric and intertwines the representation
of $\cA$ on $\cE_0$ with the representation of $\cA$ on $\hat\cE$,
the $\cA$-separating vector $1 \in \cE_0$ is mapped by $q$ to an
$\cA$-separating vector.

(f) Since $\cE_0$ is $U_S$-cyclic in $\cE_+$, this follows from the fact that
$1$ is $\cA$-cyclic in $L^2(Q,\Sigma_0,\mu)$.
\end{proof}

\begin{Definition} The preceding lemma shows that, if
$((Q,\Sigma,\mu), \Sigma_0, U,\theta)$ is a f\/inite
$(G,S,\tau)$-measure space,
then $(\hat\cE, \hat U, \cA, q(1))$ is a positive semigroup
structure for
\begin{gather*} \cA = \big\{ \hat M_f \colon f \in L^\infty(Q,\Sigma_0,\mu)\big\}.\end{gather*}
We call it the {\it associated positive semigroup structure.}
\end{Definition}

We now come to our version of Klein's reconstruction theorem.
Note that every discrete group is in particular a $0$-dimensional Lie group,
so that the following theorem applies in particular to discrete groups.

\begin{Theorem}[Reconstruction Theorem] \label{thm:1.8G}
Let $(G,\tau)$ be a~symmetric Lie group and $S \subeq G$ be a~$\sharp$-invariant
subsemigroup satisfying $G= S \cup S^{-1}$. Then every positive semigroup structure for $(G,S,\tau)$ is associated to some
$(G,S,\tau)$-probability space $((Q,\Sigma,\mu), \Sigma_0, U,\theta)$.
\end{Theorem}

\begin{proof} We follow the proof for $(G,S,\tau) = (\R,\R_+, -\id)$ in \cite[Theorem~2.4]{Kl78}.
In view of the Gelfand representation theorem, there exists a compact
space $Q_0$ with $\cA \cong C(Q_0)$ as $C^*$-algebras. Then
$Q := Q_0^G$, the space of all functions $\omega \colon G \to Q_0$, is compact with respect to the
product topology. The group $G$ acts on $Q$ by
\begin{gather*} (U_g \omega)(x) := \omega\big(g^{-1}x\big) \end{gather*}
and, accordingly, on functions on $Q$ by $(U_g F)(\omega) = F(U_{g^{-1}}\omega)$.
We further obtain an involution $(\theta\omega)(g) := \omega(\tau(g))$
satisfying
\begin{gather*} \theta U_g \theta = U_{\tau(g)} \qquad \mbox{for} \quad g \in G.\end{gather*}

Let $\Sigma_0$ be the smallest $\sigma$-algebra on $Q$ for which $X_\1(\omega) := \omega(\1)$ is measurable,
$\Sigma_g := U_g \Sigma_\1$, and let $\Sigma$ be the $\sigma$-algebra generated by the $\Sigma_g$, $g \in G$.
To construct a measure $\mu$ on $\Sigma$, it suf\/f\/ices to construct a Radon measure on $Q$, i.e., a positive functional on
the $C^*$-algebra $C(Q)$. By continuity, it actually suf\/f\/ices to do this on a dense unital subalgebra. Such a subalgebra is generated by the $G$-translates of functions in $\Xi := C(Q_0, \{0,1\})$, which is a generating subset of~$\cA$.
Now let $\cR \subeq C(Q)$ be the subspace spanned by functions of the form
\begin{gather*} F(\omega) := \chi_1(\omega(g_1)) \cdots \chi_n(\omega(g_n)), \qquad
n \in \N_0, \quad \chi_j \in \Xi, \quad g_j \in G, \quad \mbox{i.e.},
\quad
F = \prod_{j = 1}^n U_{g_j} \chi_j, \end{gather*}
when we identify $C(Q_0)$ with the subalgebra $X_\1^*C(Q_0)$ of $C(Q)$. Clearly, $\cR$ is a~subalgebra of $C(Q)$, and the Stone--Weierstra\ss{} theorem implies that it is dense.

Since $G = S \cup S^{-1}$, every f\/inite subset of $G$ is of the form $\{g_1,\ldots, g_n\}$ with
$g_1 \prec_S \cdots \prec_S g_n$, i.e., $s_j := g_j^{-1} g_{j+1} \in S$ for $j = 1,\ldots, n-1$.
Next we show that that, for $(A_j, g_j) \in \cA \times G$ with $g_1 \prec_S \cdots \prec_S g_n$, the operator
\begin{gather}
 \label{eq:prod}
 A_1 P_{s_1} A_2 \cdots P_{s_{n-1}} A_n
\end{gather}
does not depend on the way we enumerate the $g_j$.
Suppose that we also have $g_{j+1} \prec_S g_j$ for some~$j$.
Then $h := s_j = g_j^{-1} g_{j+1} \in H(S)=S\cap S^{-1}$.
We put $\tilde g_i := g_i$ for $i \not=j,j+1$ and
set $\tilde g_j := g_{j+1}$, $\tilde g_{j+1} = g_j$.
Accordingly, we put $\tilde A_j := A_{j+1}$ and
$\tilde A_{j+1} := A_j$. This leads to
\begin{gather*} \tilde s_i =
\begin{cases}
 s_i & \text{for } i \not= j-1,j,j+1, \\
 s_{j-1} s_j = s_{j-1}h & \text{for } i = j-1, \\
 h^{-1} & \text{for } i = j, \\
 s_j s_{j+1} = h s_{j+1} & \text{for } i = j+1.
\end{cases}\end{gather*}
We thus obtain
\begin{align*}
P_{\tilde s_{j-1}} \tilde A_j P_{\tilde s_{j}} \tilde A_{j+1} P_{\tilde s_{j+1}}
&= P_{s_{j-1}h} A_{j+1} P_h^{-1} A_j P_{hs_{j+1}}
= P_{s_{j-1}} \big(P_h A_{j+1} P_h^{-1}\big) A_j P_h P_{s_{j+1}} \\
&= P_{s_{j-1}} A_j \big(P_h A_{j+1} P_h^{-1}\big) P_h P_{s_{j+1}}
= P_{s_{j-1}} A_j P_{s_j} A_{j+1} P_{s_{j+1}},
\end{align*}
where we use that $P_h A_{j+1} P_h^{-1} \in \cA$, $H(S)$ normalizes $\cA$, and
the commutativity of $\cA$.
The above argument shows the asserted independence
of \eqref{eq:prod} because any increasing rearrangement can be
obtained by successive transpositions.

For $n \in \N$ and $s_1, \ldots, s_{n-1} \in S$, the map
\begin{gather*} \cA^n \cong C(Q_0)^n \to \R, \qquad
(A_1, \ldots, A_n ) \mapsto
\la A_1 P_{s_1} A_2 \cdots P_{s_{n-1}} A_n \Omega, \Omega \ra \end{gather*}
is $n$-linear, so that there exists a well-def\/ined linear functional
$\rho \colon \cR \to \R$ satisfying
\begin{gather}
 \label{eq:ansatz}
 \rho\left(\prod_{j = 1}^n U_{g_j} f_j\right)
= \big\la f_1 P_{g_1^{-1}g_2} f_2 \cdots P_{g_{n-1}^{-1}g_n} f_n \Omega,\Omega\big\ra
\end{gather}
for $f_1, \ldots, f_n \in \Spann \Xi$, $g_1 \prec_S \cdots \prec_S g_n$, $n \in \N_0$.

To see that $\rho$ is positive, we observe that every non-negative
function $F \in \cR$ can be written in the form
\begin{gather*} F = \sum_k c_k F_k \qquad \mbox{with}
\quad F_k := \prod_{j = 1}^n U_{g_{j}} \chi_{j,k},\qquad \mbox{where} \quad
F_k F_{k'} = 0\quad \mbox{for} \quad k \not=k'.\end{gather*} Then
$c_k \geq 0$ for all $k$, so that
\begin{gather*} \rho(F) = \sum_k c_k \rho\left(\prod_{j = 1}^n U_{g_{j,k}} \chi_{j,k}\right) \geq 0\end{gather*}
by \eqref{eq:ansatz}.
This proves that $\rho \geq 0$, and thus, by continuity, it extends to a
positive functional on~$C(Q)$, which def\/ines a Radon measure $\mu$ on $Q$.

Next we show that $((Q,\Sigma,\mu), \Sigma_0, P,\theta)$ is a
$(G,S,\tau)$-measure space. The $G$-invariance of the right hand side of~\eqref{eq:ansatz} implies that
$G$ acts on $Q$ by measure preserving transformations.
To see that $\mu$ is also $\theta$-invariant, we f\/irst note that
$g \prec_S h$ implies $\tau(h) \prec_S \tau(g)$ because $\tau(S) = S^{-1}$.
We thus obtain
\begin{gather*} \theta\left(\prod_{j = 1}^n U_{g_j} f_j\right)
= \prod_{j = 0}^{n-1} U_{\tau(g_{n-j})} f_{n-j}.\end{gather*}
Hence the $\theta$-invariance of $\mu$ follows from
\begin{gather*}
\begin{split}
& \big\la f_1 P_{g_1^{-1}g_2} f_2 \cdots P_{g_{n-1}^{-1}g_n} f_n \Omega, \Omega \big\ra
=\big\la \Omega, f_n P_{\tau(g_n)^{-1}\tau(g_{n-1})} \cdots f_2 P_{\tau(g_2)^{-1} \tau(g_1)}f_1\Omega \big\ra\\
&\hphantom{\big\la f_1 P_{g_1^{-1}g_2} f_2 \cdots P_{g_{n-1}^{-1}g_n} f_n \Omega, \Omega \big\ra}{} =\big\la f_n P_{\tau(g_n)^{-1}\tau(g_{n-1})} \cdots f_2 P_{\tau(g_2)^{-1} \tau(g_1)}f_1\Omega, \Omega \big\ra.
\end{split}
\end{gather*}

It remains to verify $\la \theta F,F \ra\geq 0$ for $F \in L^2(Q,\Sigma_+, \mu)$.
The functions of the form
\begin{gather*} f_{g_1} \cdots f_{g_n}
= U_{g_1} M_{f_1} U_{g_1^{-1}g_2} \cdots U_{g_{n-1}^{-1}g_n} M_{f_n} U_{g_n}^{-1} 1, \qquad
f_j \in \cA, \1 \prec_S g_1 \prec_S g_2 \prec_S \cdots \prec_S g_n, \end{gather*}
form a total subset of $\cE_+$. On pairs of such functions we have
\begin{gather*}
\big\la \theta f_{g_1} \cdots f_{g_n}, f_{h_1}' \cdots f_{h_m}' \big\ra \\
\quad{}= \big\la \theta
U_{g_1} M_{f_1} U_{g_1^{-1}g_2} \cdots U_{g_{n-1}^{-1}g_n} f_n,
U_{h_1} M_{f_1'} U_{h_1^{-1}h_2} \cdots U_{h_{m-1}^{-1}h_m} f_m' \big\ra \\
\quad{}= \big\la U_{\tau(g_1)} M_{f_1} U_{\tau(g_1)^{-1}\tau(g_2)} \cdots
U_{\tau(g_{n-1})^{-1}\tau(g_n)} f_n,
U_{h_1} M_{f_1'} U_{h_1^{-1}h_2} \cdots U_{h_{m-1}^{-1}h_m} f_m' \big\ra \\
\quad{} = \big\la 1, M_{f_n} U_{\tau(g_n)^{-1}\tau(g_{n-1})} \cdots
 U_{\tau(g_2)^{-1} \tau(g_1)} M_{f_1} U_{\tau(g_1)^{-1}} U_{h_1}
M_{f_1'} U_{h_1^{-1}h_2} \cdots U_{h_{m-1}^{-1}h_m} f_m' \big\ra \\
\quad{} = \big\la \Omega, \hat M_{f_n} P_{\tau(g_n)^{-1}\tau(g_{n-1})} \cdots
 P_{\tau(g_2)^{-1} \tau(g_1)} \hat M_{f_1} P_{\tau(g_1)^{-1}} P_{h_1}
\hat M_{f_1'} P_{h_1^{-1}h_2} \cdots P_{h_{m-1}^{-1}h_m} \hat M_{f_m'}\Omega \big\ra \\
\quad{}= \big\la P_{g_1}\hat M_{f_1} P_{g_1^{-1}g_2} \cdots P_{g_{n-1}^{-1}g_n} \hat M_{f_n} \Omega,
P_{h_1} \hat M_{f_1'} P_{h_1^{-1}h_2} \cdots P_{h_{m-1}^{-1}h_m} \hat M_{f_m'}\Omega \big\ra.
\end{gather*}
This implies that
\begin{gather*} \Big\la \theta \sum_{j} c_j f_{g_{1,j}} \cdots f_{g_{n,j}},
\sum_{k} c_k f_{g_{1,k}} \cdots f_{g_{n,k}} \Big\ra\\
\qquad{}
= \Big\|\sum_j c_j P_{g_{1,j}} \hat M_{f_{1,j}} P_{g_{1,j}^{-1}g_{2,j}}
\cdots P_{g_{n-1,j}^{-1}g_{n,j}} \hat M_{f_{n,j}} \Omega \Big\|^2 \geq 0.\end{gather*}
We conclude that $\cE_+$ is $\theta$-positive and that the
canonical map $q \colon \cE_+ \to \hat\cE$ is equivalent to the map
\begin{gather*} q \colon \ \cE_+ \to \cH, \qquad
 f_{g_{1}} \cdots f_{g_{n}} \mapsto P_{g_1} \hat M_{f_{1}} P_{g_{1}^{-1}g_{2}}
\cdots P_{g_{n-1}^{-1}g_{n}} \hat M_{f_{n}} \Omega.\end{gather*}
This map satisf\/ies $q \circ M_f = \hat M_f \circ q$ and
$q \circ U_s = P_s \circ q$ for $s \in S$. Therefore
$(\cH,P,\cA,\Omega)$ is equivalent to the associated positive
semigroup structure.
\end{proof}

\begin{Remark} Theorem~\ref{thm:1.8G}
implies in particular that, for every positive
semigroup structure, $\Omega$ is separating for $\cA$ \cite[Remark~2.5]{Kl78}.
\end{Remark}

\begin{Remark} We have seen in the proof of Theorem~\ref{thm:1.8G}
that it is sometimes more convenient to deal with
the unital $W^*$-algebra $\cA = L^\infty(Q,\Sigma_0, \nu)$ as a $C^*$-algebra
and write it as $C(Q_0)$ for a compact (Hausdorf\/f) space $Q_0$.
Then the Riesz representation theorem asserts that
positive functionals on $C(Q_0)$ correspond to regular Borel measures $\mu$ on
$Q_0$. On the other hand, the restriction of any such measure to the
$\sigma$-algebra of Baire sets, i.e., the smallest $\sigma$-algebra for which
all continuous functions are measurable, already determines the corresponding
linear functional on $C(Q_0)$ by integration. Therefore the positive functionals
on $C(Q_0)$ are in one-to-one correspondence with f\/inite
positive Baire measures on $Q_0$.

Since the set $\Sigma_0/J_\nu$ of equivalence classes modulo $\nu$-null sets
corresponds to the idempotents in $\cA$, resp., $C(Q_0)$, Baire measures
on $Q_0$ correspond to measures on $(Q,\Sigma_0)$ that are absolutely continuous
with respect to $\nu$.
\end{Remark}

A slight variation of the argument in the proof of
Theorem~\ref{thm:1.8G} leads to the
following ge\-ne\-ra\-li\-zation of
the Klein--Landau reconstruction theorem \cite[Theorem~1.8]{KL75}
which deals with standard positive semigroup structures.
For the proof of this theorem,
we need a suitable version of \cite[Lemma~1.9]{KL75}
to deal with the case where the measure $\nu$ is inf\/inite.

\begin{Lemma} \label{lem:2.16} Let $(G,S,\tau)$ be a symmetric semigroup with
$G = S \cup S^{-1}$. Let $Q_0$ be a compact Hausdorff space and
$\nu_0$ a $\sigma$-finite measure on the Baire sets in $Q_0$. Suppose that
$(P_s)_{s \in S}$ is a~positivity preserving semigroup acting on $C(Q_0)$ leaving the function
$1$ invariant. Then there exists a unique $\sigma$-finite Baire measure $\mu$ on the compact
Hausdorff space $Q = Q_0^G$, such that, for
\begin{gather*} \tau(s_m) \prec_S \cdots \prec_S \tau(s_1) \prec_S \1 \prec_S t_1 \prec_S
\cdots \prec_S t_n \end{gather*}
and non-negative $g_1, \ldots, g_m, f_1, \ldots, f_n \in C(Q_0)$, we have
\begin{gather*}
 \int_Q g_m(\phi (\tau(s_m))) \cdots g_1(\phi (\tau(s_1))
f_1(\phi (t_1)) \cdots f_n(\phi (t_n)) \, d\mu(\phi ) \\
= \int_{Q_0}
\big(P_{s_1} M_{g_1} P_{s_1^{-1} s_2} M_{g_2} \cdots P_{s_{m-1}^{-1} s_m} M_{g_m}1\big)
\big(P_{t_1} M_{f_1} P_{t_1^{-1} t_2} M_{f_2} \cdots P_{t_{n-1}^{-1} t_n} M_{f_n}1\big)\, d\nu(q_0),
\end{gather*}
where, as before, $M_f$ denotes multiplication by~$f$.
\end{Lemma}

\begin{proof} First we assume that the measure $\nu$ is f\/inite.
Let $\cR \subeq C(Q)$ be as in the proof of Theorem~\ref{thm:1.8G}.
We then obtain a positive functional $\rho$ on $\cR$, which is uniquely determined by
the requirement that, for
\begin{gather*} F(q) = g_m(\phi (\tau(s_m))) \cdots g_1(\phi (\tau(s_1)))
f_1(\phi (t_1)) \cdots f_n(\phi (t_n)),\end{gather*}
we have
\begin{gather*}\rho(F) = \int_{Q_0}
\big(P_{s_1} M_{g_1} P_{s_1^{-1} s_2} M_{g_2}{\cdots} P_{s_{m-1}^{-1} s_m} M_{g_m}1\big)
\big(P_{t_1} M_{f_1} P_{t_1^{-1} t_2} M_{f_2} {\cdots} P_{t_{n-1}^{-1} t_n} M_{f_n}1\big)\, d\nu(q_0)\end{gather*}
(see the proof of \cite[Lemma~1.9]{KL75} for details). The Riesz--Markov theorem now leads to a
unique f\/inite Baire measure $\mu$ on $Q$ with $\int_Q F\, d\mu = \rho(F)$
for $F \in \cR$.

If $\nu$ is not f\/inite, we write it as $\nu = \sum_j \nu_j$, where the measures
$\nu_j$ are f\/inite. If $\mu_j$ denotes the measure on $Q$ corresponding to $\nu_j$,
the measure $\mu := \sum_j \mu_j$ satisf\/ies all our requirements.
\end{proof}

\begin{Theorem}[Reconstruction Theorem -- Standard Case] \label{thm:1.8G2}
If $(G,S,\tau)$ is a symmetric semigroup with $G = S \cup S^{-1}$, then
every standard positive semigroup structure for $(G,S,\tau)$ is associated to some
$(G,S,\tau)$-measure space $((Q,\Sigma,\mu), \Sigma_0, U,\theta)$ which is unique up to
$G$-equivariant isomorphism of measure spaces.
\end{Theorem}

\begin{proof} First we show uniqueness.
Let $((Q',\Sigma',\mu'), \Sigma_0', U',\theta')$ be a Markov
$(G,S,\tau)$-measure space for which $(L^2(M,\nu), P, L^\infty(M,\nu))$ is the associated standard positive semigroup
structure.

Lemma~\ref{lem:1.6G}(c) implies that the integrals
of products of the form $f_{g_1} \cdots f_{g_n}$,
$g_1 \prec_S \cdots \prec_S g_n$ are uniquely determined by the relation
\begin{gather*}
\int_Q f_{g_1} \cdots f_{g_n}\, d\mu'
= \int_Q M_{f_1} P_{g_1^{-1} g_2} \cdots M_{f_{n-1}} P_{g_{n-1}^{-1} g_n} M_{f_n} 1\, d\nu.
\end{gather*}
Since such products for $f_j \in L^2(M,\nu) \cap L^\infty(M,\nu)$
form a total subset of $L^2(Q',\Sigma',\mu')$, any isomorphism
$L^2(Q',\Sigma_0',\mu') \to L^2(M,\nu) \cong L^2(Q, \Sigma_0, \mu)$
extends to a $G$-equivariant isomorphism
$L^2(Q',\Sigma',\mu') \to L^2(Q, \Sigma, \mu)$.

For the existence of the measure $\mu$, we now use Lemma~\ref{lem:2.16}.
Here we write $Q_0$ for the spectrum of the $C^*$-algebra
$\cA := L^\infty(M,\nu)$, so that $\cA \cong C(Q_0)$ and $\nu$ def\/ines a Baire measure
on $Q_0$. We may then identify $(M,\fS, \nu)$ with $(Q_0, \Sigma_0, \nu)$,
where $\Sigma_0$ is the $\sigma$-algebra of Baire subsets of $Q_0$.

Now let $\mu$ be the Baire measure on $Q := Q_0^G$ from Lemma~\ref{lem:2.16}.
Then the relation $P_s^* = P_{s^\sharp}$, as operators on $L^2(M,\nu)$, leads to
\begin{gather*}
\int_Q f_{g_1} \cdots f_{g_n}\, d\mu
= \int_Q M_{f_1} P_{g_1^{-1} g_2} \cdots M_{f_{n-1}} P_{g_{n-1}^{-1} g_n} M_{f_n} 1\, d\nu
\end{gather*}
for $0 \leq f_1, \ldots, f_n \in C(Q_0)$ and $g_1 \prec_S \cdots \prec_S g_n$.
The same calculations as in the proof of Theorem~\ref{thm:1.8G} now show that $\mu$ is invariant
under $G$ and $\theta$.
\end{proof}

\begin{Remark} \label{rem:mark-meas}
In the standard case, it is interesting to take a closer look at the structure
of the measure space $(Q,\Sigma,\mu)$. Then $\cH \cong L^2(M,\fS, \nu)$, and for $L^\infty(M,\fS,\nu) \cong C(Q_0)$
we obtain on the compact space $Q := Q_0^G$ a~measure which is determined by the corresponding cylinder measures:
\begin{gather}
 \label{eq:5}
\int_Q f_{g_1} \cdots f_{g_n}\, d\mu
= \int_M M_{f_1} P_{g_1^{-1} g_2}
\cdots M_{f_{n-1}} P_{g_{n-1}^{-1} g_n} M_{f_n} 1\, d\nu
\end{gather}
for $0 \leq f_1, \ldots, f_n \in C(Q_0)$ and $g_1 \prec_S \cdots \prec_S g_n$.

For every $g \in G$, we then have a continuous projection
$q_g \colon Q \to Q_0$, so that we obtain a~homomorphism
$q_g^* \colon L^\infty(M,\fS, \nu) \cong C(Q_0) \to C(Q)$ of $C^*$-algebras.
This leads to a morphism $\fS/J_\nu \to \Sigma/J_\mu$ of Boolean
$\sigma$-algebras. At this point it is a natural question whether this
morphism can be realized by a measurable function
$f_g \colon Q \to M$. In \cite[Theorem~1.6]{Va84} this is answered in the
af\/f\/irmative for $M = \R^n$, but this implies the corresponding
result for completely metrizable separable spaces, because these spaces
can be realized as closed subsets of some $\R^\N$.
Once we have the measurable maps $f_g \colon Q \to M$, they combine to a
stochastic process $(f_g)_{g \in G}$ with state space $(M,\fS)$,
so that we obtain an isomorphism of measure spaces
$f \colon (Q,\Sigma, \mu) \to (M^G, \fS^G, \mu')$.

If $(M,\fS, \nu)$ is a polish space, i.e., $M$ carries a topology for which it is
completely metrizable and separable and $\fS$ is the $\sigma$-algebra of Borel sets,
then \cite[Corollary~35.4]{Ba96} directly
implies the existence of a Borel measure $\mu'$ on
the measurable space $(M^G, \fS^M)$ with the correct projections onto f\/inite
products:
\begin{gather*} \int_{M^G} f_1(\omega(g_1)) \cdots f_n(\omega(g_n)) \, d\mu'(\omega)
= \int_Q M_{f_1} P_{g_1^{-1} g_2}
\cdots M_{f_{n-1}} P_{g_{n-1}^{-1} g_n} M_{f_n} 1\, d\nu\end{gather*}
for $0 \leq f_1, \ldots, f_n \in C(Q_0)$ and $g_1 \prec_S \cdots \prec_S g_n$.
Therefore we obtain a realization of our $(G,S,\tau)$-measure space
on $(M^G, \fS^G, \mu')$.
\end{Remark}

\begin{Remark} \label{rem:2.21}
(a) Consider the case $(G,S,\tau) = (\R,\R_+, -\id_\R)$ and assume that
the opera\-tors~$P_t f$ are obtained from
Markov kernels $P_t \colon M \times \fS \to [0,\infty]$ in the sense that
\begin{gather*} (P_t f)(x) = \int_M P_t(x,dy) f(y) \end{gather*}
(cf.\ Appendix~\ref{subsec:6.4}).
In view of Remark~\ref{rem:mark},
the measures $P_\bt^\nu$ can be written as
\begin{gather*} P_\bt^\nu = \nu P_\bt, \qquad \mbox{i.e.,} \qquad
P_\bt^\nu(A) = \int_M \nu(dx) P_\bt(x,A)\end{gather*}
for the kernel
\begin{gather*} P_\bt(x_0,B) = \int_{M^{n}} \chi_B(x_1, \ldots, x_n)
P_{t_1}(x_0, dx_1) P_{t_2- t_1}(x_1, dx_2) \cdots
P_{t_n- t_{n-1}}(x_{n-1}, dx_n) \end{gather*}
on $M \times \fS^n$.
In particular, we have for each $x \in M$ a probability measure $P_x$ on
\begin{gather*} \big\{ \omega \in M^{\R_+} \colon \omega(0) = x\big\}, \end{gather*}
def\/ining a kernel $P \colon M \times \fS^{\R_+} \to [0,\infty]$,
such that
\begin{gather*} P^\nu = \nu P, \qquad \mbox{i.e.,} \qquad P^\nu(A) = \int_M \nu(dx) P(x,A).\end{gather*}
For any $t \geq 0$, we then have
\begin{gather*} P_t(x_0, B)
= \int_{M^{\R_+}} \chi_B(\omega(t))\, P(x_0, d\omega)
= P(x_0, \{ \omega(t) \in B\}),\end{gather*}
which leads to
\begin{gather}
 \label{eq:ptf}
 (P_t f)(x)
= \int_M P_t(x, dy) f(y)
= \int_{M^{\R_+}} P(x, d\omega)f(\omega(t)) .
\end{gather}
This is an abstract version of the Feynman--Kac--Nelson formula
that expresses the value of $P_t f$ in $x \in M$ as an integral over all
paths $[0,t] \to M$ starting in $x$ with respect to the probability measure~$P_x$.

We also note that, for $t > 0$, we have the relation
\begin{gather*} \int_{M^\R} f(\omega(t))\, d P^\nu(\omega) = \int_Q f\, d\nu \quad \mbox{for} \quad t \in \R, \end{gather*}
and, for $t < s$,
\begin{gather*} \int_M \int_M f_1(x) \nu(dx) f_2(y) P_{s-t}(x, dy)
= \int_M f_1(x) (P_{s-t} f_2)(x)\, d\nu(x)\\
\hphantom{\int_M \int_M f_1(x) \nu(dx) f_2(y) P_{s-t}(x, dy)}{}
= \int_{M^\R} f_1(\omega(t))f_2(\omega(s))\, d P^\nu(\omega). \end{gather*}

(b) In the special case where $M = G$ is a topological group and
$P_t f = f * \mu_t$, we have
\begin{gather*} (P_t f)(x) = \int_G f(xy)\, d\mu_t(y) = \int_G P_t(x,dy) f(y) \qquad \mbox{for} \quad
P_t(x,A) = \mu_t(x^{-1}A).\end{gather*}
Therefore the operators $P_t$ are given by Markov kernels. Here, the measurability
of the functions
\begin{gather*} x \mapsto \mu_t(x^{-1}A) = \int_G \chi_A(xy)\, d\mu_t(y) \end{gather*}
follows from the measurability of the function
$(x,y) \mapsto \chi_A(xy)$.

Let $P(G)$ be the path group $G^{\R}$ and let $P_*(G)$ be the subgroup of pinned paths
\begin{gather*}P_*(G)=\{\omega \in P(G)\colon \omega (0)=\1\} .\end{gather*}
We have the relations
\begin{gather*} \int_{P(G)} f_1(\omega(0)) f_2(\omega(t))\, d\mu(\omega) =
 \int_G\int_G f_1(g_1) f_2(g_1 g_2)\, d\nu(g_1) d\mu_t(g_2) \qquad \mbox{for} \quad t > 0\end{gather*}
and
\begin{gather*} \int_{P_*(G)} f(\omega(t))\, dP^\nu(\omega) = \int_G f(g)\, d\mu_t(g).\end{gather*}
This leads for $f \in L^2(G,\nu)$ and $t\geq 0$ to
\begin{gather*} (P_t f)(x)
= (f * \mu_t)(x) = \int_G f(xg)\, d\mu_t(g)
= \int_{G^{\R_+}} f(x\omega(t))\, dP(\omega).\end{gather*}
This is a group version of the Feynman--Kac--Nelson formula
\eqref{eq:ptf} which expresses $(P_t f)(x)$ as an integral over all paths
$\omega \colon [0,t] \to G$ with $\omega(0) = x$.
\end{Remark}

\begin{Remark} The assumption that $G = S \cup S^{-1}$ is very restrictive
(see the discussion in Subsection~\ref{se:2.5}). We mainly use it to ensure
that the order $\prec_S$ on $G$ is total, which is a crucial ingredient in the
construction. However, if $S \subeq G$ is a subsemigroup not satisfying the
totality condition, i.e., $S \cup S^{-1} \not=G$, then one may still consider
subsets $\cC \subeq G$ on which the order $\prec_S$ is total, so-called
$\cC$-chains. Typical examples of such $\prec_S$-chains are curves
$\cC = \gamma(I)$, $I \subeq \R$, with $\gamma(t_2) \in \gamma(t_1)S$ for
$t_1 < t_2$ in $I$. For any such chain $\cC$,
one can still construct a stochastic
process indexed by $\cC$, but we then don't have the symmetries implemented
by a $G$-action on the corresponding probability space.
\end{Remark}

\begin{Remark} \label{rem:AK-exs} The Abel--Klein reconstruction theorem
leads in particular to realizations of dilations of the corresponding
standard positive semigroup structures $(P_t)_{t \geq 0}$ of the following kind:
\begin{itemize}\itemsep=0pt
\item[\rm(a)] Positivity preserving semigroups obtained by integrating a measure
preserving action of $G$ on some measure space $(X,\fS,\nu)$ to operators
\begin{gather*} (P_t f)(x) := \int_G (g.f)(x)\, d\mu_t(g),\end{gather*}
where the measures $\mu_t$ on $G$
are symmetric. If $(X,\fS)$ is polish, the corresponding
measure space can be realized on the path space $X^\R$.
Here the requirement that $G$ acts continuously on $L^2(X,\fS,\nu)$ implies
that $(P_t)_{t \geq 0}$ is strongly continuous on $L^2(X,\fS,\nu)$
(see (b) below for a more general argument)
and therefore continuous in measure (Remark~\ref{rem:2.7}(c)).
\item[\rm(b)] Semigroups of the type $(\Gamma(\pi(\mu_t)))_{t \geq 0}$, where $(\pi, \cH)$ is a continuous
unitary representation of $G$ on $\cH$ and $\Gamma$ is the functor of second quantization
(see Section \ref{se:SecondQ} for more information about second quantization).
Here we use that, for every continuous orthogonal representation
$(\pi, \cH)$ of $G$ and a Borel probability measure $\nu$ on $G$, we can def\/ine
\begin{gather*} \pi(\nu) := \int_G \pi(g)\, d\nu(g) \end{gather*}
as the operator representing the hermitian form
\begin{gather*} q(v,w) := \int_G \la \pi(g)v, w \ra\, d\nu(g) \end{gather*}
which satisf\/ies $|q(v,w)| \leq \|v\| \|w\|$ for $v,w \in \cH$.
For any one-parameter semigroup $(\mu_t)_{t \geq 0}$ of symmetric
probability measures, we thus obtain a
hermitian one-parameter semigroup of contractions $(\pi(\mu_t))_{t \geq 0}$.
The continuity of this one-parameter semigroup follows from the continuity of the functions
\begin{gather*} t \mapsto \la \pi(\mu_t)v,w\ra
= \int_G \la \pi(g)v,w \ra \, d\mu_t(g), \end{gather*}
which is a consequence of the weak convergence $\mu_t \to \delta_\1$ for $t \to 0$
and the boundedness of the matrix coef\/f\/icients of $\pi$.
\end{itemize}
\end{Remark}

\subsection{Klein's characterization of the Markov case}

The following theorem characterizes the positive semigroup structures arising from
Markov path spaces as those for which $\Omega$ is a cyclic vector for $\cA$,
which is considerably stronger than condition (PS4)(b).

\begin{Theorem}[\protect{\cite[Theorem~3.1]{Kl78}}]
 Let $((Q,\Sigma,\mu), \Sigma_0, U,\theta)$ be a $(G,S,\tau)$-probability space and let
$(\hat\cE, \hat U, \cA, \Omega)$ be its associated positive semigroup structure. Then
$((Q,\Sigma,\mu), \Sigma_0, U,\theta)$ is Markov if and only if $\Omega$ is cyclic for~$\cA$.
\end{Theorem}

\begin{proof} The Markov property is equivalent to
$q(\cE_0) = \hat\cE$ (Proposition~\ref{prop:markov}). Since
$\cE_0 = \cA \cdot 1$, this condition is equivalent to $\Omega = q(1)$ being
$\cA$-cyclic in $\hat\cE$.
\end{proof}

\begin{Remark}[cf.\ \protect{\cite[Remark~3.3]{Kl78}}]
Suppose that $(\cH, P, \cA, \Omega)$ is a positive semigroup
structure for which $\Omega$ is $\cA$-cyclic, so that Proposition~\ref{prop:2.19} implies that
$\cH \cong L^2(M)$ for a probability space~$M$. Then condition (PS4)(c)
follows from the seemingly weaker condition
\begin{gather*} A,B \in \cA_+ := \{ C \in \cA \colon 0 \leq C\},\quad  s \in S \qquad \Rarrow \qquad
\la A\Omega, P_s B\Omega \ra \geq 0.\end{gather*}

In fact, in this case $\cA_+\Omega$ is dense in the closed cone
$L^2(Q,\Sigma_0,\mu)_+$ of non-negative functions,
so that $\la A\Omega, P_s B\Omega \ra \geq 0$ for all $A,B \in \cA_+$
implies that the operator $P_s$ on $L^2(Q,\Sigma_0,\mu)$ is positivity preserving.
Then
\begin{gather*} P_{s_1} A_1 \cdots P_{s_n} A_n \Omega \geq 0 \end{gather*}
for $A_j\in \cA_+$ and $s_j \in S$, and the assertion follows.
\end{Remark}

\begin{Remark} \label{rem:dilat}
(a) We consider the Markov case where $\cA\Omega$ is dense in
$\cH$, so that $\cE_0 \cong L^2(Q, \Sigma_0, \mu)$ is a $G$-cyclic subspace
of $\cE$. Then $\phi(g) = E_0 U_g E_0$ satisf\/ies
$\phi(s) = \hat{U_s} = P_s$ for $s \in S$. Therefore the $G$-invariant
subspace of $\cE$ generated by $\cE_0$ carries a representation
which is a unitary dilation
of the representation $(P,\cH)$ of $S$ (cf.\ Remark~\ref{rem:dil}).

(b) The set of f\/inite products $f_{g_1} \cdots f_{g_n}$,
$g_j \in G$, $f_j \in \cE_0 \cong
L^\infty(M,\fS,\nu) \cap L^1(M,\fS,\nu)$
spans a dense subspace of $\cE$. Accordingly, the restriction
$g_j \in S$ leads to a dense subspace of $\cE_+$. This permits us to
give suf\/f\/icient conditions for $\cE_+$ to be cyclic in $\cE$. Since
\begin{gather*} U_g f_{g_1} \cdots f_{g_n}= f_{gg_1} \cdots f_{gg_n}, \end{gather*}
this is the case if, for every f\/inite subset $F \subeq G$, there exists a
$g \in G$ with $g F \subeq S$, i.e.,
$g^{-1} \prec_S F$. This means that the order $\prec_S$ on $G$ is {\it
filtered}, i.e., for $g_1, g_2 \in G$, there exists a~\mbox{$g_3 \in G$}
with $g_3 \prec_S g_1, g_2$. This property is easily seen to be equivalent to
$G= S^{-1}S$ because $g_1^{-1}g_2 \in S^{-1}S$ is equivalent to the existence
of $s_1, s_2 \in S$ with $g_3 :=g_2 s_2^{-1} = g_1 s_1^{-1} \prec_S g_1, g_2$
(cf.~\cite{HN93}). We conclude that $\cE_+$ is $G$-cyclic in $\cE$ if
$G = S^{-1}S$. Note that this is in particular the case if
$G = S \cup S^{-1}$.
\end{Remark}

\subsection{Total subsemigroups of Lie groups}\label{se:2.5}

In this section we brief\/ly discuss total subsemigroups. Our standard
references are \cite[Chapter~6]{HN93} and~\cite{La87}.
For Lie groups $G$, subsemigroups $S$ with dense interior
satisfying $G= S \cup S^{-1}$ are rare. Typical examples are:
\begin{itemize}\itemsep=0pt
\item[\rm(a)] $S = [0,\infty[ = \R_+ \subeq G = \R$.
\item[\rm(b)] $S = \R_+ \rtimes \R^\times_+ \subeq G = \R \rtimes \R^\times_+$
(the $ax+b$-group).
\item[\rm(c)] A subsemigroup $S \subeq G = \tilde\SL_2(\R)$
(the simply connected covering group of $\SL_2(\R)$), whose boundary is a $2$-dimensional subgroup
isomorphic to the $ax+b$-group. Actually $G$ is the simply connected covering
of the projective group $\PSL_2(\R)$ acting by orientation preserving projective
maps on the projective line
$\bP(\R^2) \cong \bS^1$. Identifying $\tilde\bS^1$ with $\R$, we have that
\begin{gather*} S = \{ g \in G \colon g(0) \geq 0\}
= \{ g \in G \colon g(\R_+) \subeq \R_+\}.\end{gather*}
In terms of an Iwasawa decomposition
$G = KAN$ with $K \cong \R$ and $K_+ := \R_+$, we have
$S = K_+ AN = AN K_+$ and $\R \cong K \cong G/AN$.
\end{itemize}

It is easy to see that $S$ is contained in a maximal subsemigroup
of $S$ and if $\g/\rad(\g)$ is a~compact Lie algebra, then
the classif\/ication of maximal subsemigroups~$M$~\cite[Theorem~6.18]{HN93}
implies that $\partial M$ is a codimension $1$ subgroup of $G$.
According to the classif\/ication of hyperplane subalgebras of
f\/inite-dimensional Lie algebras~\cite{Ho65, Ho90}, this means that,
if $N \subeq \partial M$ is the largest connected normal subgroup of
$G$ contained in $\partial M$, then $M/N \subeq G/N$ is one of the
three types described above. The examples under (a)--(c) are also semigroups
whose boundary $\partial S$ is a~sub\-group, so that $S^0$ is {\it well bounded}
in the sense of \cite{Do76, Po77}. Such a semigroup, for which
$\partial S$ contains no non-trivial connected normal subgroup,
is isomorphic to (a), (b) or (c) \cite{Do76,Po77}.

From \cite{NO15a} we know that, for the f\/irst two types (a) and (b),
there are non-trivial ref\/lection positive representations
for $(G,S,\tau)$, where $\tau$ is such that $G^\tau = \partial S$.
On $\fsl_2(\R)$, the f\/ixed point set of any involutive automorphism
is $1$-dimensional, hence cannot be a hyperplane. Suppose that,
for $G = \tilde\SL_2(\R)$, $G^\tau \subeq \partial S$. Then
$S$ is invariant under $s \mapsto s^\sharp = \tau(s)^{-1}$,
but one can show that all involutive continuous
representations of $(S,\sharp)$ are trivial. This follows from
the L\"uscher--Mack theorem \cite{MNO15} and the fact that all
involutive contractive representations of the group
$(\Aff(\R),\sharp)$ must be trivial on the subgroup $\{0\} \rtimes \R^\times_+$.
This is a consequence of the fact that if, for a selfadjoint operator~$A$, all operators $e^{tA}$ are contractions, we have
$\Spec(A) =\{0\}$, hence $A = 0$.

\begin{Example} $S = \R^d_+ \subeq G = \R^d$ satisf\/ies $G = S \cup S^{-1}$.
Since this subsemigroup is invariant under $\OO_{d-1}(\R)$, we likewise obtain
$G_1 = S_1 \cup S_1^{-1}$ for $G_1 := \R^d \rtimes \OO_{d-1}(\R)$ and
$S_1 := \R^d_+ \rtimes \OO_{d-1}(\R)$.
Note that $(G_1,\tau_1)$ with $\tau_1(b,A) := (\theta(b), A)$ is a symmetric Lie group
with one-dimensional~$\fq$.
\end{Example}

\subsection{The discrete case}\label{subsec:discrete}

Our reconstruction theorem applies in particular
to $(G,S,\tau) = (\Z,\N_0, - \id_\Z)$.
Here we may start with a {\it Markov operator}
$P$ on a polish probability space $(M,\fS,\nu)$, so that
 $P$ is positi\-vi\-ty preserving
and satisf\/ies $P1 = 1$ (cf.\ Appendix~\ref{subsec:6.4}).
If $P$ def\/ines a symmetric operator on $L^2(M,\fS,\nu)$,
then we obtain a unitary operator
$U$ acting by translation on $Q := M^\Z$ such that
$E_0 U E_0 = P$ holds for the conditional expectation def\/ined by
$\Sigma_0 := \ev_0^*(\fS) \subeq \Sigma$.

An interesting class of examples arises from f\/inite probability
spaces $M = \{1,\ldots, n\}$ with the $\sigma$-algebra $2^M$.
Then a Markov kernel $P$ on $M$ is given by the matrix
$P_{ij} := P(i,\{j\})$ which is stochastic in the sense that
$\sum_j P_{ij} = 1$ for every $i$.
If $\mu$ is a measure on $M$, encoded in the
row vector $p = (p_j)_{1 \leq j \leq n}$ with $p_j := \mu(\{j\})$, then $\mu P$
is represented by the row vector $pP$.
In particular, the invariance condition $\nu P =\nu$ translates into
$\sum_i p_i P_{ij} = p_j$ for every $j$. So $P$ acts on measures on $M$
(identif\/ied with row vectors) by right multiplication
and on functions on $M$ (represented by column vectors) by left matrix
multiplication.

If $P$ acts as a symmetric operator on $L^2(M,\nu)$, then a
Markov process $(X_n)_{n \in \Z}$ can now be obtained from the projections
$X_n \colon Q = M^\Z \to M$ and the shift invariant
probability measure $\mu$ on $Q$ is determined by
\begin{gather} \label{eq:start}
 \mu(X_0 = i) = p_i, \qquad \mu(X_1 = i|X_0 = j) = P_{ij}.
\end{gather}
In this context, we have for integers $t_1 \leq \cdots \leq t_n$:
\begin{align*}
\int_Q f_1(\omega(t_1))\cdots f_m(\omega(t_m))\, d\mu(\omega)
&= \big\la M_{f_1} P^{t_2 - t_1} \cdots M_{f_{m-1}} P^{t_m-t_{m-1}} M_{f_m} 1, 1 \big\ra_{L^2(M,\nu)} \\
&= \int_M M_{f_1} P^{t_2 - t_1} \cdots M_{f_{m-1}} P^{t_m-t_{m-1}} f_m\, d\nu(m)
\end{align*}
(Theorem~\ref{thm:1.8G2}).
Applying this relation to $\delta$-functions $f_j = \delta_{i_j}$, we obtain
\begin{gather*} \mu(X_{t_1} = i_1, \ldots, X_{t_n} = i_n)
= p_{i_1} \big(P^{t_2-t_1}\big)_{i_1 i_2} \big(P^{t_3-t_2}\big)_{i_2 i_3}\cdot \big(P^{t_n-t_{n-1}}\big)_{i_{n-1} i_n}.\end{gather*}
For $n = 1$, we obtain in particular
\begin{gather*} \mu(X_t = i) = p_i = \mu(X_0 = i).\end{gather*}
For $n = 2$ and $t \in \N_0$, we obtain
\begin{gather*} \mu(X_0 = i, X_t = j) = p_i \big(P^t\big)_{ij}.\end{gather*}
For $t = 1$, this leads to
\begin{gather*} \mu( X_0 = i,\, X_1 = j) = p_i P_{ij} \qquad \mbox{and} \qquad
 \mu(X_1 = j\,|\,X_0 = i) = P_{ij},\end{gather*}
which is~\eqref{eq:start}.

\section{Measures on path spaces for topological groups}\label{sec:3}

In this section we build a bridge between Markov processes
and the reconstruction results from the preceding section.
To this end, we return to the classical setting,
where the symmetric Lie group is $(\R, \R_+,-\id)$,
but the measure space is the path space
$P(Q) = Q^\R$ of a polish topological group $Q$.

If $\nu$ is a measure on a locally compact group
$Q$ and $(\mu_t)_{t \geq 0}$ is a convolution semigroup of symmetric
probability measures on $Q$ satisfying $\nu * \mu_t = \nu$ for
every $t > 0$, then the Klein--Landau reconstruction
theorem applies. Since the measure $\nu$
need not be f\/inite, $\nu$ may in particular be
the right invariant Haar measure $\mu_Q$ on $Q$.
The operators $P_t f = f * \mu_t$ def\/ine a~positive semigroup structure on
$L^2(Q,\nu)$, and the reconstruction process leads to a ref\/lection positive
unitary representation of $\R$ on $L^2(P(Q), \fB^\R, \nu)$ by translation
(Theorem \ref{thm:4.11}).
If $\mu_t = \gamma_t$ is the gaussian semigroup on $\R^d$, this
construction leads to an inf\/inite measure on $P(\R^d)$ which realizes in
particular the dilation of the heat semigroup as a subrepresentation
(cf.\ \cite{NO15a,Nel64}).

Here an interesting point is that the map
$Q \times P_*(Q) \to P(Q), (q,\omega) \mapsto q \omega$
yields a product decomposition of $P(Q)$
into the subgroup of constant paths and the subgroup
$P_*(Q) $
of pinned paths. The measures constructed on $P(Q)$ are actually
product measures in these coordinates, where $Q$ carries the measure $\nu$ and
the measure on $P_*(Q)$ is a probability measure determined by the~$\mu_t$. There
is a natural measure preserving $\R$-action on $P_*(Q)$ by
\begin{gather*} (V_t \omega)(s) := \omega(-t)^{-1} \omega(s-t). \end{gather*}
The measure in $P_*(Q)$ def\/ines a
$Q$-valued stochastic process $(X_t)_{t \in \R}$ with $X_\1 = \1$
(the constant function).
For the special case, where $\mu_t = \gamma_t$ is the gaussian semigroup
on $\R^d$, this construction leads to the Wiener measure on $P_*(\R^d)$.
Presently, we do not know how to obtain a
similar factorization if $(\R,\R_+,-\id_\R)$ is replaced
by some $(G,S,\tau)$,\footnote{This is caused by the
ambiguities related to the discrepancy between f\/inite
subsets $\{g_1,\ldots, g_n\}$ of $G$ and
increasing tuples $(g_1,\ldots, g_n)$ with respect to
$\prec_S$.} so we discuss the one-dimensional case in
Subsection~\ref{subsec:4.1}.

\subsection{One-parameter convolution semigroups of measures on polish groups}
\label{subsec:4.1}

In this subsection $Q$ is a {\it polish topological group.}
We write $\fB$ for the $\sigma$-algebra of Borel subsets of~$Q$.
A kernel $K \colon Q \times \fB \to [0,\infty]$
on $(Q,\fB)$ (cf.~Appendix~\ref{subsec:6.4}) is said to be
{\it left invariant} if
\begin{gather*} K(gh, gA) = K(h,A) \qquad \mbox{for all} \quad g,h \in Q, A \in \fB.\end{gather*}
Then $\mu(A) := K(\1,A)$ is a positive Borel measure on $Q$ such that $K(g,A)=\mu (g^{-1}A)$.
If, conversely, $\mu$ is a $\sigma$-f\/inite Borel measure on $Q$, then
$K(g,A) := \mu(g^{-1}A)$ is a left invariant kernel on $Q \times \fB$ because, for every
$A \in \fB$, the function $g \mapsto \mu(g^{-1}A)$ is measurable.
In fact, $\tilde A := \{ (x,y) \in Q^2 \colon xy \in A \}$ is a Borel subset of $Q \times Q$,
and $\tilde A \cap (\{g\} \times Q) = \{g\} \times g^{-1}A$.
We thus obtain a one-to-one correspondence between left invariant
Markov kernels $K$ on $Q$ and Borel probability measures $\mu$ on~$Q$.
Since the product of two kernels corresponds to the convolution
product of the corresponding measures,
this leads to a one-to-one correspondence between convolution semigroups
$(\mu_t)_{t \geq 0}$ of Borel probability measures on $Q$ and left invariant Markov
semigroups $(P_t)_{t \geq 0}$ on $(Q,\fB)$. We recall from
Remark~\ref{rem:2.21}(b) that the kernel $P_t$ acts on functions by right convolution
\begin{gather*} P_t f = f * \mu_t \qquad \mbox{for} \quad t \geq 0,\end{gather*}
where $f*\mu_t(x)=\int_G f(xg)\, d\mu_t(g)$.

Consider the path group $P(Q) = Q^\R$ and the subgroup
\begin{gather*} P_*(Q) := \{ \omega \in P(Q) \colon \omega(0) = \1\}.\end{gather*}
Starting with a convolution semigroup $(\mu_t)_{t \geq 0}$, we obtain
for the initial distribution $\delta_\1$
the Markov process $(X_t)_{t \geq 0}$ with independent stationary increments
\cite[Corollary~35.4]{Ba96}. The distribution of this process is a
probability measure $P^{\mu,+}$ on
\begin{gather*} P(Q)^+ = \big\{ \omega \in Q^\R \colon (\forall\, t \leq 0)\,\, \omega(t) = \1\big\}
\cong Q^{\R^\times_+}.\end{gather*}
For the f\/inite distributions on $Q^n$, described
by the kernels $P_\bt$, $\bt = (t_1, \ldots, t_n)$, $t_1 < \dots < t_n$,
we obtain with
Remark~\ref{rem:mark}
\begin{align}
 P_\bt(B) &= \int_{Q^{n}} \chi_B(x_1, \ldots, x_n)
\mu_{t_1}(dx_1) \mu_{t_2- t_1}(x_1^{-1}dx_2) \cdots
\mu_{t_n- t_{n-1}}(x_{n-1}^{-1} dx_n)\notag\\
&= \int_{Q^{n}} \chi_B(x_1, x_1 x_2, \ldots, x_1 \cdots x_n)
\mu_{t_1}(dx_1)\mu_{t_2- t_1}(dx_2)\cdots \mu_{t_n- t_{n-1}}(dx_n).\label{eq:density}
\end{align}
This means that
\begin{gather*} P_\bt= (\psi_n)_* (\mu_{t_1} \otimes \mu_{t_2-t_1}\otimes \cdots \otimes
\mu_{t_n-t_{n-1}})\end{gather*}
for the map
\begin{gather} \label{eq:psi}
\psi_n(g_1, \ldots, g_n) = (g_1, g_1 g_2, \dots, g_1 \cdots g_n).
\end{gather}

\begin{Remark}
Since the inverse of the map $\psi_n$ in \eqref{eq:psi}
is given by
\begin{gather*} \psi_n^{-1}(h_1, \ldots, h_n) = \big(h_1, h_1^{-1} h_2, \cdots, h_{n-1}^{-1} h_n\big),\end{gather*}
it follows that the random variables
\begin{gather*} X_{t_1}, X_{t_1}^{-1} X_{t_2}, \dots, X_{t_{n-1}}^{-1} X_{t_n} \end{gather*}
are independent. Since the distribution of $X_t^{-1} X_{t+s}$ is $\mu_s$ for every
$t,s \geq 0$, it follows that $(X_t)_{t \geq 0}$
has independent stationary increments
(cf.~also \cite[Corollary~35.4]{Ba96}).
\end{Remark}

For a curve $\omega \colon \R \to Q$, we def\/ine
$(\theta \omega)(t) := \omega(-t)$.
We write $P(Q)^\pm$ for the subgroup consisting
of paths whose value is $\1$ on $\R_-$, resp., $\R_+$ and
$Q \subeq P(Q)$ with the subgroup of constant paths.
For the following proposition, we recall that the involution
on the space $\cM(Q)$ of f\/inite real Borel measures on $Q$ is def\/ined by
\begin{gather*} \int_Q f(g)\, d\mu^*(g) = \int_Q f\big(g^{-1}\big)\, d\mu(g).\end{gather*}

\begin{Proposition} \label{prop:4.2} Let $P^{\mu,-} := \theta_*P^{\mu,+}$ and $P^\mu$
be the image of the measure $P^{\mu,-} \otimes P^{\mu,+}$ on~$P_*(Q)$
under the bijective product map
\begin{gather*} \Phi \colon \ P(Q)^- \times P(Q)^+ \to P_*(Q), \qquad
(\omega_-, \omega_+) \mapsto \omega_- \cdot \omega_+.\end{gather*}
If $\mu_t^* = \mu_t$ holds for every $t \geq 0$, then this measure is invariant
under $\theta$ and the one-parameter group of Borel isomorphisms, given by
\begin{gather*} (V_t \omega)(s) := \omega(-t)^{-1} \omega(s - t),\qquad t \in \R.\end{gather*}
\end{Proposition}

\begin{proof} The $\theta$-invariance of $P$ follows immediately from its construction and
\begin{gather*} \theta\Phi(\omega_-, \omega_+) = \theta(\omega_- \cdot \omega_+)
= \theta(\omega_-) \cdot \theta(\omega_+)
= \theta(\omega_+) \cdot \theta(\omega_-)
= \Phi(\theta(\omega_+), \theta(\omega_-)).\end{gather*}
We also note that $\theta V_t\theta= V_{-t}$ follows from
\begin{gather} \label{eq:thetarel}
(\theta V_t \theta \omega)(s) = (V_t\theta\omega)(-s)
= (\theta \omega)(-t)^{-1} (\theta \omega)(-s -t)
= \omega(t)^{-1} \omega(s +t) = (V_{-t}\omega)(s).
\end{gather}

Since $P$ is uniquely determined by its images under evaluation in f\/inite tuples
$\bt = (t_1, \ldots, t_n)$, $t_i < t_{i+1}$,
it suf\/f\/ices to show that the corresponding distributions
$P_\bt$ on $Q^n$ satisfy $V_t P_\bt = P_{(t_1 + t, \ldots, t_n + t)}$. In view of~\eqref{eq:thetarel}, we may assume that $t > 0$.
We may further assume that
$t_k = 0$ for some $k$ and that $t_m = -t$ for some $m < k$.

First we note that, for measurable functions $f_1, \ldots,f_{k-1}, f_{k+1}, \ldots, f_n$
on $Q$, we have
\begin{gather*}
\int_{P_*(Q)} \prod_{j \not=k} f_j(\omega(t_j))\, dP^\mu(\omega) \\
\qquad{} = \int_{P(Q)^+} \int_{P(Q)^+} \prod_{j=1}^{k-1}f_j(\omega_-(-t_j))\prod_{j=k+1}^n
f_j(\omega_+(t_j))\,
dP^{\mu,-}(\omega_-) dP^{\mu,+}(\omega_+) \\
\qquad{} \buildrel{\eqref{eq:density}}\over{=}
\int_{Q^{n-1}} f_1(g_{k-1} \cdots g_1)\cdots f_{k-1}(g_{k-1})
f_{k+1}(g_{k+1})\cdots f_n(g_{k+1} \cdots g_n) \\
\qquad \quad{}\times d\mu_{t_2-t_1}(g_1) \cdots d\mu_{t_{k-1}- t_{k-2}}(g_{k-2})d\mu_{-t_{k-1}}(g_{k-1})\\
 \qquad \quad{}\times
d\mu_{t_{k+1}}(g_{k+1})d\mu_{t_{k+2} - t_{k+1}}(g_{k+2}) \cdots d\mu_{t_n - t_{n-1}}(g_n).
\end{gather*}
For $F(\omega) = \prod\limits_{j = 1}^n f_j(\omega(t_j))$, this leads with $t_m = -t$ and $t_k = 0$ to
\begin{gather*}
\int_{P(Q)} (V_t F)(\omega)\, dP^\mu(\omega) = \int_{P(Q)} F\big(\omega(t)^{-1}\omega(\cdot + t)\big)
\, dP^\mu(\omega) \\
= \int_{P(Q)^+} \int_{Q^{m-1}}
f_1\big(\omega(t)^{-1}g_{m-1} \cdots g_1\big) \cdots f_{m-1}\big(\omega(t)^{-1} g_{m-1}\big) f_m\big(\omega(t)^{-1}\big) \\
\qquad \quad{}\times f_{m+1}\big(\omega(t)^{-1} \omega(t+ t_{m+1})
\cdots f_n\big(\omega(t)^{-1}\omega(t+t_n)\big)\big) \\
\qquad \quad{}\times  d\mu_{-t_{m-1}-t}(g_{m-1}) d\mu_{t_{m-1}- t_{m-2}}(g_{m-2}) \cdots d\mu_{t_2-t_1}(g_1)
dP^{\mu,+}(\omega) \\
= \int_{Q^{n-1}}
f_1\big((g_m \cdots g_{k-1})^{-1}g_{m-1} \cdots g_1\big) \cdots f_{m-1}\big((g_m \cdots g_{k-1})^{-1}g_{m-1}\big)
f_m\big((g_m \cdots g_{k-1})^{-1}\big) \\
\qquad \quad{}\times f_{m+1}\big((g_m \cdots g_{k-1})^{-1} g_m\big) \cdots
f_{k-1}((g_m \cdots g_{k-1})^{-1} g_m \cdots g_{k-2}) \\
\qquad\quad{}\times  f_{k+1}(g_{k+1}) \cdots f_n(g_{k+1} \cdots g_n)\,
 d\mu_{-t_{m-1}-t}(g_{m-1}) d\mu_{t_{m-1}- t_{m-2}}(g_{m-2}) \cdots d\mu_{t_2-t_1}(g_1) \\
\qquad\quad{}\times d\mu_{t_{m+1}+t}(g_{m})d\mu_{t_{m+2} - t_{m+1}}(g_{m+1}) \cdots \\
\qquad\quad{}\times d\mu_{t_{k}-t_{k-1}}(g_{k-1})d\mu_{t_{k+1}-t_{k}}(g_{k+1}) \cdots d\mu_{t_n - t_{n-1}}(g_{n})\\
= \int_{Q^{n-1}}
f_1\big(g_{k-1}^{-1} \cdots g_m^{-1} g_{m-1} \cdots g_1\big) \cdots f_{m-1}\big(g_{k-1}^{-1} \cdots g_m^{-1}g_{m-1}\big)
f_m\big(g_{k-1}^{-1} \cdots g_m^{-1}\big) \\
\qquad \quad{}\times f_{m+1}\big(g_{k-1}^{-1} \cdots g_{m+1}^{-1}\big) \cdots
f_{k-1}\big(g_{k-1}^{-1}\big) f_{k+1}(g_{k+1}) \cdots f_n(g_{k+1} \cdots g_n) \\
\qquad \quad{}\times d\mu_{t_2-t_1}(g_1) \cdots
d\mu_{t_{m-1}- t_{m-2}}(g_{m-2})
d\mu_{t_m-t_{m-1}}(g_{m-1})
d\mu_{t_{m+1}-t_m}(g_{m})\\
\qquad\quad{}\times d\mu_{t_{m+2} - t_{m+1}}(g_{m+1}) \cdots
d\mu_{t_{k}-t_{k-1}}(g_{k-1}) d\mu_{t_{k+1}-t_{k}}(g_{k+1}) \cdots d\mu_{t_n - t_{n-1}}(g_{n}) \\
\buildrel{(\dagger)}\over{=} \int_{Q^{n-1}}
f_1(g_{k-1}\cdots g_1)\cdots f_{m-1}(g_{k-1} \cdots g_{m-1})
f_m(g_{k-1}\cdots g_m) \\
\qquad \quad{}\times f_{m+1}(g_{k-1} \cdots g_{m+1}) \cdots
f_{k-1}(g_{k-1}) f_{k+1}(g_{k+1}) \cdots f_n(g_{k+1} \cdots g_n) \\
\qquad\quad{}\times d\mu_{t_2-t_1}(g_1) \cdots d\mu_{t_m-t_{m-1}}(g_{m-1})
d\mu_{t_{m+1}-t_m}(g_{m})d\mu_{t_{m+2} - t_{m+1}}(g_{m+1}) \cdots \\
\qquad\quad{}\times d\mu_{t_{k}-t_{k-1}}(g_{k-1}) d\mu_{t_{k+1}-t_{k}}(g_{k+1}) \cdots d\mu_{t_n - t_{n-1}}(g_{n})\\
= \int_{Q^{n-1}}
f_1(g_{k-1}\cdots g_1)\cdots f_{k-1}(g_{k-1}) f_{k+1}(g_{k+1}) \cdots f_n(g_{k+1} \cdots g_n) \\
\qquad \quad{}\times d\mu_{t_2-t_1}(g_1) \cdots d\mu_{t_{k}-t_{k-1}}(g_{k-1})
d\mu_{t_{k+1}-t_{k}}(g_{k+1}) \cdots d\mu_{t_n - t_{n-1}}(g_{n}).
\end{gather*}
Here we have used in $(\dagger)$ the symmetry of the measures $\mu_t$ to get rid of the
inverses. The preceding calculation shows that the transformations $(V_t)_{t \in\R}$ on
$P_*(Q)$ leave the probability measure $P^\mu$ invariant.
\end{proof}

\begin{Remark} Note that the transformations $(V_t)_{t \in \R}$ on $P_*(Q)$ satisfy
\begin{gather*} V_t X_s = X_t^{-1} X_{s+t} \qquad \mbox{because} \quad
(V_t X_s)(\omega) = (V_{-t}\omega)(s) = \omega(t)^{-1}\omega(s+t).\end{gather*}
On functions on $P_*(Q)$, we obtain the action
\begin{gather*} (V_t F)(\omega) = F(\omega(t)^{-1} \omega(\cdot + t))\end{gather*}
which is isometric on $L^2(P_*(Q), \fB^{\R^\times}, P^\mu)$.
\end{Remark}

\begin{Example}(a) (Poisson semigroups)
For every $g \in K$, we obtain a convolution
semigroup of measures
by
\begin{gather*} \mu_t = e^{-t} \sum_{k = 0}^\infty \frac{t^k}{k!} \delta_{g^k} = e^{t(\delta_g - \delta_0)}
\qquad \mbox{for} \quad t \geq 0.\end{gather*}
This semigroup consists of symmetric measures if and only if $\delta_g$ is symmetric,
i.e., $g = g^{-1}$. In this case $g^{2n} = \1$ and $g^{2n+1} = g$ for all $n$, so that
\begin{gather*} \mu_t = \frac{1 + e^{-2t}}{2} \delta_\1 + \frac{1 - e^{-2t}}{2} \delta_g.\end{gather*}
Note that the limit for $t \to \infty$ is the measure $\shalf(\delta_\1 + \delta_g)$.
This particular example can also be considered as a one-parameter subsemigroup
of the semigroup of Markov matrices on $\R^2$, resp., Markov operators on the two-element
set, where $g$ is the transposition of the two elements
(see also Example~\ref{ex:findim}).

(b) (Velocity processes) If $K$ is a Lie group and $X \in \g$, then
we obtain a semigroup of measures by
\begin{gather*} \mu_t = \delta_{\exp tX} \qquad \mbox{for} \quad t \geq 0.\end{gather*}
For these semigroups the condition $\mu_t^* = \mu_t$ for every $t \geq 0$ is satisf\/ied
only if $X =0$, i.e., $\mu_t = \delta_\1$ for every~$t$.
\end{Example}

\begin{Example} The class of polish groups is quite large.

(a) Among locally compact groups, the second countable ones carry a complete left invariant
metric, turning them in a polish group. This covers in particular all f\/inite-dimensional
Lie groups with at most countably many connected components.

(b) Separable Fr\'echet--Lie groups are also polish, and this class contains in particular
all connected Banach--Lie groups whose Lie algebra is separable, gauge
groups and groups of dif\/feomorphisms of compact manifolds.

(c) Another important example of a polish group is the unitary group
$\U(\cH)_s$ of a separable Hilbert space $\cH$, endowed with the strong operator topology
(cf.\ \cite[Proposition~II.1]{Ne97} or \cite[Theorem~II.1, p.~93]{Sch73}).
\end{Example}

\begin{Remark} In recent years, many interesting one-parameter semigroups
of measures have been studied on inf\/inite-dimensional groups.

(a) On the group $\Diff(\bS^1)$ of dif\/feomorphisms of the circle with respect to
the $H^{3/2}$-metric, Brownian motion has been constructed by P.~Malliavin.
A dif\/ferent approach also exhibiting the invariance under rigid rotations
is described in \cite{Fa02} (cf.\ also \cite{Go08,Gr98}).

(b) In \cite{BS03} Brownian motion on compact groups is studied.
This is def\/ined to be a stochastic process $(X_t)_{t \geq 0}$ with
values in $K$ with
$X_0 = \1$, independent stationary increments and continuous sample paths.
In addition, it is assumed to be symmetric, biinvariant and {\it non-degenerate},
i.e., $X_t$ visits every open subset with positive probability.
For the corresponding convolution semigroup $(\mu_t)_{t \geq 0}$ this implies that
\begin{gather}
 \label{eq:decay}
\lim_{t \to 0} t^{-1} \mu_t\big(V^c\big) = 0
\end{gather}
for every open $\1$-neighborhood $V$ in $K$.
Conversely, \cite[Theorem~1.2]{BS03} characterizes the convolution semigroups
$(\mu_t)_{t \geq 0}$ corresponding to Brownian motions as those which satisfy,
in addition to \eqref{eq:decay}, that
$\mu_t \to \delta_\1$ weakly on $C(K)$, $\mu_t^* = \mu_t$,
$\mu_t$ is conjugation invariant (also called central),
and $\supp(\mu_t) = K$ for every $t > 0$.

(c) In \cite{Dr03}, Driver studies Brownian motion on the inf\/inite-dimensional
Banach--Lie group
\begin{gather*} W(K) := C_*([0,1],K) = \{ \omega \in C([0,1], K) \colon \omega(0) = \1\},\end{gather*}
where $K$ is a connected Lie group with compact Lie algebra,
i.e., $K \cong C \times \R^d$ for some $d \in \N_0$ and a compact Lie
group $C$. Here one has to construct a semigroup
$(\mu_t)_{t \geq 0}$ of {\it heat kernel measures} on the space
$W(\R^d)$. On $W(\R^d)$ one has to get hold of the smoothness
properties of the heat kernel measures $\mu_t$ corresponding to the
Wiener measure. The corresponding Hilbert space~$H(\g)$ is the space of f\/inite energy paths with values in the
Lie algebra $\g$ and the measures $\mu_t$ on~$W(K)$ are quasi-invariant
under left and right multiplication with elements
of the corresponding Cameron--Martin group~$H(K)$ \cite[Theorem~7.7]{Dr03}.
The unitary representations of $H(K)$ by left, resp.,
right multiplications has been identif\/ied recently with the
so-called energy representation of this group~\cite{ADGV15}. With similar techniques, the existence of heat kernel measures
is also obtained for the space of pinned paths with values in a~compact Riemannian manifold~$M$.
\end{Remark}

\subsection{Standard path space structures for locally compact groups}\label{subsec:4.2}

In this section we assume that the group $Q$ is locally compact and that the convolution semigroup $(\mu_t)_{t \geq 0}$ of probability measures on $Q$ is
{\it strongly continuous} in the sense that $\lim\limits_{t \to 0} \mu_t = \delta_\1 = \mu_0$ weakly on the space $C_b(Q)$ of bounded
continuous functions on~$Q$. We further assume that $\nu$ is a measure on $Q$ satisfying
$\nu * \mu_t = \nu$ for every $t > 0$, and, in addition, that the operators
\begin{gather*} P_t f := f * \mu_t \end{gather*}
on $L^2(Q,\nu)$ are symmetric. If
$\nu$ is a right Haar measure, then the symmetry of the operators $P_t$ is equivalent to
$\mu_t^* = \mu_t$.

\begin{Remark} Let $\mu_Q$ be a right Haar measure on $Q$ and assume that all
the measures $\mu_t$ are symmetric.
Integrating the right regular representation $(\pi^r(g)f)(x) = f(xg)$
of $Q$ on $L^2(Q,\mu_Q)$, we obtain a $*$-representation of
the convolution algebra $\cM(Q)$ on $L^2(Q)$ by
\begin{gather*} \pi^r(\mu) := \int_Q \pi^r(g)\, d\mu(g), \qquad
\big(\pi^r(\mu)f\big)(x) = \int_Q f(xg)\, d\mu(g) = (f * \mu)(x).\end{gather*}
Then $P_t := (\pi^r(\mu_t))_{t \geq 0}$ is a strongly continuous semigroup
of hermitian contractions on $L^2(Q)$ (here we use $\mu_t^* = \mu_t$)
which are Markov operators.
This is a positive semigroup structure because the continuity in measure
(SPS4) follows from the continuity of the action of~$Q$
(cf.\ Remark~\ref{rem:AK-exs}). Hence the Klein--Landau reconstruction
Theorem~\ref{thm:1.8G2} provides a path space model for the corresponding
dilation representation of~$\R$~\cite{NO15a}.
\end{Remark}

For $P_t$ as above, formula \eqref{eq:5}
\begin{gather*} \int_{P(Q)} f_{t_1} \otimes \cdots \otimes f_{t_n}\, dP^\mu_\bt
= \int_Q f_1 P_{t_2 -t_1} f_2\cdots P_{t_n-t_{n-1}} f_n\, d\nu\end{gather*}
in Remark~\ref{rem:mark-meas} specializes for
$t_1 \leq \cdots \leq t_n$ and $0 \leq f_j$ to
\begin{gather*}
\int_{P(Q)} f_1 \otimes \cdots \otimes f_n\, dP^\mu_\bt \\
\qquad{} = \int_{Q^n} f_1(g_1) f_2(g_1 g_2) \cdots f_n(g_1 \cdots g_n)\,
d\nu(g_1) d\mu_{t_2-t_1}(g_2) \cdots d\mu_{t_n-t_{n-1}}(g_n) \\
\qquad{}= \int_{Q^n} f_1(g_1) f_2(g_2) \cdots f_n(g_n)\,
d\psi_*(\nu \otimes \mu_{t_2-t_1} \otimes \cdots \otimes \mu_{t_n-t_{n-1}}),
\end{gather*}
where
$\psi(g_1,\ldots, g_n) = (g_1, g_1 g_2, \dots, g_1 \cdots g_n)$.
We conclude that
\begin{gather*} P^\mu_\bt = \psi_*(\nu \otimes \mu_{t_2-t_1} \otimes \cdots \otimes \mu_{t_n-t_{n-1}}).\end{gather*}

\begin{Lemma} \label{lem:consist} Let $\nu$ be a measure on $Q$ satisfying
$\nu * \mu_t = \nu$ for every $t > 0$. Then we obtain for
$t_1 \leq \cdots \leq t_n$ and
$\bt := (t_1, \ldots, t_n)$ on $Q^n$ a consistent family of measures
\begin{gather*} P^\mu_\bt := (\psi_n)_* (\nu \otimes \mu_{t_2-t_1}\otimes \cdots \otimes
\mu_{t_n-t_{n-1}}).\end{gather*}
If $Q$ is a polish group, this leads to a unique measure $P^\mu$ on $P(Q)$ with
$(\ev_\bt)_* P^\mu = P^\mu_\bt$ for $t_1 < \ldots < t_n$.
\end{Lemma}

\begin{proof} This follows from Remark~\ref{rem:mark}.
For the sake of clarity, we
give a direct argument for the consistency of the measures $P^\mu_\bt$.
Pick $j \in \{1,\ldots, n\}$ and let
$p \colon Q^n \to Q^{n-1}$ denote the projection omitting the $j$th component.
Then
\begin{gather*} (p \circ \psi_n)(g_1, \ldots, g_n)
= \psi_{n-1}(g_1, \ldots, g_{j-1}, g_j g_{j+1}, g_{j+2}, \ldots, g_n) \end{gather*}
implies for $j \geq 2$ that
\begin{gather*}
 p_*(\psi_n)_* (\nu \otimes \mu_{t_2-t_1}\otimes \cdots \otimes
\mu_{t_n-t_{n-1}}) \\
\qquad{} = (\psi_{n-1})_* (\nu \otimes \cdots \otimes \mu_{t_j - t_{j-1}} *
\mu_{t_{j+1} - t_j} \otimes \cdots \otimes \mu_{t_n-t_{n-1}}) \\
\qquad{}  = (\psi_{n-1})_* (\mu_{t_1} \otimes \cdots \otimes \mu_{t_{j+1} - t_{j-1}}
\otimes \cdots \otimes \mu_{t_n-t_{n-1}}).
\end{gather*}
For $j = 1$, we obtain
\begin{align*}
 p_*(\psi_n)_* (\nu \otimes \mu_{t_2-t_1}\otimes \cdots \otimes
\mu_{t_n-t_{n-1}})
&= (\psi_{n-1})_* (\nu * \mu_{t_2 - t_1} \otimes \mu_{t_3 - t_2} \otimes
\cdots \otimes \mu_{t_n-t_{n-1}}) \\
&= (\psi_{n-1})_* (\nu \otimes \mu_{t_3 - t_2} \otimes \cdots
\otimes \cdots \otimes \mu_{t_n-t_{n-1}}).
\end{align*}
This implies consistency of the measures $P^\mu_\bt$ on $Q^n$.
The consistency condition implies the existence
of a measure $P^\mu$ on $Q^{\R}$ with
$(\ev_\bt)_*\mu = P^\mu_\bt$ for
$t_1 \leq \cdots \leq t_n$
(cf.\ Def\/inition~\ref{def:g-process}).
\end{proof}

From the Klein--Landau reconstruction theorem we immediately obtain the following specialization. We refer to~\cite{NO15a} for other constructions of this dilation.

\begin{Theorem} \label{thm:4.11} Suppose that $Q$ is a second countable locally compact group.
Let $\mu$ be the measure on~$Q^\R$ corresponding to the
symmetric convolution semigroup $(\mu_t)_{t \geq 0}$ of probability measures
on~$Q$ and the measure $\nu$ on $Q$ for which the operators $P_t f = f * \mu_t$
define a positive semigroup structure on $L^2(Q,\nu)$. Then the translation action
$(U_t \omega)(s) := \omega(s-t)$ on $P(Q) = Q^\R$ is measure preserving
and $\mu$ is invariant under $(\theta\omega)(t) := \omega(-t)$.
We thus obtain a reflection positive one-parameter group of Markov type
on
$\cE := L^2(P(Q), \fB^\R,\mu)$ with respect to
$\cE_+ := L^2(P(Q), \fB^{\R_+},\mu)$, for which
$\cE_0 := \ev_0^*(L^2(Q,\nu))\cong L^2(Q,\nu)$
and $\hat \cE \cong L^2(Q,\nu)$ with $q(F) = E_0 F$ for $F \in \cE_+$. We further have
\begin{gather*} E_0 U_t E_0 = P_t \qquad \mbox{holds for} \quad P_t f = f * \mu_t,\end{gather*}
so that the $U$-cyclic subrepresentation generated by $\cE_0$
is a unitary dilation of the hermitian one-parameter semigroup $(P_t)_{t \geq 0}$ on $L^2(Q,\nu)$.
\end{Theorem}

\begin{Example} (a) For $Q = \R^d$, the heat semigroup is given on
$L^2(\R^d)$ by
\begin{gather*} e^{t\Delta}f = f * \gamma_t \qquad \mbox{where} \quad
d\gamma_t(x) = \frac{1}{(2\pi t)^{d/2}} e^{-\frac{1}{2}\frac{\|x\|^2}{t}}\,
dx.\end{gather*}
We call the corresponding measure on $Q^\R$ the {\it Lebesgue--Wiener
measure} (cf.\ Theorem~\ref{thm:4.11}).

(b) If $Q$ is a f\/inite-dimensional Lie group and $X_1, \ldots, X_n$ is a basis of the Lie algebra,
then we obtain a left invariant {\it Laplacian} by
$\Delta := \sum\limits_{j =1}^n L_{X_j}^2,$
where $L_{X_j}$ denotes the right invariant vector f\/ield with
$L_{X_j}(\1) = X_j$. Then there also exists a semigroup $(\mu_t)_{t \geq 0}$
of probability measures on~$Q$ such that \cite[Section~8]{Nel69}
\begin{gather*} e^{t\Delta} f = f * \mu_t \qquad \mbox{for} \quad t \geq 0.\end{gather*}
Accordingly, we obtain a {\it Haar--Wiener measure} on the path space
$Q^\R$.
\end{Example}

\section[Gaussian $(G,S,\tau)$-probability spaces]{Gaussian $\boldsymbol{(G,S,\tau)}$-probability spaces}\label{sec:2}

In this section we discuss the second quantization functor and its connection to
gaussian $(G,S,\tau)$-probability spaces. We then discuss equivalence of gaussian measures for
reproducing kernel Hilbert spaces. The main results of this section are contained
in Subsection~\ref{se:4.4}, where
we discuss gaussian measures on the space of distributions on a Lie group.
Here the distribution vectors of
unitary representations play an important role.

\subsection{Second quantization and gaussian processes} \label{se:SecondQ}

\begin{Definition} \label{def:1.1} Let $\cH$ be a real Hilbert space. A {\it gaussian
random process indexed by~$\cH$}
is a~random process $(\phi(v))_{v \in \cH}$ on a probability space $(Q,\Sigma, \mu)$
such that
\begin{itemize}\itemsep=0pt
\item[\rm(GP1)] $\{\phi(v) \colon v \in \cH\}$ is full, i.e., these random variables generate the
$\sigma$-algebra $\Sigma$ modulo zero sets.
\item[\rm(GP2)] Each $\phi(v)$ is a gaussian random variable of mean zero.
\item[\rm(GP3)] $\la \phi(v), \phi(w) \ra
= \la v, w\ra$ is the inner product on $\cH$.
\end{itemize}
\end{Definition}

According to \cite[Theorems~I.6, I.9]{Si74}, gaussian random processes indexed by $\cH$
exist and are unique up to isomorphisms of probability spaces. This means that,
if $(Q,\Sigma,\mu)$ is the corresponding probability space, then the algebra
$L^\infty(Q,\mu)$ with its state given by $\mu$ is uniquely determined by $\cH$
(cf.\ \cite[Section~1.1]{Si74}). Its projections correspond to $\Sigma/J_\mu$, where
$J_\mu$ is the ideal of zero sets\footnote{A short proof for the uniqueness can be derived from reproducing kernel techniques. If
$\phi_j \colon \cH \to L^2(Q_j, \mu_j)$, $j=1,2$, are two realizations, then the corresponding
covariance kernels on $\cH$ coincide, so that there exists a~unique unitary operator $\Phi \colon L^2(Q_1, \mu_1) \to L^2(Q_2, \mu_2)$ such that
$\Phi \circ \phi_1 = \phi_2$. Accordingly, $\Phi \circ \cA_1 \circ \Phi^{-1} = \cA_2$
for the von Neumann algebras $\cA_j := \{ e^{i\phi_j(v)} \colon v \in \cH\}'' \cong L^\infty(Q_j, \mu_j)$.}.

\begin{Theorem} \label{thm:dual} Let $V$ be a real vector space,
$V^*$ be its algebraic dual and
$\fB^*$ be the smallest $\sigma$-algebra for which all evaluation functions
$V^* \to V, \alpha \mapsto \alpha(v)$, are
measurable. Then a function $\chi \colon V \to \C$ is the Fourier transform
$\chi = \hat\mu$ of a measure $\mu$ on $(V^*, \fB^*)$
if and only if $\chi$ is positive definite and continuous
on every finite-dimensional subspace. In this case $\mu$ is uniquely determined.
\end{Theorem}

\begin{proof} Since the Fourier transform $\hat\mu$ is positive def\/inite and sequentially
continuous, it is in particular continuous on every f\/inite-dimensional subspace of~$V$.
The converse follows from \cite[Theorem~16.2]{Ya85}.
\end{proof}

\begin{Example} \label{ex:1.1} Theorem~\ref{thm:dual}
 implies that the gaussian process indexed by
$\cH$ may be realized by the probability measure
$\gamma_\cH$ on $(\cH^a,\fB^*)$ (the algebraic dual)
whose Fourier transform is
\begin{gather*}\hat\gamma_\cH(v) = e^{- \frac{\|v\|^2}{2}} . \end{gather*}
Here $\phi(v)(\alpha) = \alpha(v)$, and
\begin{gather*} \int_{\cH^a} e^{i \alpha(v)}\, d\gamma_\cH(\alpha) = e^{-\frac{\|v\|^2}{2}} \qquad \mbox{implies}\quad
\int_{\cH^a} \alpha(v)^2\, d\gamma_\cH(\alpha) = \|v\|^2.\end{gather*}
This leads to
\begin{gather*} \big\la e^{i\phi(v)}, e^{i\phi(w)}\big\ra = e^{-\frac{\|v-w\|^2}{2}}
= e^{-\frac{\|v\|^2}{2}}e^{-\frac{\|w\|^2}{2}} e^{\la v,w\ra},\end{gather*}
so that $K_v := e^{i\phi(v)} e^{\frac{\|v\|^2}{2}}$ satisf\/ies
\begin{gather*} \la K_v, K_w \ra = e^{\la v,w \ra}.\end{gather*}
\end{Example}

\begin{Example} \label{ex:gauss-rd}
If $\cH$ is f\/inite-dimensional, so that $\cH \cong \R^d$ for some $d \in \N_0$,
then $\cH^a \cong \R^d$ and $d\gamma_\cH(x) = \frac{1}{(2\pi)^{d/2}}
e^{-\frac{\|x\|^2}{2}}\, dx$ is a gaussian measure.
\end{Example}

We now introduce the second quantization. Conceptually the easiest way to def\/ine
second quantization as a functor
is to associate to a (real or complex) Hilbert space $\cH$ its Fock space
\begin{gather*} \cF(\cH) := \hat\bigoplus_{n \in \N_0} S^n(\cH),\end{gather*}
where
$S^n(\cH) := (\cH^{\otimes n})^{S_n}$ is the closed subspace of
$S_n$-invariant vectors in the $n$-fold tensor power of $\cH$.
For $v_1, \ldots, v_n \in \cH$, we def\/ine the {\it symmetric product} by
\begin{gather*} v_1 \vee \cdots \vee v_n := P_+(v_1 \otimes \cdots \otimes v_n), \end{gather*}
where $P_+ \colon \cH^{\otimes n} \to S^n(\cH)$ is the orthogonal projection.
The inner products of such elements are given by
\begin{align*}
\la v_1 \vee \cdots \vee v_n, w_1 \vee \cdots \vee w_n \ra
&= \la v_1 \vee \cdots \vee v_n, w_1 \otimes \cdots \otimes w_n \ra \\
&= \frac{1}{n!} \sum_{\sigma \in S_n} \la v_{\sigma(1)}, w_1 \ra
\cdots \la v_{\sigma(n)}, w_m \ra.
\end{align*}
In particular, with
$v^n=\underbrace{v\vee \cdots \vee v}_{n\text{-times}}$, we have
\begin{gather} \label{eq:powerrel}
\la v^n, w^n \ra = \la v, w \ra^n \qquad \mbox{and} \qquad \|v^n\| = \|v\|^{n}.
\end{gather}

Clearly, every contraction $A \colon \cH \to \cK$ def\/ines a contraction $\Gamma(A) \colon
\cF(\cH) \to \cF(\cK)$ by
\begin{gather*} \Gamma(A)(v_1 \vee \cdots \vee v_n)
:= Av_1 \vee \cdots \vee A v_n\end{gather*}
and it is clear that $\Gamma(AB) = \Gamma(A)\Gamma(B)$ and
$\Gamma(A^*) = \Gamma(A)^*$. In particular, we obtain a representation
of the involutive semigroup of contractions on $\cH$ in the Fock space $\cF(\cH)$
and $\cF$ def\/ines an endofunctor from the category
of Hilbert spaces whose morphism are contractions into itself.
The problem with this approach is that it completely ignores positivity issues.

If $\cH$ is a Hilbert space, then any $v \in \cH$ def\/ines a
function
$\la v, \cdot\ra$, $u\mapsto \la v, u\ra$,
 which is linear if $\cH$ is real and
antilinear if~$\cH$ is complex. For
\begin{gather*}\Exp(v) := \sum_{n = 0}^\infty \frac{1}{\sqrt{n!}} v^n , \end{gather*}
we now derive
from \eqref{eq:powerrel} that $\Exp(v) \in \cF(\cH)$ and
\begin{gather*} \la \Exp(v), \Exp(w) \ra
= \sum_{n = 0}^\infty \frac{1}{n!} \la v, w\ra^n = e^{\la v, w \ra}.\end{gather*}
This leads to an embedding
\begin{gather*} \Phi \colon  \ \cF(\cH) \to \C^\cH, \qquad
\Phi(\xi)(v)
:= \la \xi, \Exp(v)\ra,\end{gather*}
where
 \begin{gather*} \Phi(v_1 \vee \cdots \vee v_n)(v)
= \frac{1}{\sqrt{n!}} \la v_1 \vee \cdots \vee v_n, v^n \ra
= \frac{1}{\sqrt{n!}} \prod_{j = 1}^n \la v_j, v\ra.\end{gather*}
The image of $\Phi$ is the reproducing kernel space $\cF_{\rm RK}(\cH)$ with kernel
\begin{gather*} K(v,w) = \la \Exp(w), \Exp(v) \ra = e^{\la w, v\ra}.\end{gather*}
We may thus identify $\cF(\cH)$ with
the reproducing kernel Hilbert space $\cF_{\rm RK}(\cH)$.

For a contraction $A \colon \cH \to \cK$, we have
\begin{gather*} \Phi(Av_1 \vee \cdots \vee A v_n)(v)
= \frac{1}{\sqrt{n!}} \la Av_1 \vee \cdots \vee A v_n, v^n \ra
= \frac{1}{\sqrt{n!}} \la v_1 \vee \cdots \vee v_n, (A^*v)^n \ra,\end{gather*}
so that the operator $\Gamma(A)$ acts on the reproducing kernel
space $\cF_{\rm RK}(\cH) \to \cF_{\rm RK}(\cK)$ simply by
\begin{gather*} (\Gamma(A)F)(v) := F(A^*v), \qquad \|A\| \leq 1, \quad F \in \cH_K \subeq \C^\cH, \quad v \in \cK .\end{gather*}

\subsection{Application to ref\/lection positive representations}

Typical examples of gaussian $(G,S,\tau)$-probability spaces arise as follows.
Let $\cH$ be a real Hilbert space
and $(Q,\Sigma,\gamma_\cH)$ be a realization of the gaussian random process
$(\phi(v))_{v \in \cH}$ indexed by $\cH$ (Def\/inition~\ref{def:1.1}).
Using the realization from Example~\ref{ex:1.1}, where $Q = \cH^a$ is the algebraic
dual space, we obtain an action of the orthogonal group $\OO(\cH)$
on $(Q,\Sigma,\mu)$ by measure preserving automorphisms.

For the following proposition, we recall the concept of a ref\/lection
positive representation of $(G,S,\tau)$ from Def\/inition~\ref{def:1.4}.

\begin{Lemma}\label{lem:3.6}
If $\pi \colon S \to B(\cH)$ be is a continuous $*$-representation
of the topological involutive semigroup $(S,\sharp)$ by contractions,
then we obtain on
$\Gamma(\cH) \cong L^2(\cH^a, \gamma_\cH)$
by $P_s := \Gamma(\pi(s))$ a~standard positive semigroup
structure on the probability space $(\cH^a, \gamma_\cH)$.
\end{Lemma}

\begin{proof} Here we use that
\begin{gather*} \eps^2 \gamma_\cH(\{|P_s f - P_{s_0}f| \geq \eps\})
\leq \int_{\cH^a} |P_s f - P_{s_0}f|^2\, d\gamma_\cH \to 0 \end{gather*}
for \looseness=1 $s \to s_0$ follows from the continuity of the representation
$\Gamma \circ \pi$ of $S$ on $\Gamma(\cH)$ (cf.\ Lem\-ma~\ref{lem:6.7}).
\end{proof}

\begin{Proposition}\label{pr:4.6} Let $(U,\cE,\cE_+,\theta)$ be a reflection positive orthogonal
representation of $(G,S,\tau)$ for which $\cE_0$ is $U$-cyclic and $\cE_+$ is generated
by $(U_s \cE_0)_{s\in S}$. Then second quantization leads to a~$(G,S,\tau)$-probability space
$((Q_\cH,\Sigma,\gamma_\cH), \Sigma_0, \Gamma(U),\Gamma(\theta))$,
where $\Sigma_0\subeq \Sigma$ is the smallest $\sigma$-algebra for which
the functions $(\phi(v))_{v \in \cE_0}$ are measurable.
\end{Proposition}

\begin{proof} (GP1-4) Clearly, every $\Gamma(U_g)$ and $\Gamma(\theta)$
are automorphisms of the algebra $L^\infty(Q_\cH,\Sigma,\gamma_\cH)$ satisfying
\begin{gather*} \Gamma(\theta) \Gamma(U_g) \Gamma(\theta) = \Gamma(\theta U_g \theta) =
\Gamma(U_{\tau(g)}) \qquad \mbox{and} \qquad \theta E_0 = E_0.\end{gather*}
The continuity of the unitary representation
$(\Gamma(U_g))_{g \in G}$ on $L^2(Q_\cH,\gamma_\cH)$ implies the continuity
in measure of the $G$-action on $L^\infty(Q_\cH,\gamma_\cH)$
(Lemma~\ref{lem:3.6}).

(GP5) Our def\/inition of $\Sigma_0$ implies that
$\Sigma_+$ is the smallest $\sigma$-algebra for which the functions
$(\phi(v))_{v\in \cE_+}$ are measurable and
since $\cE_0$ is $U$-cyclic in $\cE$, (GP5) is also satisf\/ied.

Ref\/lection positivity of the representation $\Gamma(U)$
of $(G,S,\tau)$ follows from \cite[Remark~3.8]{NO15a} and
$\Gamma(\cE_+) = L^2(Q_\cH,\Sigma_+, \gamma_\cH)$.
\end{proof}

\subsection{Equivalence of gaussians measures for reproducing kernel Hilbert spaces}

Let $X$ be a set and $E = \C[X]$ the free complex vector spaces over $X$.
Then positive def\/inite kernels $K$ on $X$ are in one-to-one correspondence
with positive semidef\/inite hermitian forms on~$E$. Any such kernel
def\/ines a Hilbert subspace $\cH_K \subeq E^*$ with continuous point evaluations.

More generally, we may consider for a real locally convex space $E$
continuous bilinear hermitian kernels $K \colon E \times E \to \C$ and
the corresponding subspaces of the topological dual space~$E'$~\cite{Sch64}.
Such a kernel is positive def\/inite if and only if the
canonical sesquilinear extension to the complexif\/ication $E_\C$ is a
positive semidef\/inite hermitian form.
Suppose that $E$ is nuclear. Then, for any such $K$, the function
\begin{gather*} \vphi_K(v) := e^{-\frac{1}{2} K(v,v)} \end{gather*}
on $E$ is continuous and positive def\/inite, hence is the Fourier transform
of a uniquely determined gaussian measure $\gamma_K$ on $E'$.
We want to express conditions on pairs of kernels $K$ and $Q$ which
characterize the equivalence of the measures $\gamma_K$ and $\gamma_Q$ on
$E'$.

According to \cite{Ka48}, two gaussian measures are either mutually singular
of equivalent. The following theorem is a reformulation of \cite[Theorem~10.1]{Ya85}
(cf.\ also \cite[Theorem~4.1/4.2]{Jo68} or \cite[Corollary~6.4.11]{Bo98}).

\begin{Theorem}\label{thm:gauss-ker-equiv} Let $E$ be a nuclear real locally convex space.
For two continuous positive semidef\/inite hermitian forms $K$ and $Q$ on $E$, the corresponding gaussian measures~$\gamma_K$ and~$\gamma_Q$ on $E'$ are equivalent if and only if $\cH_K = \cH_Q$ and there exists an operator
$T \in \GL(\cH_K)$ for which $TT^* -\1$ is Hilbert--Schmidt and
\begin{gather*} Q(x,y) = \la TK_y, TK_x \ra \qquad \mbox{for} \quad x,y \in E.\end{gather*}
\end{Theorem}

\begin{Remark} (a) That the gaussian measure $\gamma_K$ on $E'$ determines the
Hilbert subspace \mbox{$\cH_K \subeq E'$} follows from \cite[Theorem~9.1]{Ya85} which
asserts that $\cH_K$ consists precisely of those linear functionals
$\alpha \in E'$ for which $\gamma_K$ is quasi-invariant under the translation
$\tau_\alpha(\beta) := \alpha +\beta$. As a~consequence, the equivalence of
$\gamma_K$ and $\gamma_Q$ implies $\cH_K = \cH_Q$.

(b) Recall from \cite[Theorem~I.2.8]{Ne00} that $\cH_K = \cH_Q$ is equivalent to the
existence of positive constants $c_1, c_2 > 0$ such that the kernels
\begin{gather*} 
 K - c_1 Q \qquad \mbox{and} \qquad c_2 Q - K
\end{gather*}
are positive def\/inite. If $K$ and $Q$ are real-valued, this is equivalent to
\begin{gather*} c_1 Q(x,x) \leq K(x,x) \leq c_2 Q(x,x) \qquad \mbox{for} \quad x \in E.\end{gather*}
This in turn is equivalent to the existence of a bounded invertible
positive operator $A \in B(\cH_K)$ such that
\begin{gather*} Q(x,y) = K^A(x,y) := \la AK_y, K_x \ra \qquad \mbox{for} \quad
x,y \in E.\end{gather*}
In view of \cite[Corollary~I.2.6]{Ne00}, we have
$A = TT^*$, where $T \colon \cH_Q \to \cH_K, f \mapsto f$ is the identity.
This implies that
\begin{gather*} T^*(K_x) = Q_x \qquad \mbox{for} \quad x \in E.\end{gather*}
In particular, we have
\begin{gather*} Q(x,y) = \la Q_x, Q_y \ra_{\cH_Q} = \la TT^* K_x, K_y \ra_{\cH_K}.\end{gather*}
We conclude that, for $f \in \cH_Q = \cH_K$, we have
\begin{gather*} \|f\|_{\cH_K}^2 = \|Tf\|_{\cH_K}^2 = \la T^*T f, f \ra_{\cH_Q}.\end{gather*}
Therefore the equivalence of the corresponding gaussian measure
is equivalent to $T^*T - \1$ being Hilbert--Schmidt.
\end{Remark}

\begin{Remark} Let $E$ be a real vector space and endow it with the
f\/inest locally convex topology for which all
seminorms on $E$ are continuous. Then $E$ is nuclear if and only if
$E$ is of at most countable dimension \cite[Proposition~50.1, Theorem~51.2]{Tr67}.
In any case, its topological dual space is $E' = E^*$ because
every linear functional on $E$ is continuous.

Regardless of the nuclearity of $E$,
every positive def\/inite function $\vphi \colon E \to \C$
which is continuous on all f\/inite-dimensional subspace is the Fourier transform
of a $\fB^*$-measure on $E^*$ (Theorem~\ref{thm:dual}).
This applies in particular to all functions of the form
$\vphi(v) := e^{-\frac{1}{2} K(v,v)}$, where $K \colon E \times E \to \R$ is a
positive semidef\/inite symmetric bilinear form on $E$.
\end{Remark}

\subsection{Gaussian measures on distributions on Lie groups} \label{se:4.4}

If $G$ is a Lie group, then $\cD(G) := C^\infty_c(G)$ is an involutive
algebra with respect to the convolution product and the involution
$\vphi^*(g) := \oline{\vphi(g^{-1})} \Delta_G(g^{-1})$,
where $\Delta_G$ is the modular function satisfying
\begin{gather*}
 \Delta_G(y) \int_G f(xy)\, d\mu_G(x) = \int_G f(x)\, d\mu_G(x)
\qquad \mbox{for} \quad f \in C_c(G), y \in G.
\end{gather*}
Accordingly, we call a distribution $D \in \cD'(G)$ (the space of antilinear continuous functionals on~$\cD(G)$) {\it positive
def\/inite}, if it is a positive functional on this algebra, i.e.,
\begin{gather*} 
D(\vphi^* * \vphi) \geq 0 \quad \mbox{for} \quad \vphi \in \cD(G).
\end{gather*}
Since $\cD(G)$ is nuclear, every positive def\/inite distribution
$D \in \cD'(G)$ determines a gaussian measure $\gamma_D$ on $\cD'(G)$.

For a unitary representation $(\pi, \cH)$ of $G$ we
write ${\cal H}^{-\infty}$ for the space of continuous antilinear
functionals on ${\cal H}^\infty$, the space of {\it distribution vectors},
and note that we have a natural linear
embedding $\cH \into \cH^{-\infty}$, $v \mapsto \la v, \cdot \ra$.
Accordingly, we also write $\la \alpha, v \ra
= \oline{\la v, \alpha\ra}$ for $\alpha(v)$,
$\alpha \in \cH^{-\infty}$ and $v \in \cH^\infty$.
The group $G$ acts naturally on $\cH^{-\infty}$ by
\begin{gather*} (\pi^{-\infty}(g)\alpha)(v) := \alpha\big(\pi(g)^{-1}v\big),\end{gather*}
so that we obtain a $G$-equivariant chain of continuous inclusions
\begin{gather}\label{eq:rig}
 {\cal H}^\infty \subeq {\cal H} \subeq {\cal H}^{-\infty}
\end{gather}
(cf.\ \cite[Section~8.2]{vD09}). It is $\cD(G)$-equivariant,
if we def\/ine the representation of $\cD(G)$ on $\cH^{-\infty}$ by
\begin{gather*}
\big(\pi^{-\infty}(\vphi)\alpha\big)(v)
:= \int_G \vphi(g) \alpha\big(\pi(g)^{-1}v\big)\, d\mu_G(g)
= \alpha(\pi(\vphi^*)v).
\end{gather*}

\begin{Proposition}[\protect{\cite[Proposition~2.8]{NO14}}]\label{prop:2.10}
Let $D \in \cD'(G)$ be a positive definite distribution on the
Lie group $G$ and $\cH_D$ be the corresponding
reproducing kernel Hilbert space with kernel
$K(\vphi,\psi) := D(\psi^* * \vphi)$ obtained by completing
$\cD(G) * D$ with respect to the scalar product
\begin{gather*} \la \psi * D, \vphi * D \ra = D(\psi^* * \vphi). \end{gather*}
Then the
following assertions hold:
\begin{itemize}\itemsep=0pt
\item[\rm(i)] $\cH_D \subeq \cD'(G)$ and
the inclusion $\gamma_D \colon \cH_D \to \cD'(G)$ is continuous.
\item[\rm(ii)] We have a unitary representation
$(\pi_D, \cH_D)$ of $G$ by
\begin{gather*} \pi_D(g)E = g_*E, \qquad \mbox{where} \quad
(g_*E)(\vphi) := E(\vphi \circ \lambda_g) \end{gather*}
and the integrated representation
of $\cD(G)$ on $\cH_D$ is given by $\pi_D(\vphi)E = \vphi * E$.
\item[\rm(iii)] There exists a unique distribution
vector $\alpha_D \in \cH_D^{-\infty}$ with
$\alpha_D(\vphi *D) = D(\vphi)$ and
\begin{gather*} \pi^{-\infty}(\vphi)\alpha_D = \vphi * D\qquad \mbox{for} \quad \vphi \in \cD(G).\end{gather*}
\item[\rm(iv)] $\gamma_D$ extends to a $\cD(G)$-equivariant injection
$\cH_D^{-\infty} \into \cD'(G)$ mapping $\alpha_D$ to~$D$.
\end{itemize}
\end{Proposition}

\begin{Remark}We consider a ref\/lection positive distribution vector
$\alpha \in \cH^{-\infty}$ for a unitary representation
$(\pi, \cH)$ of $G$. This leads to an embedding
\begin{gather*} \eta_\alpha \colon \ \cH \to \cD'(G), \qquad
\eta_\alpha(v)(\vphi) := \la v, \pi^{-\infty}(\vphi) \alpha\ra
= \la \pi(\vphi^*)v, \alpha\ra. \end{gather*}
which is injective if and only if $\alpha$ is cyclic.
This establishes a one-to-one correspondence between distribution
vectors and $G$-equi\-variant continuous linear maps
$\cH \to \cD'(G)$ (Proposition~\ref{prop:2.10}). Actually we obtain an equivariant embedding
\begin{gather*} \cH^{-\infty} \into \cD'(G) \end{gather*}
by dualizing the linear
map $\cD(G) \to \cH$, $\vphi \mapsto \pi^{-\infty}(\vphi)\alpha$. This in turn
leads to the positive def\/inite function
\begin{gather*} S(\vphi)
:= e^{-\shalf \|\pi^{-\infty}(\vphi)\alpha\|^2} = e^{-\shalf D(\vphi^* * \vphi)}
\qquad \mbox{for} \quad D(\vphi) :=\alpha(\pi^{-\infty}(\vphi)\alpha),\quad
\phi \in \cD(G).\end{gather*}
We thus obtain a $G$-invariant gaussian probability
measure $\gamma_D$ on $\cD'(G)$ by Minlos' theorem.
\end{Remark}

\begin{Lemma} \label{lem:multbound}
If $\alpha \in \cH^{-\infty}$ is cyclic and $(\rho,V)$ is a finite-dimensional
irreducible representation of $G$, then the multiplicity of $\rho$ in $\cH$ is bounded by
$\dim V$. In particular, the $V$-isotypic subspace of $\cH$ is finite-dimensional
and $\dim \cH^G \leq 1$.
\end{Lemma}

\begin{proof} Let $n \in \N$ and assume that $V^n$ is a subrepresentation of $\cH$.
Then $V^n$ is f\/inite-di\-men\-sional and generated by a distribution vector which actually
must be an element $(v_1, \ldots, v_n)$ $\in V^n$.
Suppose that $\lambda_1, \ldots, \lambda_n\in \C$ satisfy $\sum_j \lambda_j v_j = 0$. Then
$\sum_j \lambda_j \pi(g)v_j = 0$ for every $g \in G$, and hence
$\sum_j \lambda_j w_j = 0$ for every $(w_1, \ldots, w_n) \in V^n$ because
$(v_1,\ldots, v_n)$ is cyclic in $V^n$. This leads to $\lambda_j = 0$ for every~$j$,
so that the elements $v_1,\ldots, v_n$ are linearly independent.
\end{proof}

\begin{Theorem} \label{thm:8.3} Let $G$ be a Lie group and
$D, E \in \cD'(G)$ be positive definite distributions.
Then the corresponding gaussian measures $\gamma_D$ and $\gamma_E$ on
$\cD'(G)$ are equivalent if and only if the following conditions are satisfied
\begin{itemize}\itemsep=0pt
\item[\rm(i)] $D$ can be written
as an orthogonal sum $D = D_0 + \sum_{n \in J} D_n$, where
$J \subeq \N$ and the representation on the subspaces $\cH_{D_n}$, $n \in J$, are
finite-dimensional isotypic and mutually disjoint.
\item[\rm(ii)] $E = D_0 + \sum_{n \in J} E_n$ with
$\cH_{D_n} = \cH_{E_n}$, and there exist intertwining operators
$T_0 = \id_{\cH_{D_0}}$ and $T_n \in B_G(\cH_{D_n})$ with
$T_n D_n = E_n$ in $\cH_{D_n}^{-\infty} = \cH_{D_n}$ and
$\sum_{n \in J} \|T_n T_n^* - \1\|_2^2 < \infty.$
\end{itemize}
\end{Theorem}

\begin{proof} We shall use Theorem~\ref{thm:gauss-ker-equiv}. If $\gamma_D \sim \gamma_E$, then $\cH_D = \cH_E \subeq \cD'(G)$ and
the identity map $T \colon \cH_D \to \cH_E$, $f \mapsto f$
is a $G$-equivariant operator, so that $T^*T \in B(\cH_D)$ is a $G$-intertwining
operator. The requirement that $T^*T - \1$ is Hilbert--Schmidt implies
that its range is a sum of f\/inite-dimensional subrepresentations.
In view of the preceding lemma, it can be written as
$\oplus_{j \in J} \cH_j$, where the $\cH_j$ are isotypic, f\/inite-dimensional and mutually disjoint.
Then $T_n := T\res_{\cH_{D_n}}$ maps $\cH_{D_n}$ into itself and
\begin{gather*} \|T T^*- \1\|_2^2 = \sum_{n \in J} \|T_n T_n^*- \1\|_2^2.\end{gather*}
The converse implication follows from Theorem~\ref{thm:gauss-ker-equiv}.
\end{proof}

\begin{Corollary} If $D$ is a positive definite distribution on $G$, then the following
are equivalent
\begin{itemize}\itemsep=0pt
\item[\rm(a)] $\cH_D$ contains no $G$-invariant subspace of finite positive dimension.
\item[\rm(b)] For any two different distribution vectors
$E,E' \in \cH_D^{-\infty} \subeq \cD'(G)$,
the corresponding gaussian measures $\gamma_{E}$ and $\gamma_{E'}$ are inequivalent.
\end{itemize}
\end{Corollary}

\begin{Example} \label{ex:8.5} We discuss the special case $G = \R^n$. According to
the Bochner--Schwartz theorem, a distribution $D \in \cD'(G)$ is positive
def\/inite if and only if it is the Fourier transform $D = \hat\mu$ of a tempered
measure $\mu$ on the dual group $\hat G \cong \R^d$, and then
$\cH_D \cong L^2(\R^d,\mu)$ with the representation
\begin{gather*} (\pi_D(x)f)(y) = e^{-ixy} f(y).\end{gather*}

For $D = \hat\mu$ and $E = \hat\nu$, the equality of the corresponding
Hilbert spaces is equivalent to the existence of positive constants
$c_1, c_2 > 0$ with
\begin{gather*} c_1 \nu \leq \mu \leq c_2 \nu,\end{gather*}
which \looseness=-1 is equivalent to the equivalence of the measures $\mu$ and $\nu$
with the additional requirement that $\delta := \frac{d\mu}{d\nu}$
satisf\/ies $0 < c_1 \leq \delta \leq c_2$ (cf.~\cite{Ne00} or \cite{Jo68}). Then
$L^2(\R^d,\mu) = L^2(\R^d,\nu)$, and the identity
$T \colon L^2(\R^d,\nu) \to L^2(\R^d,\mu)$ is an isomorphism of Banach spaces.
We then have $T^*(f) = \delta f$ and the equivalence of the gaussian measures
is equivalent to the multiplication operator
$M_{\delta - 1} = T^*T - \1$ being a Hilbert--Schmidt operator on $L^2(\R^d, \nu)$
(Theorem~\ref{thm:gauss-ker-equiv}).
This is equivalent to the condition that the
restriction of $\nu$ to the subset $\{\delta \not= 1\}$ is atomic,
so that the values of $\delta$ in these points are def\/ined, and
the Hilbert--Schmidt condition can be expressed as
\begin{gather*} \sum_{\delta(x) \not=1} |\delta(x) -1|^2 < \infty.\end{gather*}
We conclude that $\gamma_D\sim \gamma_E$ is equivalent to
$\mu = \nu$ on the complement of an at most countable set
$S$ of atoms for both measures, which satisf\/ies
\begin{gather*} \sum_{s \in S} \left|\frac{\mu(\{s\})}{\nu(\{s\})} - 1\right|^2 < \infty.\end{gather*}
\end{Example}

\begin{Remark} (a) Theorem~\ref{thm:8.3} generalizes in an obvious way to
continuous positive def\/inite functions on a topological group~$G$.

(b) Example~\ref{ex:8.5} generalizes in the obvious fashion to
positive def\/inite functions on a locally compact abelian group, or,
more generally, on a nuclear abelian group (cf.\ \cite{Ba91}).
\end{Remark}

The following theorem covers in particular the case of one-dimensional
Lie groups.

\begin{Theorem} \label{thm:3.18} Let $A$ be a selfadjoint operator
on the Hilbert space $\cH$.
Then the following are equivalent:
\begin{itemize}\itemsep=0pt
\item[\rm(a)] The gaussian measure $\gamma_\cH$ can be realized on $\cH^{-\infty}(A)$,
the dual space of
\begin{gather*} \cH^\infty(A) = \bigcap_{n \in \N_0} \cD\big(A^n\big).\end{gather*}
\item[\rm(b)] There exists an $N \in \N$ such that the bounded operator
$(\1 + A^2)^{-N}$ is Hilbert--Schmidt.
\item[\rm(c)] The Fr\'echet space $\cH^\infty(A)$ is nuclear.
\end{itemize}
\end{Theorem}

\begin{proof} (a) $\Rarrow$ (b)
Let $E \subeq \R$ be a bounded subset,
$\cH_1 := P(E)\cH$ (for the spectral measure~$P$ of~$A$) and
$\cH_2 := \cH_1^\bot$. Then we accordingly have
$A = A_1 \oplus A_2$, where the operator $A_1$ is bounded
and $\cD(A) = \cH_1 \oplus \cD(A_2)$. This implies that
\begin{gather*} \cH^\infty(A) = \cH_1 \oplus \cH_2^\infty(A_2) \qquad \mbox{and thus} \qquad
\cH^{-\infty}(A) = \cH_1 \oplus \cH_2^{-\infty}(A_2).\end{gather*}
Therefore $\gamma_\cH$ can be realized on $\cH^{-\infty}$ if and only if
$\gamma_{\cH_1}$ can be realized on $\cH_1$, which means that $\cH_1$ is
f\/inite-dimensional, and $\gamma_{\cH_2}$ can be realized on $\cH_2^{-\infty}(A_2)$.

Therefore (a) implies that all spectral projections $P([a,b])$, $a < b$, have
f\/inite-dimensional range. As a consequence,
$A$ is diagonalizable
with discrete spectrum and f\/inite-dimensional eigenspaces; in particular
$\cH$ is separable. Let $(\lambda_n)_{n \in \N}$ be the eigenvalues of $A$,
counted with multiplicities and $(e_n)_{n \in \N}$ an ONB of $\cH$ with
$A e_n = \lambda_n e_n$ for every $n \in \N$. Now
\begin{gather*} \cH^{-\infty}(A)
= \bigg\{ \sum_n x_n e_n \colon (\exists\, N \in \N)
\sum_n \big(1 + \lambda_n^2\big)^{-2N} |x_n|^2 < \infty \bigg\}\end{gather*}
is the union of the subspaces
\begin{gather*} \cH^{-2N}(A)
= \bigg\{ \sum_n x_n e_n \colon \sum_n \big(1 + \lambda_n^2\big)^{-2N} |x_n|^2 < \infty \bigg\}.\end{gather*}

We realize the gaussian measure $\gamma_\cH$ on the product space $\C^\N$.
Then every subspace $\cH^{-2N}(A)$ is measurable, and
\begin{gather*} 1 = \gamma_\cH\big(\cH^{-\infty}(A)\big) = \lim_{N \to \infty} \gamma_\cH\big(\cH^{-2N}(A)\big) \end{gather*}
implies that $\gamma_\cH(\cH^{-2N}(A))>0$ for some $N \in \N$.
From the example in \cite[p.~153]{Ya85}, it now follows that
\begin{gather*} \big\|\big(\1 + A^2\big)^{-N}\big\|_2^2 = \sum_n \big(1 + \lambda_n^2\big)^{-2N} < \infty\end{gather*}
(see also \cite[Theorem~5.2]{Dr03}).

(b) $\Leftrightarrow$ (c): The space
$\cD^\infty(A) := \bigcap_{n \in \N} \cD(A^n)$ coincides with the space
$\cD^\infty(B)$ for $B := \1 + A^2$,
and $B$ has the additional property that $B^n \leq B^{n+1}$
for $n \in \N_0$. The topology on $\cD^\infty(B)$ is def\/ined by the seminorms
$p_n(v) := \|B^n v\|$, $n \in \N_0$.
In view of \cite[Def\/inition~50.1]{Tr67}, the nuclearity of the
space $\cH^\infty = \cD^\infty(B)$ is equivalent to the condition that,
for every $n \in \N$,
there exists an $m > n$, such that the natural map $\cD^m(B) \to \cD^n(B)$ is
nuclear.
Since this map can be identif\/ied with the map $B^{n-m} \colon \cH \to \cH$, we see that
the nuclearity of $\cH^\infty$ is equivalent to the existence of some nuclear
power of~$B^{-1}$,
which is equivalent to the existence of some power which is Hilbert--Schmidt.
This means
that $\cD^\infty(A)$ is nuclear if and only if some operator $(\1 + A^2)^{-N}$,
$N \in \N$, is Hilbert--Schmidt.

(c) $\Rarrow$ (a) follows from the Bochner--Minlos theorem.
\end{proof}

\begin{Example} (a) In the context of ref\/lection positivity on curved spacetimes,
a natural class of selfadjoint operators arises as follows \cite{JR07a, JR07b}.
We call a connected complete orientable Riemannian manifold $(M,g)$ a
{\it quantizable static space-time} if there exists a complete Killing vector f\/ield $\xi$
orthogonal to a hypersurface $\Sigma \subeq M$ such that all integral curves
of $\xi$ intersect $\Sigma$ exactly once. Then the f\/low of $\xi$
induces an isometry $\Phi \colon \R \times \Sigma \to M$ of Riemannian manifolds
and $\theta(\Phi(t,x)) = \Phi(-t,x)$ is an isometric involution exchanging the
two open subsets $\Omega_\pm = \{ \Phi(t,x) \colon {\pm}t > 0, x \in \Sigma\}$.
Let $\Delta$ be the Laplacian of $(M,g)$ and $m > 0$. Then the {\it free covariance operator} $C := (m^2 - \Delta)^{-1}$ is a bounded selfadjoint operator on $L^2(M)$.
For $s \in \R$, we write $\cH_s$ for the completion of $L^2(M)$ w.r.t.\ the scalar product
$\la f,g \ra := \la f, C^{-s} g \ra_{L^2(M)}$
(the Sobolev spaces) and obtain a Fr\'echet space
\begin{gather*} \cS := \cD^\infty\big(m^2 - \Delta\big) = \cD^\infty(\Delta) = \bigcap_{s \in \R} \cH_s,\end{gather*}
but in general this space is not nuclear because the Laplacian may have continuous spectrum (which contradicts the statement in~\cite[Def\/inition~2.1]{JR07b}).
But if $\Sigma$ is compact and
\begin{gather*} H := \frac{1}{2} + \frac{1}{2}\big(Q_t^2 - \Delta\big), \end{gather*}
where $Q_t$ is the multiplication with the global time function, then
there exists a power $(\1 + H)^{-N}$ which is trace class, so that
\begin{gather*} \cS := \cD^\infty(H) \subeq L^2(M) \end{gather*}
is a nuclear space. The operator $H$ is the sum of the Hamiltonian of the harmonic oscillator in the time direction and the Laplacian of~$\Sigma$
\cite[Proposition~7.5]{An13}. We refer to \cite[Section~7.4]{An13} for a~discussion of the corresponding gaussian measures on the dual space~$\cS'$.

(b) For $M = \R^d$ and $\cH = L^2(\R^d)$, the Hamiltonian of the $d$-fold harmonic oscillator
\begin{gather*} H = \frac{1}{2}\sum_{j = 1}^d Q_j^2 - \partial_j^2, \qquad
(Q_j f)(x) = x_j f(x), \qquad (\partial_j f)(x) = \frac{\partial}{\partial x_j} f(x) \end{gather*}
leads to $\cD^\infty(H) = \cS(\R^d)$ which also is a nuclear space.
\end{Example}

\begin{Corollary} \label{con:3.18} Let $(\pi, \cH)$ be a continuous unitary representation of the finite-dimensional
Lie group $G$ and put $\Delta = \sum_j X_j^2$ for a basis $X_1, \ldots, X_n$ of $\g$. Then the following are equivalent:
\begin{itemize}\itemsep=0pt
\item[\rm(a)] The gaussian measure $\gamma_\cH$ can be realized on $\cH^{-\infty}$.
\item[\rm(b)] There exists an $N \in \N$ such that $(\1-\oline{\dd\pi(\Delta)})^{-N}$ is a
Hilbert--Schmidt operator.
\item[\rm(c)] The Fr\'echet space $\cH^\infty$ is nuclear.
\item[\rm(d)] $\pi$ is trace class, i.e., for each $f \in C^\infty_c(G)$, the operator
$\pi(f)$ is trace class.
\end{itemize}
\end{Corollary}

\begin{proof} In view of Theorem~\ref{thm:3.18}, the equivalence of (a)-(c) follows from Nelson's theorem asserting that
$\cH^\infty = \cH^\infty(A)$ holds for the selfadjoint operator $A := \oline{\dd\pi(\Delta)}$~\cite[Theorem~4.4.4.5]{Wa72}.
For the equivalence with~(d) we refer to \cite[Proposition~1.11]{DNSZ16}.
\end{proof}

\begin{Proposition}\label{prop:4.20} Let $(\pi, \cH)$ be a unitary representation of the Lie group~$G$. Then the space~$\cH^\infty$ is nuclear in the following cases:
\begin{itemize}\itemsep=0pt
\item[\rm(a)] If $G$ is compact and $\pi$ is distribution cyclic.
\item[\rm(b)] If $G$ is connected semisimple and $\pi$ is irreducible.
\item[\rm(c)] If $G$ is connected nilpotent and $\pi$ is irreducible.
\end{itemize}
\end{Proposition}

\begin{proof} (a) (Sketch) We denote irreducible representation with highest weight $\lambda$ by $(\pi_\lambda, V_\lambda)$. We also set $d(\lambda):=\dim V_\lambda$. We have $\cH = \oplus_{\lambda} \cH_\lambda$, where $\cH_\lambda$ is the isotypic subspace of highest weight~$\lambda$.
The multiplicity of the simple $G$-module $V_\lambda$ in $\cH_\lambda$ is bounded by $d(\lambda)$ (Lemma~\ref{lem:multbound}).
We can then choose the basis $X_j$ such that $-\Delta$ coincides with the Casimir element of $\g$. Let $\rho = {1\over 2} \sum_{\beta \in \Delta^+} \beta$.
Then
\begin{gather*} -\dd \pi (\Delta)\res_{\cH_\lambda} = \big(\|\lambda + \rho\|^2 - \|\rho\|^2\big) \1 \end{gather*}
by Freudenthal's lemma \cite[Lemma IX.5.2]{Ne00}. We therefore have
\begin{gather*}
\begin{split}
& \big\| (\1-\dd \pi (\Delta ) )^{-N}\big\|_2^2
= \sum_\lambda \big(1 + \|\lambda + \rho\|^2 - \|\rho\|^2\big)^{-2N} \dim \cH_\lambda\\
&\hphantom{\big\| (\1-\dd \pi (\Delta ) )^{-N}\big\|_2^2}{} \leq \sum_\lambda \big(1 + \|\lambda + \rho\|^2 - \|\rho\|^2\big)^{-2N} d(\lambda )^2 .
\end{split}
\end{gather*}
This expression is f\/inite for some $N \in \N$ because $d(\lambda )$ is bounded by a~polynomial in~$\|\lambda\|$ \cite[Lemma~4.4.2.3]{Wa72}.
This means that there exists an $N \in \N$ for which $(\1 -\dd \pi (\Delta ) )^{-N}$ is Hilbert--Schmidt. Therefore $\cH^\infty = \cD^\infty(\dd \pi (\Delta )) = \cD^\infty((\1 -\dd \pi (\Delta) )^N)$ is nuclear.

(b) (\cite[Theorem~2.1]{DD16}) Write $\fg=\fk \oplus \fs$ for the Cartan decomposition of $\fg$. Then we can choose the basis $X_j$ such that
$\Delta =\Omega +2\Delta_{\fk}$ where $\Omega $ is the Casimir element and $\Delta_\fk$ is the Laplacian for $\fk$. As $\Omega$ acts by a scalar and the dimension of $\cH_\lambda$ is bounded by $d(\lambda )^2$, the claim follows as in~(a).

(c) (\cite[Proposition~1.9(a)]{DD16}) According to \cite[Theorem~3, p.~103]{Ki04}, we can realize $\pi$ in $L^2(\R^n,dx)$ such that $\cH^\infty$ coincides with the Schwartz space $\mathcal{S}(\R^n)$ which is nuclear.
\end{proof}

\begin{Example} For the distribution $D = \delta_\1$ on the compact Lie group $G$, the corresponding representation is the regular representation on $\cH_D \cong L^2(G)$. In this case $L^2(G)^\infty = C^\infty(G)$ (f.i.\ by the Dixmier--Malliavin theorem) is a nuclear space and $L^2(G)^{-\infty} = C^{-\infty}(G)$ is the space of distributions on~$G$. In particular, the measure $\gamma_D$ can be realized in this space.
\end{Example}

\begin{Remark} The condition $\cH^\infty = \cH$ is equivalent to the smoothness of the representation, which in turn is equivalent to the boundedness of the operators $\dd\pi(X)$, $X \in \g$. If this is the case, then $\cH^{-\infty} = \cH$. If $\cH$ is inf\/inite-dimensional it is not nuclear, and since $\cH \subeq \cH^a$ (the algebraic dual) is a zero set for $\gamma_\cH$, the corresponding gaussian measure can not be realized on~$\cH^{-\infty}$.

Examples arise from the multiplication representation of an abelian Lie group $G$ on the space $\cH = L^2(\hat G, \mu)$, where $\mu$ is a compactly supported measure on the dual group $\hat G = \Hom(G,\T)$. These representations are smooth and generated by the cyclic vector $1$. If $\cH$ is inf\/inite-dimensional, then it is not nuclear. In particular, $\gamma_D$ is not realized on the subspace $\cH_D^{-\infty}$ of $\cD'(G)$.
\end{Remark}

\begin{Example} The multiplication representation of $G = \R$ on $L^2(\R,\mu)$ for the measure $\mu = \sum_{n \in \N} \frac{1}{2^n} \delta_{\frac{1}{n}}$ is norm continuous and $\cH^\infty= \cH$ is not nuclear.
\end{Example}

\subsection[Ergodicity of $\gamma_\cH$ for $G$]{Ergodicity of $\boldsymbol{\gamma_\cH}$ for $\boldsymbol{G}$}

Let $(\pi, \cH)$ \looseness=-1 be an orthogonal representation of $G$ on the real Hilbert space $\cH$ and $\gamma_\cH$ the corresponding gaussian measure. We consider the so-called {\it gaussian action} of $G$ on the gaussian probability space $(\cH^a, \gamma_\cH)$ on the algebraic dual~$\cH^a$. The measure $\gamma_\cH$ is $G$-ergodic if and only if
\begin{gather*}\Gamma(\cH)^G = L^2(\cH^a, \gamma_\cH)^G = \C 1.\end{gather*}
In this section we derive a criterion for this condition to be satisf\/ied.

\begin{Lemma} \label{lem:fixv}
Let $(\pi_j, \cH_j)_{j = 1,2}$ be two orthogonal representations of the group $G$.
\begin{itemize}\itemsep=0pt
\item[\rm(i)] If $\cH_1$ has no non-zero finite-dimensional subrepresentations, then all finite-dimensional invariant subspaces of $\cH_1 \otimes \cH_2$ are zero. In particular, $(\cH_1 \otimes \cH_2)^G = \{0\}$.
\item[\rm(ii)] The subspace $(\cH_1 \otimes \cH_2)_f$ generated by the finite-dimensional invariant subspaces coincides with $\cH_{1,f} \otimes \cH_{2,f}$.
\end{itemize}
\end{Lemma}

\begin{proof} (i) First we show that the subspace $(\cH_1 \otimes \cH_2)^G$ of f\/ixed vectors is trivial. Since $(\cH_1 \otimes \cH_2)^G$ can be identif\/ied with the space of Hilbert--Schmidt intertwining operators $A \colon \cH_2^* \to \cH_1$, any such operator leads to the self intertwining operator $AA^* \in B_2(\cH_1)$ and its eigenspaces are f\/inite-dimensional $G$-invariant subspaces, hence trivial. Let $\cF \subeq \cH_1 \otimes \cH_2$ be a
f\/inite-dimensional invariant subspace. Then $\id_\cF \in \cF \otimes \cF^* \subeq \cH_1 \otimes \cH_2 \otimes\cH_1^* \otimes \cH_2^*$ is a f\/ixed vector, so that $\cF = \{0\}$ follows from the preceding argument.

(ii) We write
\begin{gather*} \cH_1 \otimes \cH_2 = (\cH_{1,f} \otimes \cH_{2,f})
\oplus \big(\cH_{1,f}^\bot \otimes \cH_2\big) \oplus \big(\cH_{1,f} \otimes \cH_{2,f}^\bot\big) \end{gather*}
and apply the preceding proposition to see that only the f\/irst summand contains non-zero f\/inite-dimensional invariant subspaces.
\end{proof}

For unitary representations, one can also introduce the terminology from measure preserving actions on a probability space,
where $\cH \cong L^2_0(X,\Sigma,\mu)
= \{ f \in L^2(X,\Sigma,\mu)\colon \int_X f\, d\mu = 0\}$.

\begin{Definition} \label{def:3.27} Let $(\pi, \cH)$ be a unitary
representation and $(\pi^*, \cH^*)$ the dual representation.
We say that $\pi$ is:
\begin{itemize}\itemsep=0pt
\item[\rm(a)] {\it ergodic} if $\cH^G = \{0\}$,
\item[\rm(b)] {\it weakly mixing} if $\pi \otimes\pi^*$ is ergodic,
\item[\rm(c)] {\it mixing} if $G$ is locally compact and all matrix coef\/f\/icients are contained in $C_0(G)$.
\end{itemize}
\end{Definition}

The following proposition is an elaboration of the main result of
\cite{SeI57} whose main focus is the equivalence of (i) and (ii).

\begin{Theorem} \label{thm:3.28} {\rm(I.E.~Segal)} For an orthogonal representation
$(\pi, \cH)$ of the group $G$, the following are equivalent:
\begin{itemize}\itemsep=0pt
\item[\rm(i)] $\cH$ contains no non-zero finite-dimensional invariant subspaces.
\item[\rm(ii)] The gaussian measure $\gamma_\cH$ is $G$-ergodic.
\item[\rm(iii)] $\pi$ is weakly mixing.
\item[\rm(iv)] For every orthogonal representation $(\rho, \cK)$, the
representation $\pi \otimes \rho$ is ergodic.
\end{itemize}
If these conditions are satisfied, then, for every $N \in \N$,
the product measure $\gamma_\cH^{\otimes N}$ is also ergodic.
\end{Theorem}

\begin{proof} (i) $\Leftrightarrow$ (ii):
(cf.\ \cite[Proposition~A.1.12, Corollary~A.7.15]{BHV08})\footnote{We thank Bachir Bekka for this reference.}
If $\cF \subeq \cH$ is a f\/inite-dimensional invariant subspace,
then $\gamma_\cF$ (which is equivalent to Lebesgue measure on $\cF$)
is not ergodic. Now
$\gamma_\cH \cong \gamma_\cF \otimes \gamma_{\cF^\bot}$ implies
that $\gamma_\cH$ is not ergodic.

If, conversely, all f\/inite-dimensional invariant subspaces of $\cH$ are trivial,
then all $G$-f\/ixed vectors in $S^n(\cH) \subeq \cH^{\otimes n}$ are trivial if $n > 0$
(Lemma~\ref{lem:fixv}). Hence the assertion follows from the $G$-equivariant decomposition
$\Gamma(\cH) \cong \C \Omega \oplus \bigoplus_{n > 0} S^n(\cH)$.

(i) $\Rarrow$ (iv) follows from Lemma~\ref{lem:fixv}(i).

(iv) $\Rarrow$ (iii) is trivial.

(iii) $\Rarrow$ (i): If $\cF \subeq \cH$ is a f\/inite-dimensional invariant subspace,
then $\id_\cF \subeq \cF \otimes \cF^* \cong \cH \otimes \cH^*$ is a f\/ixed vector.

If (i) is satisf\/ied, then
the canonical representation on $\cH^N$ also contains no non-zero
f\/inite-dimensional invariant subspace, so that the
$G$-action on the product spaces $\Gamma(\cH^N) \cong \Gamma(\cH)^N$
with the product measure $\gamma_{\cH^N} \cong \gamma_{\cH}^{\otimes N}$
is also ergodic.
\end{proof}

\begin{Remark} Let $\cH$ be a complex Hilbert space and $(\pi, \cH)$ be a unitary
representation
of $G$ on~$\cH$. We write $\cH^\R$ for the underlying real Hilbert space and
$\cF_\C(\cH)$ for the associated Fock space over~$\C$.
Then
\begin{gather*} \Gamma\big(\cH^\R\big) \cong \cF_\C\big(\big(\cH^\R\big)_\C\big) \cong \cF_\C(\cH \oplus \oline\cH)
\cong \cF_\C(\cH) \otimes_\C \cF_\C(\oline \cH)
\cong B_2(\cF_\C(\cH)).\end{gather*}
Therefore the requirement that $\cF_\C(\cH)^G = \C \Omega$ is weaker than the ergodicity
of the measure $\gamma_\cH$.

Let $\cH = \cH_f \oplus \cH_f^\bot$ denote the decomposition into the closed subspace
$\cH_f$ generated by all f\/inite-dimensional invariant subspaces and its orthogonal complement
$\cH_f^\bot$.\footnote{Note that the representation on $\cH_f$ factors through a
representation of a compact group.}
Then
\begin{gather*} \cF_\C(\cH) \cong \cF_\C(\cH_f) \otimes \cF_\C\big(\cH_f^\bot\big), \end{gather*}
and Lemma~\ref{lem:fixv} implies that the subspace $\cF_\C(\cH_f) \otimes \bigoplus_{n > 0} S^n(\cH_f^\bot)$
contains no non-zero f\/ixed vectors. Therefore
$\cF_\C(\cH)^G = \cF_\C(\cH_f)^G$.

If $G$ is abelian, then $\cH_f$ is spanned by eigenvectors for certain characters
$X \subeq \hat G$, and the corresponding characters of $S^n(\cH_f)$ are the f\/inite products
$\chi_1 \cdots \chi_n$, $\chi_j \in X$. Therefore $\cF_\C(\cH)^G = \C \Omega$ is equivalent to the
condition that all products
$\chi_1 \cdots \chi_n$, $\chi_j \in X$, $n > 0$, are non-trivial.
\end{Remark}

\begin{Definition}[\protect{\cite[Def\/inition~2.14]{BM00}}] A measure preserving action of $G$ on a
f\/inite measure space $(X,\Sigma,\mu)$ is said to be {\it weakly mixing} if the representation
on the subspace
\begin{gather*} L^2_0(X,\Sigma,\mu)_0 := 1^\bot \subeq L^2(X,\Sigma,\mu) \end{gather*}
contains no non-zero f\/inite-dimensional invariant subspaces. In view of Theorem~\ref{thm:3.28}, this is equivalent to the representation on $L^2_0(X,\Sigma,\mu)$ to be weakly mixing in the sense of Def\/ini\-tion~\ref{def:3.27}.
\end{Definition}

The following proposition justif\/ies Def\/inition~\ref{def:3.27}(b).
\begin{Proposition} A measure preserving action of $G$ on $(X,\Sigma,\mu)$ is weakly mixing if and only
if the corresponding action on the pair space $(X^2, \Sigma \otimes \Sigma, \mu \otimes \mu)$
is ergodic.
\end{Proposition}

\begin{proof} We have
\begin{gather*} L^2(X \times X, \mu \otimes \mu)^G
\cong (L^2(X,\mu) \otimes L^2(X,\mu))^G
\subeq L^2(X,\mu)_f \otimes L^2(X,\mu)_f.\end{gather*}
If the action is weakly mixing, then $L^2(X,\mu)_f = \C 1$ implies
that $L^2(X \times X, \mu \otimes \mu)^G = \C 1$, so that the product action is ergodic.

If, conversely, the product action is ergodic and $\cF \subeq L^2_0(X,\mu)$ is a f\/inite-dimensional invariant subspace, then
\begin{gather*} \cF \otimes \cF^* \subeq L^2(X,\mu) \otimes L^2(X,\mu) \cong
L^2(X \times X, \mu \otimes \mu) \end{gather*}
leads to a f\/ixed vector in $L^2(X \times X, \mu \otimes \mu)$, which implies
$\cF =\{0\}$. We conclude that the action on $(X,\mu)$ is weakly mixing.
\end{proof}

\begin{Example} If $(\pi, \cH)$ is an orthogonal representation
for which the corresponding
{\it gaussian action} on $\Gamma(\cH) = (\cH^a, \gamma_\cH)$
is ergodic, then it is weakly mixing by Theorem~\ref{thm:3.28}.
\end{Example}

\appendix

\section{Continuity of a stochastic process on $G$}
\label{subsec:6.3}

This appendix refers
to Example~\ref{ex:2.9}. Clearly, the most natural continuity requirement
from the perspective of representation theory is that the representation of
$G$ in $L^2(B^G, \fB^G, \nu)$ is continuous. In this subsection we collect some remarks
that are useful for the verif\/ication of this continuity.

\begin{Lemma} \label{lem:6.7} For a semigroup
$(P_s)_{s \in S}$ of positivity preserving operators on $L^\infty(X,\fS,\mu)$,
the strong continuity of the representation of $S$ on
$L^2(X,\fS,\nu)$ implies continuity in measure
\begin{gather*} \lim_{s \to s_0} \nu(|P_s f - P_{s_0} f| \geq \eps) = 0 \qquad \mbox{for every} \quad \eps > 0. \end{gather*}
\end{Lemma}

\begin{proof} This follows from
$\eps^2 \nu(|P_s f - P_{s_0} f| \geq \eps^2)
\leq \int_{X} |P_s f - P_{s_0} f|^2\, d\nu \to 0.$
\end{proof}

\begin{Corollary} \label{cor:6.7} For a square integrable stationary $\R$-valued
process $(X_g)_{g \in G}$, the continuity of the representation of $G$ on
$L^2(B^G, \fB^G, \nu)$ implies that
\begin{gather*} \lim_{g\to \1} \nu(|X_g - X_\1| \geq \eps) = 0 \quad \mbox{
for every} \quad \eps > 0. \end{gather*}
\end{Corollary}

\begin{Remark} Let $G$ be a group acting in a measure preserving way on the f\/inite
measure space $(Q,\Sigma, \mu)$. It is easy to see that the continuity of the representation
on $L^2(Q,\Sigma,\mu)$ is equivalent to the continuity of the orbit maps of the
characteristic functions, which in turn is equivalent to the continuity of the maps
\begin{gather*} d_A \colon \ G \to \R, \qquad d_A(g) := \mu( (gA) \Delta A) \qquad \mbox{for} \quad A \in \Sigma.\end{gather*}
Actually this condition is equivalent to the continuity of the $G$-action on the metric
space $(\Sigma/J_\mu, d)$, where $d(A,B) = \mu(A \Delta B)$ and $J_\mu$ denotes the ideal of
$\mu$-zero sets.

Next we observe that the set of all bounded functions
$f \in L^2(Q,\Sigma,\mu)$ for which the $G$-orbit map in $L^2(Q,\Sigma, \mu)$ is continuous
is a subalgebra. This follows immediately from the estimate
\begin{gather*} \|U_g(fh)-fh\|_2 \leq \|h\|_\infty \|U_g f - f\|_2 + \|f\|_\infty \|U_g h - h\|_2.\end{gather*}
This implies that
\begin{gather*} \Sigma_c := \{ A \in \Sigma \colon d_A \in C(G,\R) \} \end{gather*}
contains $\varnothing$ and $Q$ and is stable under complements, f\/inite intersections and f\/inite unions.
{}From the closedness of the subspace of continuous vectors in $L^2(Q,\Sigma,\mu)$ we further
derive that $\Sigma_c$ is stable under countable unions, hence a $\sigma$-subalgebra.
We conclude that it suf\/f\/ices to verify the continuity of the functions $d_A$ for a collection
of subsets generating the $\sigma$-algebra~$\Sigma$.
\end{Remark}

\begin{Lemma} For a real-valued stationary process $(X_g)_{g \in G}$, the condition
\begin{gather*} \lim_{g\to \1} \nu(|X_g - X_\1| \geq \eps) = 0 \end{gather*}
of continuity in measure implies continuity of the $G$-representation on
$L^2(Q,\Sigma,\nu)$.
\end{Lemma}

\begin{proof} In view of the preceding remark, it suf\/f\/ices to show that, for the sets
$A_{g_0} := \{X_{g_0} \geq a\}$, $a \in \R$, the function $d_A(g)
= \nu(gA_{g_0} \Delta A_{g_0}) = \nu(A_{gg_0} \Delta A_{g_0})$ is continuous in~$\1$. Note that
\begin{gather*} A_{gg_0} \Delta A_{g_0} \subeq \{ |X_{gg_0} -X_{g_0}| \geq \delta\}
\cup \{ X_{gg_0} \in ]a,a+\delta[\} \cup \{ X_{g_0} \in ]a,a+\delta[\}.\end{gather*}
For $\delta$ suf\/f\/iciently small, the last two sets on the right have measure
at most $\frac{\eps}{3}$, so that
\begin{gather*} \nu(\{ |X_{g_0} -X_{gg_0}| \geq \delta\}) < \frac{\eps}{2}
\qquad \mbox{leads to} \quad \nu(gA_{g_0} \Delta A_{g_0}) < \eps.\tag*{\qed}\end{gather*}
\renewcommand{\qed}{}
\end{proof}

\section{Markov kernels}
\label{subsec:6.4}

In this appendix we discuss brief\/ly some basic properties of Markov kernels that are needed for this article.

\begin{Definition}[\protect{\cite[Section~36]{Ba96}}]
(a) Let $(Q,\Sigma)$ and $(Q',\Sigma')$ be measurable spaces. Then a~function
\begin{gather*} K \colon \ Q \times \Sigma' \to [0,\infty] \end{gather*}
is called a {\it kernel} if
\begin{enumerate}
\item[\rm(K1)] for every $A' \in \Sigma'$, the function $K^{A'}(\omega) := K(\omega, A')$ is
$\Sigma$-measurable, and
\item[\rm(K2)] for every $\omega \in Q$, the function $K_\omega(A') := K(\omega, A')$
is a (positive) measure.
\end{enumerate}
A kernel is called a {\it Markov kernel} if the measures $K_\omega$ are probability measures.

(b) A kernel $K \colon Q \times \Sigma' \to [0,\infty]$ associates to a measure
$\mu$ on $(Q,\Sigma)$ the measure
\begin{gather*} (\mu K)(A') := \int \mu(d\omega) K(\omega,A')\end{gather*}
on $(Q',\Sigma')$.
To every measurable function $f' \colon Q' \to [0,\infty]$, it associates the function
\begin{gather*} Kf' \colon \ Q \to [0,\infty], \qquad (Kf')(\omega)
:= \int_{Q'} K(\omega, d\omega') f'(\omega').\end{gather*}
Now the Markov property corresponds to $K1 = K$ for the constant function~$1$.

(c) If $(Q_j, \Sigma_j)_{j=1,2,3}$ are measurable spaces, then
composition of kernels $K_1$ on $Q_1 \times \Sigma_2$ and
$K_2$ on $Q_2 \times \Sigma_3$ is def\/ined by
\begin{gather*} (K_1 K_2)(\omega_1, A_3) = \int_{Q_2}
K_1(\omega_1, d\omega_2) K_2(\omega_2, A_3) .\end{gather*}

In particular, we obtain on a measurable space $(Q,\Sigma)$ the concept of a {\it semigroup
$(P_s)_{s \in S}$ of $($Markov$)$ kernels} by the requirement that $P_s P_t = P_{st}$ for
$st \in S$. Here the classical case is $S = \R_+$.
\end{Definition}

\begin{Remark} \label{rem:b.2}
(a) For $Q = Q'$ and $\Sigma = \Sigma'$, every Markov kernel $K$ def\/ines a positivity
preserving operator on measurable functions by
\begin{gather*} (Kf)(\omega)
:= \int_Q K(\omega, d\omega') f(\omega').\end{gather*}
The Markov property implies that $K1 = 1$ and that
$\|Kf\|_\infty \leq \|f\|_\infty.$

(b) For a measure $\mu$ on $(Q,\Sigma)$ we then have
\begin{gather*} \int_Q f\, d(\mu K)
= \int_Q \int_Q \mu(d\omega) K(\omega, d\omega') f(\omega')
= \int_Q Kf d\mu.\end{gather*}
Therefore the relation $\mu K = \mu$ is equivalent to
the invariance of $\mu$ as a functional on non-negative bounded measurable functions
under the operator $f \mapsto Kf$.
\end{Remark}

\begin{Remark} \label{rem:mark}
For a Markov semigroup $(P_t)_{t \geq 0}$ on $(Q,\Sigma)$ and a probability measure
$\mu$ on $(Q,\Sigma)$, we obtain for $0 \leq t_1 < \cdots < t_n$ and
$\bt = (t_1, \ldots, t_n)$ a probability measure $P_\bt^\mu$ on $Q^n$ \cite[Satz~36.4]{Ba96}:
\begin{gather*} P_\bt^\mu(B) = \int_{Q^{n+1}} \chi_B(x_1, \ldots, x_n)
\mu(dx_0) P_{t_1}(x_0, dx_1) P_{t_2- t_1}(x_1, dx_2) \cdots
P_{t_n- t_{n-1}}(x_{n-1}, dx_n).\end{gather*}
The measure $P_\bt^\mu$ can also be written as $P_\bt^\mu = \mu P_\bt$ for the kernel
\begin{gather*} P_\bt(x_0,B) = \int_{Q^{n}} \chi_B(x_1, \ldots, x_n)
P_{t_1}(x_0, dx_1)P_{t_2- t_1}(x_1, dx_2)\cdots P_{t_n- t_{n-1}}(x_{n-1}, dx_n)\end{gather*}
on $Q \times \Sigma^n$.
This is a projective family of measures. If $(Q,\Sigma)$ is a~polish space, then this leads to a~stochastic process $(X_t)_{t\geq 0}$ with state
space~$(Q,\Sigma)$~\cite[Corollary~35.4]{Ba96}.
The probability measure~$\mu$ is the distribution of~$X_0$. It is called the
{\it initial distribution}. According to~\cite[Theorem~42.3]{Ba96}, the so obtained
process has the Markov property.
\end{Remark}

\subsection*{Extending stationary Markov processes to the real line}
\label{subsubsec:6.3.1}

Let $(P_t)_{t \geq 0}$ be a Markov semigroup on $(Q,\Sigma)$ and $\nu$ be a measure on $(Q,\Sigma)$. We obtain for
\begin{gather*} -s_m < -s_{m-1} < \cdots < -s_1 < 0 \leq t_1 < \cdots < t_n \end{gather*}
a measure $P_{\bs,\bt}^\nu$ on $Q^{m+n}$ by
\begin{gather*}
 P_{\bs,\bt}^\nu(B) = \int_{Q^{n+m+1}} \chi_B(y_m, \ldots, y_1, x_1, \ldots, x_n)\\
\hphantom{P_{\bs,\bt}^\nu(B) =}{}\times P_{s_m- s_{m-1}}(y_{m-1}, dy_m)\cdots P_{s_2- s_1}(y_1, dy_2)P_{s_1}(x_0, dy_1)\nu(dx_0)\\
\hphantom{P_{\bs,\bt}^\nu(B) =}{}\times P_{t_1}(x_0, dx_1)P_{t_2- t_1}(x_1, dx_2)\cdots P_{t_n- t_{n-1}}(x_{n-1}, dx_n).
\end{gather*}
This means that $ P_{\bs,\bt}^\nu = \nu P_{\bs,\bt}$
for the kernel
\begin{gather*}
 P_{\bs,\bt}(x_0,B)  = \int_{Q^{n+m}} \chi_B(y_m, \ldots, y_1, x_1, \ldots, x_n)\\
\hphantom{P_{\bs,\bt}(x_0,B)  =}{}\times P_{s_m- s_{m-1}}(y_{m-1}, dy_m)\cdots P_{s_2- s_1}(y_1, dy_2)P_{s_1}(x_0, dy_1)\\
\hphantom{P_{\bs,\bt}(x_0,B)  =}{} P_{t_1}(x_0, dx_1) P_{t_2- t_1}(x_1, dx_2)\cdots P_{t_n- t_{n-1}}(x_{n-1}, dx_n).
\end{gather*}
This is a projective family of measures. If $(Q,\Sigma)$ is a~polish space, this leads to a stochastic process $(X_t)_{t\in \R}$ with state space $(Q,\Sigma)$ \cite[Corollary~35.4]{Ba96}. If $\nu$ is a probability measure, then the measure~$\nu$ is called the {\it initial distribution} of the process. It coincides with the distribution of $X_0$. Suppose, in addition, that
\begin{gather*} 
 \int_{Q^2} g(\omega) \nu(d\omega) P_t(\omega, d\omega') f(\omega')
= \int_Q (P_t f)g \, \nu(dx_0) = \int_Q f(P_t g) \nu(dx_0)   \\
\hphantom{\int_{Q^2} g(\omega) \nu(d\omega) P_t(\omega, d\omega') f(\omega')}{} =  \int_{Q^2} g(\omega)\nu(d\omega') P_t(\omega', d\omega) f(\omega')
\qquad \mbox{for} \quad 0 \leq f,g.
\end{gather*}
Then
\begin{gather*}
 P_{\bs,\bt}^\nu(B)
 = \int_{Q^{n+m+1}} \chi_B(y_m, \ldots, y_1, x_1, \ldots, x_n)\\
\hphantom{P_{\bs,\bt}^\nu(B)=}{}\times P_{s_m- s_{m-1}}(y_{m-1}, dy_m)\cdots P_{s_2- s_1}(y_1, dy_2)P_{s_1}(x_0, dy_1)\nu(dx_0)\\
\hphantom{P_{\bs,\bt}^\nu(B)=}{}\times P_{t_1}(x_0, dx_1) P_{t_2- t_1}(x_1, dx_2) \cdots P_{t_n- t_{n-1}}(x_{n-1}, dx_n)\\
\hphantom{P_{\bs,\bt}^\nu(B)}{} = \int_{Q^{n+m+1}} \chi_B(y_m, \ldots, y_1, x_1, \ldots, x_n)\\
\hphantom{P_{\bs,\bt}^\nu(B)=}{}\times P_{s_m- s_{m-1}}(y_{m-1}, dy_m)\cdots P_{s_2- s_1}(y_1, dy_2)\nu(dy_1) P_{s_1}(y_1, dx_0)\\
\hphantom{P_{\bs,\bt}^\nu(B)=}{}\times P_{t_1}(x_0, dx_1)P_{t_2- t_1}(x_1, dx_2)\cdots P_{t_n- t_{n-1}}(x_{n-1}, dx_n) \\
 \hphantom{P_{\bs,\bt}^\nu(B)}{}= \int_{Q^{n+m}} \chi_B(y_m, \ldots, y_1, x_1, \ldots, x_n)\\
\hphantom{P_{\bs,\bt}^\nu(B)=}{}\times P_{s_m- s_{m-1}}(y_{m-1}, dy_m)\cdots P_{s_2- s_1}(y_1, dy_2)\nu(dy_1) P_{s_1+t_1}(y_1, dx_1)\\
\hphantom{P_{\bs,\bt}^\nu(B)=}{}\times P_{t_2- t_1}(x_1, dx_2) \cdots P_{t_n- t_{n-1}}(x_{n-1}, dx_n)\\
\hphantom{P_{\bs,\bt}^\nu(B)}{} = \int_{Q^{n+m}} \chi_B(y_m, \ldots, y_1, x_1, \ldots, x_n)\\
\hphantom{P_{\bs,\bt}^\nu(B)=}{}\times P_{s_m- s_{m-1}}(y_{m-1}, dy_m)\cdots \nu(dy_2) P_{s_2- s_1}(y_2, dy_1)P_{s_1+t_1}(y_1, dx_1) \\
\hphantom{P_{\bs,\bt}^\nu(B)=}{}\times P_{t_2- t_1}(x_1, dx_2)\cdots P_{t_n- t_{n-1}}(x_{n-1}, dx_n) = \cdots \\
\hphantom{P_{\bs,\bt}^\nu(B)}{}= \int_{Q^{n+m}} \chi_B(y_m, \ldots, y_1, x_1, \ldots, x_n)\\
\hphantom{P_{\bs,\bt}^\nu(B)=}{}\times \nu(dy_m) P_{s_m- s_{m-1}}(y_m, dy_{m-1})\cdots P_{s_2- s_1}(y_2, dy_1)P_{s_1+t_1}(y_1, dx_1)\\
\hphantom{P_{\bs,\bt}^\nu(B)=}{}\times P_{t_2- t_1}(x_1, dx_2)\cdots P_{t_n- t_{n-1}}(x_{n-1}, dx_n).
\end{gather*}
If $\nu$ is a probability measure, we thus obtain a stationary process with values in $Q$. For $\bt = (t_1, \ldots, t_n)$ and $t_1 < \ldots < t_n$ in $\R$, we then have for the distribution of this process
\begin{gather*} P_{\bt}^\nu(B)
= \int_{Q^n} \chi_B(x_1, \ldots, x_n)\,
\nu(dx_1) P_{t_2- t_1}(x_1, dx_2)\cdots P_{t_n- t_{n-1}}(x_{n-1}, dx_n).\end{gather*}
This formula immediately implies that the translation invariance of the measure $P^\nu$ on $Q^\R$.

\subsection*{Acknowledgements}

The research of P.~Jorgensen was partially supported by the Binational Science Founda\-tion Grant number 2010117. The research of K.-H.~Neeb was supported by DFG-grant NE \mbox{413/7-2}, Schwerpunktprogramm ``Darstellungstheorie''. The research of G.~\'Olafsson was supported by NSF grant DMS-1101337.
The authors wish to thank the Mathematisches Forschungsinstitut Oberwolfach for hosting a Workshop on ``Ref\/lection Positivity in Representation Theory, Stochastics and Physics'' November, 30~-- December~6, 2014. The present research was started at the workshop, and it has benef\/itted from our discussions with the participants there.

\addcontentsline{toc}{section}{References}
\pdfbookmark[1]{References}{ref}
\LastPageEnding

\end{document}